\documentclass{article}

\usepackage[reset,a4paper,vmargin=3truecm,hmargin=2truecm]{geometry}
\usepackage{amsmath,amsthm,amssymb,amscd,epic}
\usepackage{graphicx,bbm,amsthm,graphicx}
\usepackage{mathtools}
\usepackage{leftidx}
\usepackage{hyperref}
\usepackage{cool}
\usepackage{bm}
\usepackage{bbold}
\usepackage{subfig}
\usepackage{tikz}
\usepackage{mathrsfs}

\usepackage{tikz}
\usetikzlibrary{matrix,arrows,decorations.pathmorphing}
\usetikzlibrary{cd}

\theoremstyle{plain}
\newtheorem{thm}{Theorem}[section]
\newtheorem{lem}[thm]{Lemma}
\newtheorem{prop}[thm]{Proposition}
\newtheorem{claim}[thm]{Claim}
\newtheorem{cor}[thm]{Corollary}

\newtheorem{conject}[thm]{Conjecture}

\theoremstyle{definition}

\newtheorem{ex}{Example}[section]

\theoremstyle{remark}
\newtheorem{rem}{Remark}[section]

\theoremstyle{observation}

\DeclareMathOperator{\tr}{tr}

\DeclareMathOperator{\tridiag}{tridiag}
\newcommand{\Tridiag}[4]{\tridiag\mat{#1 & #2 \\ #3}_{#4}}

\newcommand{\HRabi}[1]{H_{#1}} 
\newcommand{\cp}[2]{P^{(#1)}_{#2}} 

\newcommand{\N}{\mathbb{N}} 
\newcommand{\Z}{\mathbb{Z}} 
\newcommand{\Q}{\mathbb{Q}} 
\newcommand{\R}{\mathbb{R}} 
\newcommand{\C}{\mathbb{C}} 

\newcommand{\ibQRM}[1]{ibQRM$_{#1}$}
\newcommand{\Jker}[1]{\mathcal{K}_{#1}}

\DeclareMathOperator{\Spec}{Spec}

\newcommand{\mat}[1]{\begin{bmatrix}#1\end{bmatrix}}

\newcommand{\e}{\varepsilon}
\newcommand{\sign}{\rm{sign}}

\title{Degeneracy and hidden symmetry \\
-- an asymmetric quantum Rabi model with an integer bias}
\author{Cid Reyes-Bustos and Masato Wakayama}

\begin{document}

\maketitle

\begin{abstract}  

  The hidden symmetry of the asymmetric quantum Rabi model (AQRM) with a half-integral bias (\ibQRM{\ell}) was
  uncovered in recent studies by the explicit construction of operators $J_\ell$ commuting with the Hamiltonian.
  The existence of such symmetry has been widely believed to cause the degeneration of the spectrum, that is, the crossings
  on the energy curves. In this paper we propose a conjectural relation between the
  symmetry and degeneracy for the \ibQRM{\ell} given explicitly in terms of two polynomials appearing
  independently in the respective investigations. Concretely, one of the polynomials appears as the quotient of the constraint
  polynomials that assure
  the existence of degenerate solutions while the other determines a quadratic relation (in general, it defines a curve of
  hyperelliptic type) between the \ibQRM{\ell} Hamiltonian and its basic commuting operator $J_\ell$. 
  Following this conjecture, we derive several interesting structural insights of the whole spectrum.  
  For instance, the energy curves are naturally shown to lie on a surface determined by the family of hyperelliptic
  curves by considering the coupling constant as a variable. This geometric picture contains the generalization of the parity
  decomposition of the symmetric quantum Rabi model. Moreover, it allows us to describe a remarkable approximation of
    the first $\ell$ energy curves by the zero-section of the corresponding hyperelliptic curve.
  These investigations naturally lead to a geometric picture of the (hyper-)elliptic surfaces given by the
  Kodaira-N\'eron type model for a family of curves over the projective line in connection with the energy curves, which may be
  expected to provide a complex analytic proof of the conjecture. 



\end{abstract}

\setcounter{tocdepth}{2}
\tableofcontents

\section{Introduction}
\label{sec:introduction}
  
Symmetry is a fundamental concept in mathematics and physics. In quantum physics, the presence of non-trivial operators commuting with the Hamiltonian of a system indicates the existence of quantities that are conserved under the time evolution of the system and is usually important for the solution of the Schr\"odinger equation associated to the Hamiltonian. Symmetries in the Hamiltonian system are also intimately related to the practical concept of integrability in quantum systems (see e.g. \cite{B2011, C2011, R1972}). 
The quantum Rabi model (QRM) \cite{R1936, JC1963, Ni2010} is one of the simplest models in quantum physics. It describes the interaction between a quantum harmonic oscillator and a two-level system. Despite its simplicity, the QRM exhibits enormous applications in areas such as quantum optics, solid state physics and quantum information theory (see e.g.  \cite{bcbs2016, Y2017, YS2018}). The focus of the present paper is the asymmetric quantum Rabi model (AQRM) obtained by adding the bias term (measured by a real number) to the QRM Hamiltonian.
In contrast with the QRM, which possesses an obvious $\Z_2=\{\pm 1\}$ symmetry, the AQRM Hamiltonian does not appear to have such a symmetry due to the presence of the bias term. Even without an obvious symmetry, it has been shown \cite{KRW2017} that the energy levels of the AQRM exhibit crossings when the bias term is given by an integer.

In the present paper we thus focus on the asymmetric quantum Rabi model with an integral bias parameter $\ell\; (\ell\in \Z)$, referred henceforth as {\it the integral biased quantum Rabi model} (\ibQRM{\ell}). Mainly, we study the so-called hidden symmetry of the \ibQRM{\ell} Hamiltonian (see e.g. \cite{A2020}), recently computed in \cite{MBB2020} for small values of the bias parameter, by relating it with the degeneracy of its spectrum. As one of the main contributions, we give a concrete realization of the commonly held expectation that the level crossings are in fact related to the existence of symmetry.

Let us start by briefly introducing the AQRM Hamiltonian and the role of the bias term. Recall that the Hamiltonian $\HRabi{\e}$ of the
AQRM ($\hbar=1$) is given by
\begin{equation}\label{eq:aH}
  \HRabi{\e} = \omega a^\dag a+\Delta \sigma_z +g\sigma_x(a^\dag+a) + \frac{\e}{2} \sigma_x,
\end{equation}
where $a^\dag$ and $a$ are the creation and annihilation operators of the bosonic mode,
i.e. $[a,\,a^\dag]=1$ and
\[
\sigma_x = \begin{bmatrix}
 0 & 1  \\
 1 & 0
\end{bmatrix}, \qquad \qquad
\sigma_z= \begin{bmatrix}
 1 & 0  \\
 0 & -1
\end{bmatrix}
\]
are the Pauli matrices, $2\Delta$ is the energy difference between the two levels,
$g$ denotes the coupling strength between the two-level system and the bosonic mode with frequency $\omega$
(subsequently, we set $\omega=1$ without loss of generality), and \(\e\) is a real number. 
In general, the AQRM actually provides a more realistic description of the circuit QED experiments employing flux qubits than the QRM itself \cite{Ni2010,Y2017,YS2018}. The Hamiltonian of the \ibQRM{\ell} corresponds to the case $\varepsilon = \ell \in \Z$ and, without loss of generality
we may assume $\ell \ge 0$.

The bias term $\frac\varepsilon2 \sigma_x$ was originally introduced to break the symmetry of the model and in particular, to eliminate the double degeneracy in the spectral curves. However, as mentioned above, when the bias-parameter $\varepsilon$ takes an integer value,
crossings in the spectral curves appear again as illustrated in Figure \ref{fig:EigencurvesAQRM} (see also \cite{HH2012, KRW2017} for more examples). This was observed first experimentally in \cite{LB2015JPA}, proved for $\ell = 1$ in \cite{W2016JPA} and clarified in full generality in \cite{KRW2017} by the study of the constraint polynomials arising from the Juddian, or equivalently quasi-exact solutions. Nevertheless, the proof of the existence of degeneracies in the AQRM for integer bias in \cite{KRW2017} did illuminate neither the nature nor the relation between degeneracies and an underlying possible symmetry on the Hamiltonian. We also remark that the spectral crossings have been shown to induce conical intersections in the energy landscape of the AQRM \cite{BLZ2015,LB2021b}, that is, in the surface determined by the spectrum of $\HRabi{\e}$ in the space $(g,\e,E)\in \R^3$.

\begin{figure}[h!]
  \centering
  \subfloat[$\varepsilon = 0.6$]{
    \includegraphics[height=6cm]{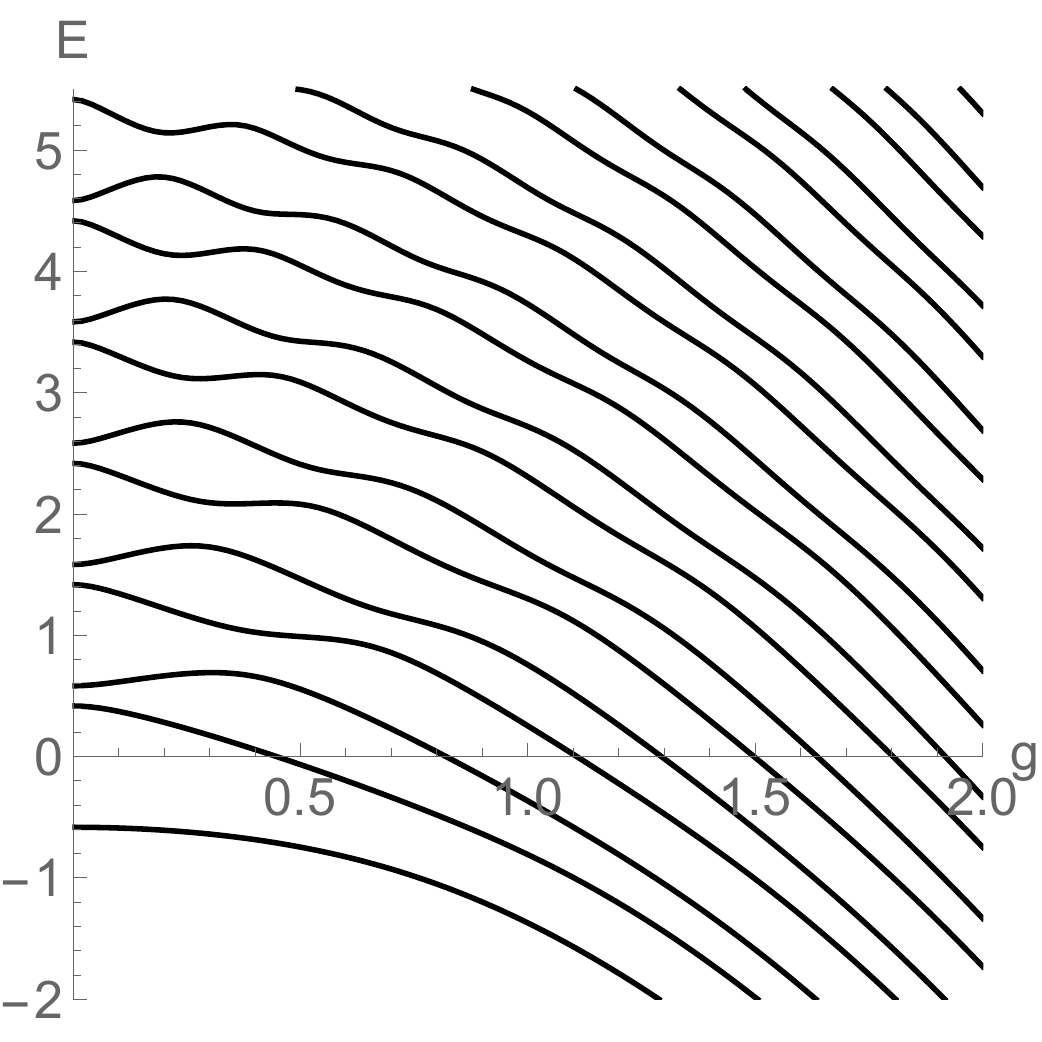}}
  ~ \qquad \qquad
  \subfloat[$\varepsilon = 1$]{
    \includegraphics[height=6cm]{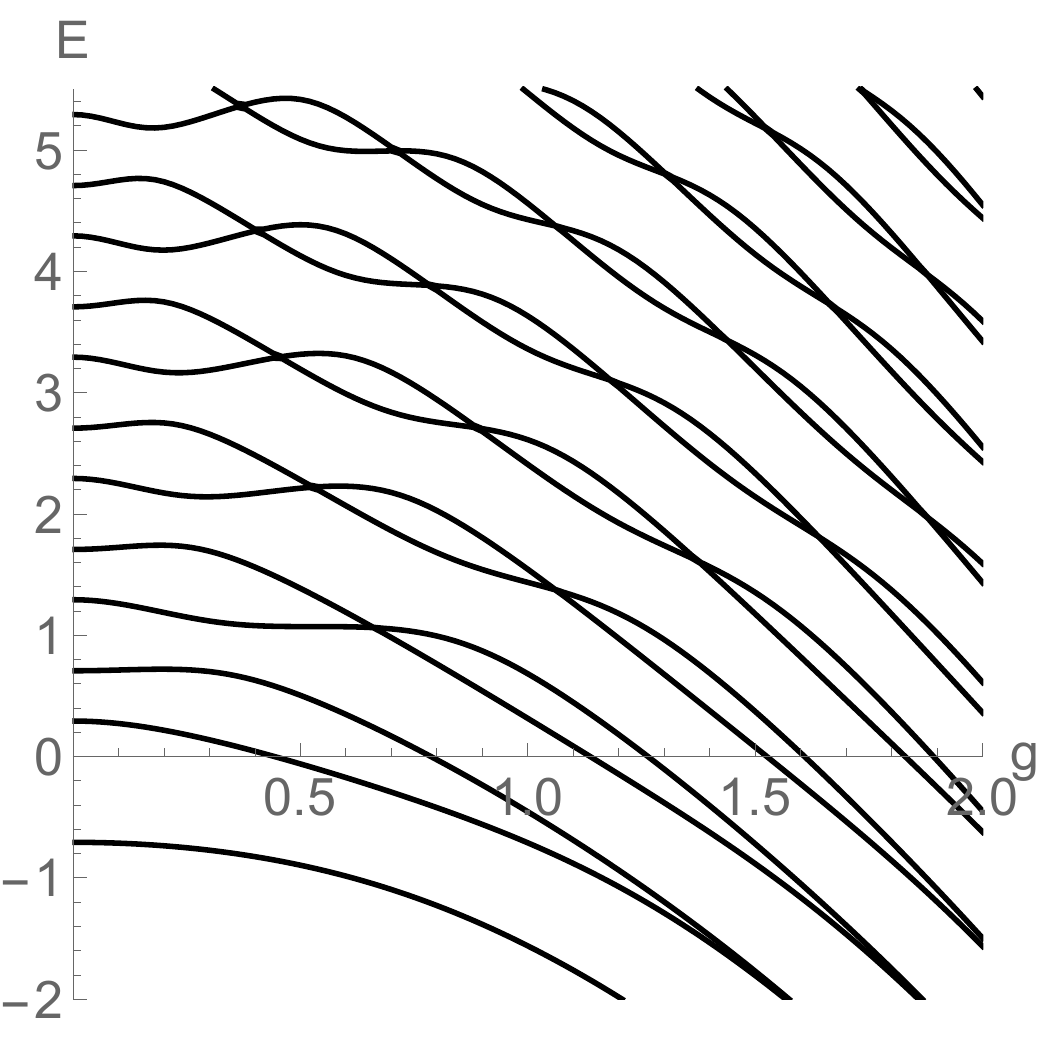}}
  \caption{Spectral curves of AQRM for $\Delta = \frac12$. Notice the crossings in the case $\varepsilon = 1$.}
  \label{fig:EigencurvesAQRM}
\end{figure}

In the symmetric QRM ($\ell =0$) case, the existence of a commuting involution operator leads directly to a decomposition
of the ambient Hilbert space $\mathcal{H}$ into two parity subspaces $\mathcal{H}_{\pm}$. Whenever there is a degeneracy on the spectrum
it may be verified that it occurs between one eigenvalue of each parity. This was one of the motivations for the search of a similar decomposition for the AQRM. For instance, it was shown in \cite{A2020} that any invariant decomposition necessarily depends on the system parameters and it follows that there is no uniform decomposition for the case $\ell \neq 0$. Recently, the hidden symmetry of the AQRM and generalizations has been widely discussed from different points of view \cite{B2019, GD2013,LB2020}.

A significant milestone in the search for the hidden symmetry of the \ibQRM{\ell} was the paper \cite{MBB2020} where the authors gave an algorithm to obtain an operator $J_{\ell}$ commuting with the Hamiltonian.
The algorithm, unfortunately, does not result in a general expression for the operator $J_\ell$, however, it provides the basic guidance
to further unveil the hidden symmetry. For the reference of the reader, in Appendix \ref{sec:expl-expr-j_ell} we give the explicit
expression of $J_{\ell}$ for small values of $\ell$.

Shortly after, in \cite{RBW2021} the basic properties of the operator $J_{\ell}$ were proved, including the fundamental property that
$J_{\ell}$ satisfies
\begin{equation}
  \label{eq:Jquad}
  J_{\ell}^2 = p_{\ell}(\HRabi{\ell}; g,\Delta), 
\end{equation}
for some polynomial $p_{\ell}(x; g ,\Delta)$ of degree $\ell$, first noted in \cite{MBB2020} for small values of $\ell$.
In \cite{MBB2020}, and also in \cite{RBW2021} based in geometric considerations, it was argued that the quadratic relation \eqref{eq:Jquad} represents the fact that the operator $J_{\ell}$ gives rise to a $\Z_2$ symmetry (parity).
We note that, however, there does not appear to be an straightforward way to describe the invariant decomposition (expressed as a ``parity'') of the Hilbert space induced by $J_\ell$, if one exists.
We also mention that the method of \cite{MBB2020} for the computation of the hidden symmetry has been extended to other systems with bias terms (see \cite{LLMB2021} for the biased Dicke model and \cite{LLMB2021a} for other models related to the QRM). 

In this paper, we initiate an algebro-geometric approach to the study of the hidden symmetry based on the relation  \eqref{eq:Jquad}.
Notably, we propose a conjecture giving a direct relation between \eqref{eq:Jquad} and the degeneracy of the spectrum of \ibQRM{\ell} by the polynomials appearing independently in each study. Concretely, the polynomials $p_{\ell}(x; g, \Delta)$ in \eqref{eq:Jquad} are, up to a change of variable, equal to the polynomials $A_{\ell}^{N}((2g)^2,\Delta^2)$ appearing as quotients of the two constraint polynomials controlling the degeneracy \cite{KRW2017}. In \S\ref{sec:preliminary-results} we recall the background materials on symmetry and degeneracy for the \ibQRM{\ell} and in \S\ref{sec:constr-polyn} we give the precise statement of the main conjecture (Conjecture \ref{conject:main}) and some of its immediate consequences.

Further consequences of the main conjecture are discussed in \S\ref{sec:symm-decomp}. In particular, in Theorem \ref {thm:detexpP}, under the assumption of the main conjecture, we give a determinant expression for $p_{\ell}(x; g,\Delta)$ by a matrix that controls the degeneracy \cite{KRW2017}. In fact, the determinant expression for $p_{\ell}(x; g,\Delta)$ turns out to be equivalent to the main conjecture. Moreover, the conjecture also allows a deeper understanding of the kernel $\mathcal{K}_\ell$ of the operator $J_\ell$, which is related to the question of the existence of invariant (parity) decompositions. 

In \S\ref{sec:approximation} we initiate the study of the spectrum of the \ibQRM{\ell} by the associated geometric structure provided by the quadratic relation \eqref{eq:Jquad}. Concretely, associated to the relation \eqref{eq:Jquad} between $J_\ell$ and $\HRabi{\ell}$, the equation 
\begin{equation}
  \label{eq:introelliptic}
  y^2=p_\ell(x;g,\Delta)
\end{equation}
defines in general a hyperelliptic curve (see e.g. \cite{M1984, S1994}). The curve \eqref{eq:introelliptic}, by definition, contains the joint eigenvalues $(\lambda,\mu_\lambda)$ of the operators $H_{\ell}$ and $J_{\ell}$. Similarly, we may consider the algebraic surface
\[
  \mathcal{S}_\ell = \left\{ (x,y,g) \in \R^3\,|\, y^2=p_\ell(x;g,\Delta)  \right\},
\]
giving rise to a generalization of the usual spectral curves.  In \S\ref{sec:parity} we describe how the geometric picture gives a natural setting for the parity decomposition for the case of the QRM (see Figure \ref{fig:3dQRM}(a)) and a generalization for the case of $\ell>0$.
Actually, the surface $\mathcal{S}_\ell$ draws a resolution of singularities for the degenerate crossing point of the spectral curves.

Moreover, another numerical result worth mentioning is that there is a remarkable approximation of the first $\ell$ spectral curves by the $(x,g)$-plane curves defined by $p_\ell(x;g,\Delta)=0$, the section of the algebraic surface $\mathcal{S}_\ell$ by $y=0$. 
This was first observed in \cite{RBW2021} and in \S\ref{sec:approximation} we discuss the approximation in light of the geometric picture of the spectrum. The nature of the approximation remains largely mysterious. For instance, it is not known whether there are crossings between the three curves in the $(g,x)$-plane: spectral (energy) curves $E=\lambda(g, \Delta)$, energy baselines $E=N-g^2+\frac\ell2$ and curves defined by $p_\ell(E; g,\Delta)=0$ (see Figure \ref{fig:Eigencurves3}).

On the topic of approximation of the eigenvalues of the AQRM by polynomials, it is relevant to mention the adiabatic approximation (AA) that approximates the eigenvalues using Laguerre polynomials, obtained by considering the exact solutions in the extremal case $\Delta \to 0$. In the recent paper \cite{LB2021b}, the authors present a generalized adiabatic approximation (GAA) for the AQRM where the Laguerre polynomials are replaced by the constraint polynomials. The GAA gives better agreement with the exact values than the AA for a large family of parameters $g,\Delta>0$ and, by construction always agrees on the degenerate points (see also Remark \ref{rem:gaa}). The GAA in \cite{RBW2021} is, however, only valid for eigenvalue curves that contains crossings.
Since the polynomial $p_\ell(x;g,\Delta)$ is given by the ratio of two constraint polynomials if the aforementioned conjecture holds, 
by combining the approximation of the first $\ell$ eigenstates by the curves $p_\ell(x;g,\Delta)=0$ and the GAA in the $(x,g)$-plane, we might say that the constraint polynomials know the essential features of the shape of the spectral curves of \ibQRM{\ell}. It is indeed quite remarkable that the existence of degenerate points determines  the spectral structure of the system to a great extent. 

In \S~\ref{sec:hyperelliptic}, to explore another aspect of the geometric picture, we consider the surface $\mathcal{S}_\ell$ as
a (hyper)-elliptic surface. The additional structure obtained in this setting is expected to contribute to the proof of the Conjecture \ref{conject:main} (via the equivalent conditional Theorem \ref{thm:detexpP}). For the case $\ell=3$ we briefly discuss the algebraic surface by means of the Kodaira-N\'eron model of elliptic surfaces.
Since the generic fiber of the surface turns out to have a particularly simple expression, we easily describe the singular fibers of the model and its type in the Kodaira classification. In this setting, the degeneracy of eigenvalues may be reworded in terms of the group
operator in the generic fiber as we briefly discuss in \S \ref{sec:dens-judd-solut} related to the distribution of Juddian points.
Finally, in \S\ref{sec:addit-remarks-degen} we describe how the divisibility of constraint polynomials at the Juddian points (degeneracy of the spectrum for the AQRM) resembles the study \cite{CZ2010, RTW2021} of the divisibility of polynomials and degeneracy of integral points for curves on surfaces along a certain blow-up. In Diophantine geometry, as it is pointed out in \cite{RTW2021}, divisibility conditions are connected to the celebrated Vojta's conjectures in many ways. Therefore, the discussion in \S\ref{sec:addit-remarks-degen} suggests a potential connection of the spectrum of the \ibQRM{\ell} with problems in Diophantine approximation and arithmetic geometry (see e.g. \cite{HS2000}).

\section{Revisiting the hidden symmetry of the \ibQRM{\ell}}
\label{sec:preliminary-results}

We denote by $\C[a,a^{\dag}]$ the Weyl algebra generated by the elements $a$ and $a^{\dag}$ and by ${\rm Mat}_2(\C[a,a^{\dag}])$ the $2\times2$ matrix algebra over $\C[a,a^{\dag}]$ . The degree of a monomial $a^k (a^\dag)^\ell \in \C[a,a^{\dag}]$ is defined to be $k+\ell$ and for a general element $\mathfrak{f} \in \C[a,a^{\dag}]$ the degree is defined to be the maximum degree of the monomials appearing in $\mathfrak{f}$.

We denote the \ibQRM{\ell} Hamiltonian by
\begin{equation*}
      \HRabi{\ell} =  a^\dag a+\Delta \sigma_z +g\sigma_x(a^\dag+a) + \frac{\ell}{2} \sigma_x, \quad (\ell\in \Z).
\end{equation*}
The system parameter $\Delta >0 $ is taken as fixed and we consider spectral curves (energy curves) with respect to the variation
of the parameter $g>0$.

The main result of \cite{RBW2021} (and \cite{MBB2020})  is the existence\footnote{Note that in \cite{RBW2021} the operator $J_{\ell}$ (resp. the Hamiltonian $H_{\ell}$) was denoted by $J_{\frac{\ell}{2}}$ (resp. $H_{\tfrac{\ell}{2}}$).} of a self-adjoint operator $J_{\ell}$ such that
\[
  [\HRabi{\ell},J_{\ell}] = 0,
\]
that is, it is a symmetry of the Hamiltonian $\HRabi{\ell}$. The operator $J_{\ell}$ has the form
\[
  J_{\ell} = \mathcal{P} Q_0^{(\ell)},
\]
where $\mathcal{P} = \exp(i \pi a^\dag a)$ is the photon number parity operator and $Q_0^{(\ell)} \in {\rm Mat}_2(\C[a,a^{\dag}])$.
The components of $Q_0^{(\ell)}$ have degree $\ell$ as polynomials on $a$ and $a^\dag$. Moreover, $Q_0^{(\ell)}$ is the matrix with components
of least degree satisfying this condition and any such matrix of larger degree is a multiple of $Q_0^{(\ell)}$. We also note that a
general expression for the components of $Q_0^{(\ell)}$ for general $\ell$ is not known currently.

The property of the operator $J_{\ell}$ that is more relevant to the current study is the quadratic equation
\begin{equation} \label{eq:Jquad2}
  J_{\ell}^2 = p_{\ell}(\HRabi{\ell};g,\Delta),
\end{equation}
where $p_{\ell} \in \R[x,g,\Delta]$ is a polynomial of degree $\ell$. As mentioned in the introduction, there is a relation between this
quadratic equation and the $\Z_2$-symmetry, we refer the reader to \S \ref{sec:approximation} and the aforementioned papers
for the details.

An immediate consequence of \eqref{eq:Jquad2} is that the eigenvalues of $J_{\ell}$ are of the form
\begin{align}
  \label{eq:mu}
  \mu = \pm  \sqrt{p_{\ell}(\lambda; g ,\Delta)},
\end{align}
for some eigenvalue $\lambda \in \Spec(\HRabi{\ell})$. In general, it is not easy to determine the sign of the eigenvalue (see \S~\ref{sec:sign-eigenv-j_yell} for the case $g=0$). Moreover, since $J_{\ell}$ is a self-adjoint operator,
\begin{equation}
  \label{eq:nonneg}
  p_{\ell}(\lambda;g,\Delta) \geq 0,
\end{equation}
holds for all eigenvalues $\lambda \in \Spec(\HRabi{\ell})$.

At this point it is not possible to discount the possibility that equality in \eqref{eq:nonneg} may hold, equivalently, that the
operator $J_{\ell}$ has a nontrivial kernel that we denote by $\Jker{\ell}= \Jker{\ell}(g, \Delta) \subset \HRabi{\ell}$. As we see later, the existence
of a non-trivial kernel $J_{\ell}$ is related to the potential existence of a parity decomposition for the Hamiltonian
$H_{\ell}$.
Note that since $\Delta$ is fixed, the kernel $\Jker{\ell}$ depends on the parameter $g >0$, and the same holds for the Hamiltonians $H_\ell$ and $J_{\ell}$. In general, we omit the dependence on the notation. 

As mentioned in the introduction, the eigenvalues curves of $\HRabi{\ell}$ may contain degenerate eigenvalues
(see also \S \ref{sec:dens-judd-solut}). For $g>0$, let $\lambda$ be a degenerate eigenvalue and denote by $V_{\lambda}$ the corresponding eigenspace
under $\HRabi{\ell}$. Since $\lambda$ is doubly degenerated, the dimension of $V_{\lambda}$ is exactly $2$. 

\subsection{Degenerate eigenspaces in the Bargmann picture}
\label{sec:degen-eigensp}

To investigate the action of $J_{\ell}$ on the degenerate eigenspace $V_\lambda$, it is convenient to consider a realization of the operators
$a$ and $a^\dag$ in a particular Hilbert space. For the analysis of the solutions of the QRM and its generalizations, including the
\ibQRM{\ell},  the Bargmann space realization is frequently used (see e.g. \cite{Sc1967AP} and \cite{B2011,KRW2017} for the
case of the QRM). The Bargmann space $\mathcal{B}$ is the space of entire functions finite with respect to the norm induced by the
inner product
\[
  (f|g)= \frac1\pi \int_{\C} \overline{f(z)}g(z)e^{-|z|^2}d(\Re(z))d(\Im(z)).
\]
In this realization, the operators $a$ and $a^\dag$ are mapped to
\[
  a \to \partial_z := \tfrac{d}{d z}, \qquad a^\dag \to z,
\]
and thus, we consider the Hamiltonian $\mathcal{H_\ell}$ and the operator $J_\ell$ as operators acting on
$\mathcal{H} = \mathcal{B}\otimes \C^2$.

It is widely known that any degenerate eigenvalue $\lambda$ of the \ibQRM{\ell} is of the form $\lambda = N \pm \frac{\ell}{2} - g^2$ for some integer
$N\geq 1$ (in general, eigenvalues of that form are called ``exceptional''). Moreover, the degenerate eigenspace $V_\lambda$ consists of
Juddian (quasi exact) solutions, that is, solutions that consist of a polynomial multiplied by an exponential factor. We refer
the reader to \cite{KRW2017} (see also \cite{LB2015JPA,W2016JPA}) for a detailed discussion on exceptional solutions and degeneracy for the
AQRM.

We now give the explicit form of a basis of the degenerate eigenspace $V_\lambda$ with $\lambda = (N +\ell)- \frac{\ell}{2} - g^2 = N + \frac{\ell}{2} - g^2$.
First, we have an eigenfunction $\Psi^{(N,\ell)}= (\psi^{(N,\ell)}_1,\psi^{(N,\ell)}_2)^t$ with components given by
\begin{align*}
  \psi_1^{(N,\ell)}(z) & = e^{-g z}\left( \frac{4g^2 K_{N-1}^{(N,\frac{\ell}{2})}}{\Delta} \frac{(g+z)^N}{(2g)^N} - \Delta \sum_{n=1}^{N-1} \frac{K_n^{(N,\frac{\ell}{2})}}{n-N} \frac{(g+z)^n}{(2g)^n} \right) \\
  \psi_2^{(N,\ell)}(z) & = e^{-g z} \sum_{n=1}^{N-1} K_n^{(N,\frac{\ell}{2})} \frac{(g+z)^n}{(2g)^n},
\end{align*}
and another eigenfuction $\Phi^{(N+\ell,-\ell)} = (\phi^{(N+\ell,-\ell)}_1,\phi^{(N+\ell,-\ell)}_2)^t$ with components
\begin{align*}
  \phi_1^{(N+\ell,-\ell)}(z) & = e^{g z} \sum_{n=1}^{N+\ell-1} K_n^{(N+\ell,-\frac{\ell}{2})} \frac{(g+z)^n}{(2g)^n} \\
  \phi_2^{(N+\ell,-\ell)}(z) & = e^{g z} \left( \frac{4g^2 K_{N+\ell-1}^{(N+\ell,-\frac{\ell}{2})}}{\Delta} \frac{(g+z)^{N+\ell}}{(2g)^{N+\ell}} - \Delta \sum_{n=1}^{N+\ell-1} \frac{K_n^{(N+\ell,-\frac{\ell}{2})}}{n-N+\ell} \frac{(g+z)^n}{(2g)^n} \right).
\end{align*}

In both cases, the coefficients are defined as
\begin{align*}
  K_0^{(N,\tfrac{\varepsilon}{2})} &= 1,\\
  K_1^{(N,\tfrac{\varepsilon}{2})} &=  2g + \frac{1}{2g}\left( -N -\epsilon + \frac{\Delta^2}{N} \right) \\
  K_n^{(N,\tfrac{\varepsilon}{2})} &= \frac{1}{n}\left( 2g + \frac{1}{2g}\left(n-1 -N -\varepsilon + \frac{\Delta^2}{N-n+1}\right) \right) K_{n-1}^{(N,\tfrac{\varepsilon}{2})} - \frac{1}{n} K_{n-2}^{(N,\tfrac{\varepsilon}{2})},
\end{align*}
for $n\ge2$ and $\varepsilon \in \R$.

In the discussion above, we have assumed the existence of the degenerate eigenvalue $\lambda$.
In general, the conditions on the parameters $g,\Delta\geq0$ for the existence of the eigenvalue $\lambda = N + \tfrac{\ell}{2}-g^2$, and therefore,
the existence of the Juddian solutions $\Psi^{(N)}$ and $\Phi^{(N+\ell)}$, are given by the constraint conditions
\begin{align}
  \label{eq:constraint}  K_{N}^{(N,\tfrac{\ell}{2})} = 0 \quad \text{and} \quad
  K_{N+\ell}^{(N+\ell,-\frac{\ell}{2})} = 0,
\end{align}
respectively. These constraint conditions are usually expressed in an equivalent polynomial form
  \[
    \cp{N,+ \frac{\ell}{2}}{N}((2g)^2,\Delta^2)=0 \quad \text{and} \quad \cp{N+\ell,- \frac{\ell}{2}}{N+\ell}((2g)^2,\Delta^2)=0.
  \]
Here, the constraint
polynomial $\cp{N,\tfrac{\varepsilon}{2}}{N}(u,v)$ is the $N$-th polynomial of the family $\cp{N,\tfrac{\varepsilon}{2}}{k}(u,v)$
defined by three-term recurrence relation
\begin{align*}
  \cp{N,\tfrac{\varepsilon}{2}}{0}(u,v) &= 1, \\
  \cp{N,\tfrac{\varepsilon}{2}}{1}(u,v) &= u + v - 1 - \varepsilon, \\
  \cp{N,\tfrac{\varepsilon}{2}}{k}(u,v) &= (k u + v - k(k + \varepsilon) ) P_{k-1}^{(N,\tfrac{\varepsilon}{2})}(u,v) - k(k-1)(N-k+1) u P_{k-2}^{(N,\tfrac{\varepsilon}{2})}(u,v),
\end{align*}
for \(k \geq 2 \).

The fact that for the \ibQRM{\ell} any Juddian solution $\lambda=N+\tfrac{\ell}{2}-g^2$ is degenerate implies that the two constraint conditions in \eqref{eq:constraint} have the same positive roots, as was shown in \cite{KRW2017}. This result is essential for further developments and we continue this discussion in \S \ref{sec:constr-polyn}.

\begin{ex} \label{ex:cPoly}
  For \(N=2,4\), we have the following constraint polynomials
  \begin{align*}
    \cp{2,\frac{\e}{2}}{2}(u,v) &= 2 u^2 + 3 u v + v^2 - 4(2+ \e) u - (3 \e +5)v + 2 (\e+1) (\e+2) \\
    \cp{4,-\frac{\e}{2}}{4}(u,v) &= 24 u^4 + 50 u^3 v + 35 u^2 v^2 + 10 u v^3 + v^4 \\
                               &+ 96 (\e-4) u^3  +2 (75 \e-271) u^2 v +10 (7 \e-23) u v^2+10 (\e-3) v^3 \\
                               &+144 (\e-4) (\e-3) u^2 + 2 \left(75 \e^2-473 \e+722\right) u v+\left(35 \e^2-200 \e+273\right) v^2 \\
                               &+96 (\e-4) (\e-3) (\e-2) u + \left(50 \e^3-404 \e^2+1030 \e-820\right) v +24 (\e-4) (\e-3) (\e-2) (\e-1).
  \end{align*}
\end{ex}

\subsection{Action of $J_\ell$ on degenerate eigenspaces}
\label{sec:action-j_ell}

Let us now consider the action of $J_{\ell}$ on eigenfunctions of a degenerate eigenspace $V_{\lambda}$. Since $\HRabi{\ell}$ and $J_{\ell}$ commute,
the subspace $V_{\lambda}$ is a $J_{\ell}$-invariant subspace. There are two possibilities, one is that $V_{\lambda}$ is also an eigenspace of $J_{\ell}$
for a unique eigenvalue $\mu$, or that $V_{\lambda}$ decomposes as
\[
  V_{\lambda} = \C \phi_{1} \oplus \C \phi_{2},
\]
where $\phi_{i}$ are eigenfunctions of $J_{\ell}$ corresponding to different eigenvalues $\mu_i$ ($i=1,2$). Necessarily, by \eqref{eq:mu} we
must have $\mu_{1} = -\mu_{2}$.

\begin{thm} \label{thm:decomposition}
  Let $V_{\lambda}$ be the eigenspace for a degenerate eigenvalue $\lambda$ of $\HRabi{\ell}$. Then, the space $V_{\lambda}$ decomposes as
  \[
    V_{\lambda} = \C \phi_{\mu} \oplus \C \phi_{-\mu},
  \]
  where $\phi_{\mu}$ (resp. $\phi_{-\mu}$) is an eigenfuction of $J_{\ell}$ with real eigenvalue $\mu$ (resp. $-\mu$). Moreover, we have
  $\mu^2 = p_{\ell}(\lambda; g, \Delta)$. 
\end{thm}

The fact that the two eigenvalues $\pm \mu$ appear in the decomposition of $V_{\lambda}$ into $J_{\ell}$ eigenspaces resembles that at
the degenerate points there is a solution from each ``parity'' in the QRM. We discuss a generalization of the parity decomposition for the \ibQRM{\ell} based in geometric considerations in \S \ref{sec:parity}.

\begin{proof}

  Any degenerate eigenvalue of the \ibQRM{\ell} is of the form $\lambda = N + \frac{\ell}{2} - g^2$. Then, it is clear that $\Phi^{(N+\ell,-\ell)}$ and
  $\Psi^{(N,\ell)}$ are two linearly independent solutions for $\lambda$ in the Bargmann picture.
  
  Note that $\mathcal{P} = e^{i \pi z \partial_z}$, thus the identity
  \[
    \mathcal{P} e^{g z} f(z) = e^{-g z} f(-z),
  \]
  is valid for any polynomial $f \in \C[z]$. Now, since $\mathcal{P}$ appears in the expression of $J_{\ell}$,
  neither $\Phi^{(N+\ell,-\ell)}$ or $\Psi^{(N,\ell)}$ can be eigenvectors of $J_{\ell}$, and the result follows.
\end{proof}

\begin{rem}
  An alternative way to prove Theorem \ref{thm:decomposition} for the case $\ell>0$ is to show that the operator $J_{\ell}$ is
  non-degenerate for all $g,\Delta>0$. However, since there is no known general expression for $J_{\ell}$, this approach appears
    to be considerably difficult.
\end{rem}

The presence of the factors $e^{\pm g z}$ in the Juddian solutions $\Phi^{(N+\ell,-\ell)}$ and $\Psi^{(N,\ell)}$
forces the action of $J_{\ell}$ to alternate the two eigenfunctions. Concretely, we immediately obtain the following explicit form
of the action.

\begin{cor} \label{cor:Jaction}
  Let $\Phi^{(N+\ell,-\ell)}$ and $\Psi^{(N,\ell)}$ be the linearly independent (Juddian) eigenfunctions of $V_{\ell}$ for
  $\lambda = N+\tfrac{\ell}{2}-g^2$. Then, we have
  \begin{align*}
    J_\ell \Phi^{(N+\ell,-\ell)} &= \alpha \Psi^{(N,\ell)},\\
    J_\ell  \Psi^{(N,\ell)} &= \beta \Phi^{(N+\ell,-\ell)}
  \end{align*}
  for $\alpha,\beta \in \C(g,\Delta)$ such that $\alpha \beta = p_\ell(\lambda; g,\Delta)$.
\end{cor}

\begin{ex} \label{ex:action}

  Let us consider the action of $J_{\ell}$ on a degenerate eigenspace (consisting of Juddian solutions) for the case of $\ell=1$.
  Namely we assume that $P_1^{(1, \frac12)}((2g)^2, \Delta^2)=0$ and $P_2^{(2, -\frac12)}((2g)^2, \Delta^2)=0$.
  In this case, the operator $J_1$ is realized in the Bargmann space as
  \[
    J_{1} = \mathcal{P}
    \begin{bmatrix}
      \Delta & 2 g (g-\partial_z) \\
      2g (g+z) & \Delta
    \end{bmatrix}.
  \]
  The Juddian solutions $\Psi^{(1,1)}(z)$ assured by $P_1^{(1, \frac12)}((2g)^2, \Delta^2)=0$ and
  $\Phi^{(2,-1)}(z)$ by $P_2^{(2, -\frac12)}((2g)^2, \Delta^2)=0$, respectively, are given by
  \begin{align*}
    \psi_1^{(1,1)}(z) & = e^{-g z}\left( z + \frac{2 -2g^2}{2 g } \right) \\
    \psi_2^{(1,1)}(z) & = e^{-g z} \frac{\Delta}{2 g}
  \end{align*}
  and
  \begin{align*}
    \phi_1^{(2,-1)}(z) & = e^{g z} \left( -\frac{\Delta}{2g} z + \frac{\Delta^4 + 6 \Delta^2 g^2 - \Delta^2 g^2 + \Delta^2}{4 \Delta g^2}  \right) \\
    \phi_2^{(2,-1)}(z) & = e^{g z} \left( z^2 - \frac{\Delta^2+4g^2}{2 g} z  + \frac{\Delta^4+8\Delta^2g^2 + 8g^4 }{8 g^2}  \right)
  \end{align*}
  where we have followed the normalization given in \cite{LB2015JPA}. Then, by direct computation we obtain
  \begin{align*}
    J_1 \Psi^{(1,1)}(z) &= 2 g \Phi^{(2,-1)}(z) \\
    J_1 \Phi^{(2,-1)}(z) &= \frac{8g^2 + \Delta^2}{2 g} \Psi^{(1,1)}(z),
  \end{align*}
  where we verify that $p_{1}(1+\tfrac{1}{2} - g^2; g ,\Delta) = 8g^2 + \Delta^2$ as expected.
  It is important to note that the computations are subject to the constraint condition
  $P_1^{(1, \frac12)}((2g)^2, \Delta^2)=0$ for the Juddian eigenvalue $\lambda=\frac{3}{2}-g^2$ (for $N=1$ and $\ell=1$), that is,
  \[
    \Delta^2 + 4g^2 -2 = 0.
  \]
  
  In other words, up to a constant normalization, the coefficients of the Juddian solutions
  actually take values in the domain $\Q[g,\Delta]/(\Delta^2 + 4g^2 -2)$.
  
\end{ex}

\begin{rem}

Corollary \ref{cor:Jaction} above shows that $J_\ell = \mathcal{P}Q_0^{(\ell)}$ sends, up to constant, the solution $\Phi^{(N+\ell,-\ell)}$ to $\Psi^{(N,\ell)}$. Recall the fact that the leading
  expansion (relative to the degree in the Weyl algebra $\C[z, \partial_z]$) of the operator $Q_0^{(\ell)}$ is expressed as  
  \[
    Q_0^{(\ell)}=
    \begin{bmatrix}       
      0 & \partial_z^{\ell} \\
      (-1)^\ell z^{\ell} & 0
    \end{bmatrix}
    + R_\ell,
  \]
  where $R_\ell$ is a polynomial (in $z$ and $\partial_z$) matrix with components of degree strictly less than $\ell$ (see the proof of
  Lemma 4.7 in \cite{RBW2021}). Notice that the first component
  of the polynomial part of $\Phi^{(N+\ell)}$ is of degree $N+\ell-1$ whereas the degree of the first component of the polynomial of  $\Psi^{(N)}$
  equals $N$. This fact implies that there must be a non-trivial cancellation given by the constraint relations.

\end{rem}

The key point of the proof of Theorem \ref{thm:decomposition} is the existence of solutions that are not parity solutions,
that is, common eigenfunctions of $\HRabi{\ell}$ and $J_{\ell}$. It is actually possible to construct the parity solutions directly
as we show in the following example (the QRM case was originally done by Ku\'s in \cite{K1985JMP}).

\begin{ex} \label{ex:paritysol} 
  Let us give and example of parity solutions for the case $N=1$ and $\ell=1$. Let us fix the solutions
  $\Psi^{(1,1)}(z)$ and $\Phi^{(2,-1)}(z)$ as in Example \ref{ex:action}.

  The parity solutions $\Pi_{(1,1)}^{\pm} = (\pi^{\pm}_1,\pi^{\pm}_2)^t$ are then given by
  \begin{align*}
    \pi^{\pm}_1(z) &=   \phi^{(2,-1)}_1(z)  \pm \frac{\mu_{\lambda}}{2 g} \psi^{(1,1)}_1(z) \\
    \pi^{\pm}_2(z) &=  \phi^{(2,-1)}_2(z)  \pm \frac{\mu_{\lambda}}{2 g} \psi^{(1,1)}_2(z),
  \end{align*}
  here $\mu_{\lambda}^2 = 8g^2 + \Delta^2$ and $\mu_{\lambda}$ is one of the square roots.
\end{ex}

Note that the example shows that finding the parity solutions explicitly and finding $\alpha$ and $\beta$ of Corollary \ref{cor:Jaction}
are equivalent. In particular, both require the explicit computation of the action of $J_{\ell}$ on the Juddian solutions.

We also remark that Theorem \ref{thm:decomposition} above forbids the possibility of Juddian eigenfunctions in $\Jker{\ell}$.

\begin{cor} \label{cor:NoJuddianKernel}
  No Juddian solutions $\phi_{\lambda}$ of $\HRabi{\ell}$ are in the kernel $\Jker{\ell}$ of $J_{\ell}$.
\end{cor}

\begin{proof}
  Suppose that $\phi_{\lambda}$ is in $\Jker{\ell}$, then, since Juddian eigenvalues of $\HRabi{\ell}$ are degenerate, then
  \[
    V_{\lambda} = \C \phi_{\mu} \oplus \C \phi_{-\mu},
  \]
  where $\mu^2 = p_{\ell}(\lambda; g, \Delta) = 0$, and therefore $\mu=-\mu=0$, contradicting Theorem \ref{thm:decomposition}.
\end{proof}

\begin{rem}
Since $p_{\ell}(\lambda; g, \Delta)>0$ when $\lambda$ is a Juddian eigenvalue, by this corollary, we easily obtain the projection of $V_{\lambda}$ onto each eigenspace of eigenvalues $\pm \mu_\lambda$ as 
\begin{align}\label{eq:projection}
  P_\pm(\lambda):= \frac12 ({\rm{I}}_2\pm \tilde{J}(\lambda)), 
\end{align}
where we set the involution on $V_{\lambda}$ as 
\[
  \tilde{J}(\lambda):=p_{\ell}(\lambda; g, \Delta)^{-\frac12}J_\ell|_{V_\lambda}.
\]  
\end{rem}

It is worth remarking that we cannot discard the possibility of non-degenerate (i.e. non-Juddian) exceptional solutions in the kernel of $J_{\ell}$. This situation is equivalent to the identity
\begin{equation}
  \label{eq:kernelExcept}
  p_{\ell}(N+\tfrac{\ell}{2}-g^2; g, \Delta) = 0,
\end{equation}
for some $N \geq 0$. It is thus interesting to see the expression of $p_{\ell}(N+\frac{\ell}{2}-g^2; g,\Delta)$ for small
 values of $\ell$. For instance, we have
\begin{align*}
  &p_{0}(N -g^2; g ,\Delta) =1, \\
  &p_{1}(N+\tfrac{1}{2}-g^2; g ,\Delta) = 4 g^2 (N+1) + \Delta^2 \\
  &p_{2}(N+1-g^2; g ,\Delta) =  4 (N+1)(N+2) (2g^2)^2 + 2 (3 + 2 N) (2 g^2) \Delta^2 + \Delta^4 + \Delta^2,
\end{align*}
in particular, we note that all the coefficients are positive. This suggests that Corollary \ref{cor:NoJuddianKernel} may be extended
for any type of exceptional solution. We continue this discussion in the next section.

\section{Symmetry and degeneracy of the \ibQRM{\ell}}
\label{sec:constr-polyn}

In this section we finally are in a position to state the relation between the symmetry operator and the degeneracy in the
spectrum of the \ibQRM{\ell} announced in the introduction.  As we mentioned in Example \ref{ex:action}, for the computation of the constants $\alpha$ and $\beta$ of Corollary \ref{cor:Jaction}, the use of the constraint conditions
\[
  \cp{N,\frac{\ell}{2}}{N}((2g)^2,\Delta^2) = 0, \qquad \cp{N+\ell,-\ell/2}{N+\ell}((2g)^2,\Delta^2) = 0
\]
is fundamental. In fact, in \cite{KRW2017} it was shown that the two constraint conditions satisfy a divisibility relation. Concretely,
\begin{align}
  \label{eq:div}
  \cp{N+\ell,-\ell/2}{N+\ell}(u,v) =  A^\ell_N(u,v) \cp{N,\ell/2}{N}(u,v)
\end{align}
where \( A^\ell_N(u,v)\) is a polynomial satisfying \( A^\ell_N(u,v) >0\) for $u, v > 0$. The change of variable $u=(2g)^2$ and $v=\Delta^2$
relates the above equations with the system parameters. In particular, they have the same roots, as was expected
from the double degeneracy of Juddian solutions of the \ibQRM{\ell}.

In addition, the polynomial $A^\ell_N(u,v)$ has the have the determinant expression
\[
  A^\ell_N(u,v) = \frac{(N+\ell)!}{N!}\det\Tridiag{u+\frac{v}{N+i}-\ell+2i-1}{1}{-i(\ell-i)}{1\le i\le \ell}.
\]
Here, we used the notation 
\begin{equation*}
  \Tridiag{a_i}{b_i}{c_i}{1\le i\le n}
  :=\begin{bmatrix}
    a_1 & b_1 & 0 &  \cdots & 0    \\
    c_1 & a_2 & b_2 &  \cdots  & 0\\
    \vdots & \ddots   & \ddots &  \ddots & \vdots   \\
    0 &  \cdots &  0  & a_{n-1} & b_{n-1} \\
    0 & \cdots  & 0  & c_{n-1} & a_n
  \end{bmatrix}
\end{equation*}
for a tridiagonal matrix. For the full discussion we refer the reader to \cite{KRW2017} (see also \cite{LB2015JPA}).

\begin{rem}
  The first half of the statement of Theorem 3.17 in \cite{KRW2017} should be read as 
  ``In AQRM with bias parameter $\frac\e2$ (in the notation of the present paper), the necessary and sufficient condition for
  the spectral curves $E=\lambda(g)$ with respect to $g$ to have a  crossing point is that $\frac\e2$ is a half-integer''.
    
  In addition, we note that in Corollary 3.5 of \cite{KRW2017}, the additional condition $k \leq N$ should be added.
  The condition is clear from the context and does not affect the results in \cite{KRW2017}.
\end{rem} 

In this setting, the quotient factor $A^\ell_N(u,v)$ in (\ref{eq:div}) does not appear to have an immediate interpretation in terms of
Juddian solutions. The first few values of $A^\ell_N(u,v)$ for fixed $N \ge 0$ are given by
\begin{align*}
  &A^0_N(u,v) = 1, \\
  &A^1_N(u,v) = (N+1)u +v,   \\
  &A^2_N(u,v) = (N+1)_2 u^2 + \biggl(\sum_{i=1}^2 (N+i) \biggr)u v + v (1+ v), 
\end{align*}
and we refer the reader to Example 3.3 of \cite{KRW2017} for further examples. Note that, as polynomials
on the variable $N$, we have $\deg_N(A_N^\ell((2g)^2,\Delta^2)) = \deg_N(p_\ell(N+\frac{\ell}{2}-g^2; g , \Delta)) =\ell$.
Surprisingly, we verify that the equality
\[
  A_{N}^\ell((2g)^2,\Delta^2) = p_{\ell}(N+\frac{\ell}{2}-g^2; g ,\Delta),
\]
actually holds for $0 \leq \ell \leq 6$, that is, they are identical as polynomials in $N$. These observations  and the resulting implications that we discuss in section \S \ref{sec:symm-decomp} suggest the following conjecture.

\begin{conject} \label{conject:main}
  For all natural numbers $N,\ell\geq 0$, we have
  \begin{align} \label{eq:main}
  p_{\ell}(N+\tfrac{\ell}{2} - g^2;g,\Delta) &= A^\ell_N((2g)^2,\Delta^2),
  \end{align}
  and, for $N \ge \ell$, we have
  \begin{equation*}
    p_{\ell}(N- \tfrac{\ell}{2} - g^2;g,\Delta)= A^\ell_{N-\ell}((2g)^2,\Delta^2).
  \end{equation*}
\end{conject}

The equation of Conjecture \ref{conject:main} is highly significant, as it gives a direct realization of the relation between
degeneracy and symmetry. In other words, it gives the explicit relation between the equation
\[
  J_{\ell}^2 = p_{\ell}(\HRabi{\ell};g,\Delta),
\]
and the divisibility relation 
\begin{align*}
  \cp{N+\ell,-\ell/2}{N+\ell}(u,v) =  A^\ell_N(u,v) \cp{N,\ell/2}{N}(u,v)
\end{align*}
assuring the existence of degeneracies in the spectrum of the \ibQRM{\ell}.

There is currently no proof of the remarkable statement in Conjecture \ref{conject:main}. In the rest of the paper we focus on the discussion of the consequences of the conjecture and giving further evidence for it, while pointing out possible approaches for the proof.

As an immediate consequence, since $A^\ell_N(u,v)>0$ for $u,v > 0$ we can extend the result of Corollary \ref{cor:NoJuddianKernel}
to any type of exceptional solutions $\lambda = N + \tfrac{\ell}{2} - g^2$ (not limited to Juddian solutions).

\begin{prop} \label{prop:NoExceptKernel}
 If Conjecture \ref{conject:main} is true, there are no exceptional solutions $\phi_{\lambda}$ of $\HRabi{\ell}$ associated to
 eigenvalues $\lambda = N + \tfrac{\ell}{2} - g^2 $ with $N\ge0$ in the kernel $\Jker{\ell}$ of $J_{\ell}$.
\end{prop}

Note that non-Juddian exceptional solutions with eigenvalues $\lambda = N - \tfrac{\ell}{2} - g^2$ for $N \leq \ell$ are not ruled out by the Proposition
\ref{prop:NoExceptKernel}. In \S \ref{sec:symm-decomp} we discuss this situation along with further consequences and implications of the main conjecture.

Before concluding this section, we shortly present another aspect of Conjecture \ref{conject:main}. Let us note the presence of both polynomials in equation \eqref{eq:main} in the $\mathfrak{sl}_2$-picture for the eigenvalue problem of the Hamiltonian $\HRabi{\ell}$. We direct the reader to \cite{W2016JPA} for the full description of the $\mathfrak{sl}_2$-picture. We recall that the representation theoretic picture allows us to describe the defining recurrence equation of the constraint polynomial $\cp{N,\ell/2}{N}((2g)^2,\Delta^2)$ by the continuant (i.e. the determinant of the triangular matrix) which describes the eigenfunction in the space of the corresponding irreducible finite dimensional representation $\mathbb{F}_{N} ( \dim_{\C}\mathbb{F}_N=N)$ of $\mathfrak{sl}_2$. 
  In particular, there is a $(N+\ell+1)\times(N+\ell+1)$ matrix $M_{N+\ell}^{(N+\ell,-\frac{\ell}{2})}$ (resp. $(N+1)\times(N+1)$ matrix $M_{N}^{(N,\frac{\ell}{2})}$)
  such that, up to a constant
  \[
    \det M_{N+\ell}^{(N+\ell,-\frac{\ell}{2})} = \cp{N+\ell,-\ell/2}{N+\ell}((2g)^2,\Delta^2), \qquad \det M_{N}^{(N,\frac{\ell}{2})} = \cp{N,\ell/2}{N}((2g)^2,\Delta^2).
  \]
  It is known (see Proposition 5.8 of \cite{W2016JPA}) that if the determinant vanishes, that is, the matrix is singular, then the rank of $M_{N+\ell}^{(N+\ell,\frac{\ell}2)}$ (resp. $M_{N}^{(N,\frac{\ell}{2})}$) is exactly $N+\ell$ (resp. $N$) and any vector in the kernel of $M_{N+\ell}^{(N+\ell,\frac{\ell}2)}$
  (resp. $M_{N}^{(N,\frac{\ell}{2})}$) corresponds to a Juddian solution $\lambda = N + \tfrac{\ell}{2} -g^2$. The linear independence
  of the solutions is (again) verified by the fact that they correspond (with exception to the case $\ell=1$) to different finite dimensional irreducible representations of $\mathfrak{sl}_2$.
  We now define the $(N+\ell+1)\times(N+\ell+1)$-matrix $\widetilde{M}_{N}^{(N,\frac{\ell}{2})}$  from $M_{N}^{(N,\frac{\ell}{2})}$ as the block 
  \[
    \begin{bmatrix}
      M_{N}^{(N,\frac{\ell}{2})} & \bm{0} \\
      \bm{0} & \bm{I}_\ell
    \end{bmatrix}.
  \]
  Next, we note that by elementary matrix operations we can find non-degenerate $(N+\ell+1)\times(N+\ell+1)$-matrices $\mathcal{B}$ and $\mathcal{C}$ such that
  \[
    M_{N+\ell}^{(N+\ell,-\frac{\ell}2)} \mathcal{B}= \bm{I}_{N+\ell+1} = \widetilde{M}_{N}^{(N,\frac{\ell}{2})} \mathcal{C},
  \]
  and by setting $\mathcal{A}= \mathcal{A}^\ell(N, g, \Delta)= \mathcal{C}\mathcal{B}^{-1} \in GL_{N+\ell+1}(\R)$ we obtain
  \[
    M_{N+\ell}^{(N+\ell,-\frac{\ell}2)} = \widetilde{M}_{N}^{(N,\frac{\ell}{2})} \mathcal{A},
  \]
  with $\det \mathcal{A} = A_N^{\ell}((2g)^2,\Delta^2)>0$. We note that in particular,  if $\bm{v}_0=\bm{v}_0(\lambda)(\neq0) \in {\rm Ker}( M_{N+\ell}^{(N+\ell,-\frac{\ell}2)})$,
  then $\mathcal{A} \bm{v}_0 \in \widetilde{M}_{N}^{(N,\frac{\ell}{2})}$ and $\{ \bm{v}_0, \mathcal{A} \bm{v}_0 \}$ is a basis of the eigenspace 
 $\tilde{V}_{\lambda}(\simeq V_{\lambda})$ with $\lambda = N + \frac{\ell}{2} - g^2$. It is known that the two dimensional space 
 $\tilde{V}_{\lambda} \subset \mathbb{F}_{N+1}\oplus \mathbb{F}_{N+\ell}$.
In other words, $\bm{v}_0$ (resp. $\mathcal{A} \bm{v}_0$) is an alternative expression of the eigenfunction $\Phi^{(N+\ell)}$ (resp.  $\Psi^{(N)}$ ) in
  terms of $\mathfrak{sl}_2$-representation. Hence, by Theorem \ref{thm:decomposition}, neither  $\bm{v}_0$ nor  $\mathcal{A} \bm{v}_0$ can
  be the eigenvector of (the corresponding image of) $J_\ell$ on $V_\lambda$. 

  Let us now consider the action of $J_{\ell}$ in $V_{\lambda}$. We denote $J_{\ell} |_{V_{\lambda}}$ by $J(\lambda)$. Since the eigenvalues of
  $J_{\ell}$ are $\pm \mu_\lambda$, we must have $\tr J(\lambda) = 0$. It follows that $\det J(\lambda) = - p_{\ell}(\lambda; g,\Delta)$ by the Cayley-Hamilton theorem. Therefore, Conjecture \ref{conject:main} is rewritten as
  \[
    \det \mathcal{A}^\ell(N, g, \Delta) = - \det J(N+\frac{\ell}{2}-g^2).
  \]
  Notice that, although $J(N+\frac{\ell}{2}-g^2) \bm{v}_0 = c\,\mathcal{A} \bm{v}_0$ for some constant $c$ by Corollary \ref{cor:Jaction}, this approach does not seem to give nor simplify the proof of the conjecture. In Figure \ref{dia:diagramJaction} we show the diagram of the action of the operators $J(\lambda)$ and $\mathcal{A}$ in the particular elements of the degenerate eigenspace $V_{\lambda}$. The conjecture is reduced to showing that $\mathcal{A}$ sends $\mathcal{A} \bm{v}_0$ to $\bm{v}_0$ (i.e. the action giving by the squiggly arrow in the diagram of Figure \ref{dia:diagramJaction}), equivalently that  ${\rm tr} A = 0$. 

  \begin{figure}[h]
    \centering
    \begin{tikzcd}
  V_{\lambda}\ni \Phi^{(N+\ell)} 
  \arrow[r,mapsto,"\simeq"]
  \arrow[d,mapsto,"J(\lambda)"']
  & \bm{v}_0\in \tilde{V}_\lambda
  \arrow[d,mapsto,"\mathcal{A}"',shift right =3.5ex]
  \\
  V_{\lambda}\ni \Psi^{(N)} \arrow[u,mapsto,shift right=1.5ex]   \arrow[r,mapsto,"\simeq"] &  \mathcal{A}\bm{v}_0\in \tilde{V}_\lambda
  \arrow[u,squiggly,shift left=1ex,"?"']
\end{tikzcd}
    \caption{Actions of $J(\lambda)$ and $\mathcal{A}$}
    \label{dia:diagramJaction}
  \end{figure}

  Another possible approach to the proof may be to consider the analytic continuation of the non-unitary principal series
  representation to discrete series representation, which are regarded as the sub-quotient of non-unitary principal series
  (cf. \cite{HT1992, VR1976}). In fact, the non-Juddian exceptional eigenfunction, with eigenvalue of the form $N+\frac{\ell}2-g^2$ for
  some $N\in \N$, is regarded as a vector that belongs to one of the lowest weight irreducible representations (i.e. discrete series representations) of $\mathfrak{sl}_2(\R) $\cite{KRW2017}. 

\begin{rem}
  Recall that Juddian eigenstates are captured in a finite dimensional irreducible representation \cite{W2016JPA}. Since finite dimensional representations never co-exist with the discrete series representations as sub-representation nor sub-quotient of the same non-unitary principal series \cite{KRW2017}, there is no chance that Juddian and non-Juddian can share the same eigenvalue.
\end{rem}

It is also worth remarking that the distinction of spherical and non-spherical representations given in \cite{W2016JPA}
does not give a ``parity'' of the eigenspaces of $\HRabi{\ell}$. Actually, when $\ell$ is even the pair $(\bm{v}_0, \mathcal{A} \bm{v}_0)$ belongs either (spherical, non-spherical)-representations or (non-spherical, spherical)-representations depending on the parity of $N$, but when $\ell$ is odd the pair belongs either (spherical, spherical)-representations or (non-spherical, non-spherical)-representations (see Table 1 and 2 in \cite{W2016JPA}).

\section{Consequences of the main conjecture}
\label{sec:symm-decomp}

In this section we assume that Conjecture \ref{conject:main} holds and discuss several significant consequences. In addition, some of the
results of this section can be considered as further evidence for the validity of Conjecture \ref{conject:main}.

First, it is natural to attempt to extend the local property (i.e. observed at degenerate points) of Conjecture \ref{conject:main}
to get a determinant expression for the polynomials $p_{\ell}(x;g,\Delta)$ in general. 

\begin{thm} \label{thm:detexpP}
  Under the assumption of Conjecture \ref{conject:main}, the polynomial $p_{\ell}(x;g,\Delta)$ is given by
  \begin{equation} 
    \label{eq:polyPandA}
  p_{\ell}(x;g,\Delta) = \det \left( \Delta^2 \bm{I}_\ell + \bm{M}_{\ell}(x,g) \right),
  \end{equation}
  where $\bm{M}_{\ell}(x,g)$ is the matrix
  \[
    \bm{M}_{\ell}(x,g) = \Tridiag{((2g)^2-\ell+2i-1)(x - \tfrac{\ell}2 + g^2 + i) }{(x - \tfrac{\ell}2 + g^2 + i)}{-i(\ell-i)(x - \tfrac{\ell}2 + g^2 + i+1)}{1\le i\le \ell}.
  \]
\end{thm}

\begin{proof}
 Denote by $q(x; g, \Delta)$ the right-hand side of equation (\ref{eq:polyPandA}). Then, 
 we verify that
 \[
   q(N + \frac{\ell}{2} -g^2; g,\Delta) = A_{N}^{\ell}((2g)^2,\Delta^2) = p_{\ell}(N + \frac{\ell}{2} -g^2;g,\Delta),
 \]
 holds for $N = 0,1,\ldots,\ell$. Since $q(x; g, \Delta)$ and $p_{\ell}(x;g,\Delta)$ are polynomials of degree $\ell$ the equality follows.
\end{proof}

\begin{rem}
    Clearly, an unconditional proof of Theorem \ref{thm:detexpP} is equivalent to a proof of the main Conjecture \ref{conject:main}.
However, it is not easy to know which statement is more difficult to prove.
The complex geometric picture in \S \ref{sec:hyperelliptic} may be useful for such an approach.
\end{rem}

Note that the determinant expression of the tridiagonal matrix in the right-hand side of \eqref{eq:polyPandA} defines a family of polynomials by a three-term recurrence relation (similar to the case of orthogonal polynomials). In particular, it gives an expression of $p_\ell(x; g,\Delta)$ in terms of sums and products of lower degree polynomials. However, there does not appear to be a way to obtain the general form of the operator $J_{\ell}$ from such expression.

Let us recall from \cite{RBW2021} that the roots of the polynomial $p_{\ell}(x;g,\Delta)$ provide a good approximation for the
  first $\ell$ eigenvalues for $g/\Delta\gg 1$ (the deep strong coupling region \cite{YS2018}) for the explicitly computed cases. Concretely, for a fixed $\Delta>0$ the curve
\[
  p_{\ell}(E; g ,\Delta) = 0
\]
in the $(E,g)$ approximates the first $\ell$ eigencurves of $\HRabi{\ell}$. This suggests another relation of the hyperelliptic curve
\eqref{eq:hyperelliptic} and the spectrum of $\HRabi{\ell}$. The explanation of the approximation property and its relation with
the hidden symmetry is one of the open problems in the study.  

Using the determinant expressions of Theorem \ref{thm:detexpP}, we can give further evidence for Conjecture \ref{conject:main}.
In Figure \ref{fig:Eigencurves2} we show the spectral curves of $\HRabi{\ell}$ along with the curves defined by the right hand
side of \eqref{eq:polyPandA}, that is, $p_{\ell}(x;g,\Delta) =0$, for cases that the polynomial $p_{\ell}(x;g,\Delta)$ has not been
computed directly. The curves given by the determinant expression have the expected shape and approximate the first $\ell$
spectral curves for $g/\Delta \gg 1$ as expected. In \S \ref{sec:approximation-1} we revisit the approximation property
from a geometric point of view.

\begin{figure}[h!]
  \centering
  \subfloat[$\ell = 6$]{
    \includegraphics[height=4cm]{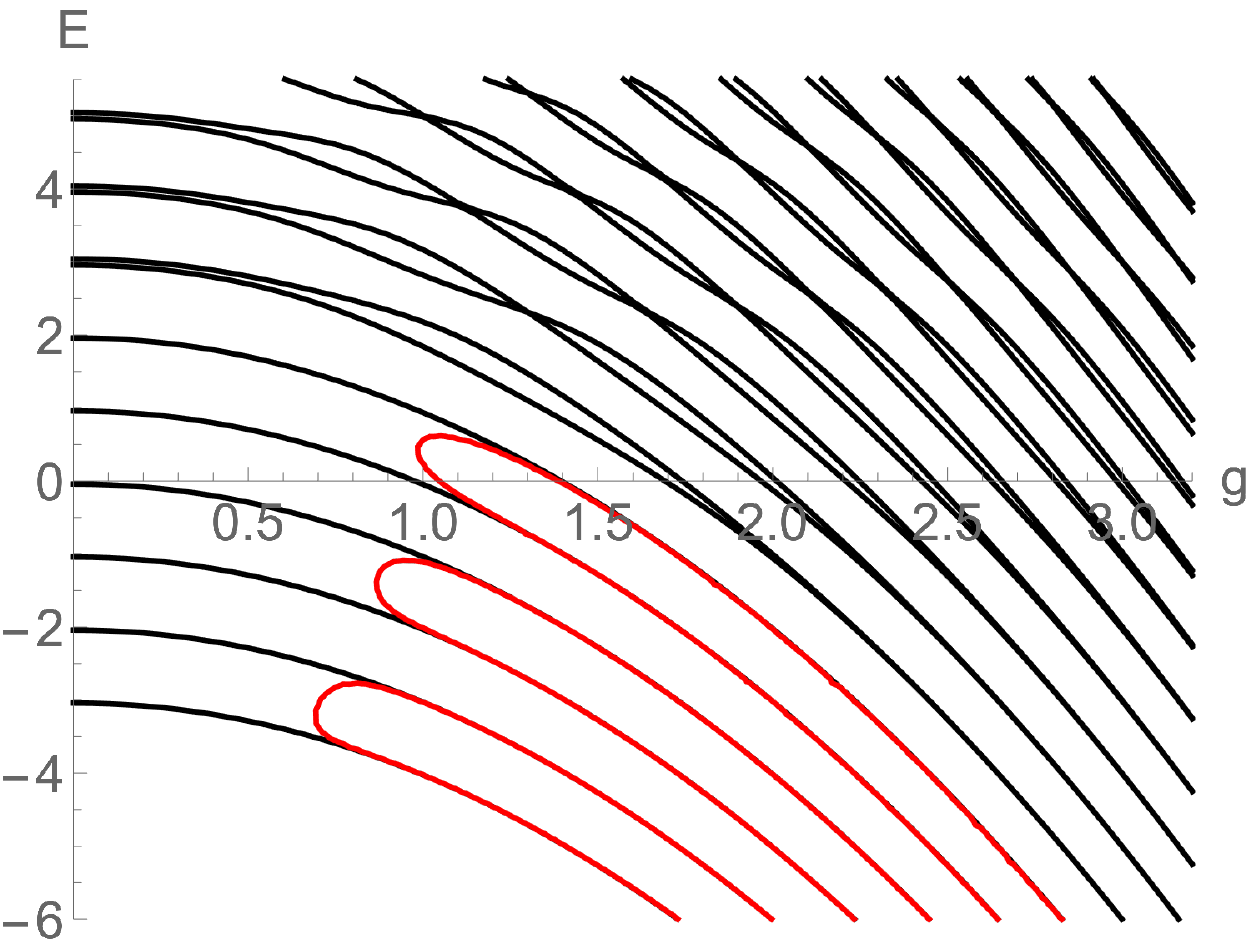}}
  ~
  \subfloat[$\ell = 10$]{
    \includegraphics[height=4cm]{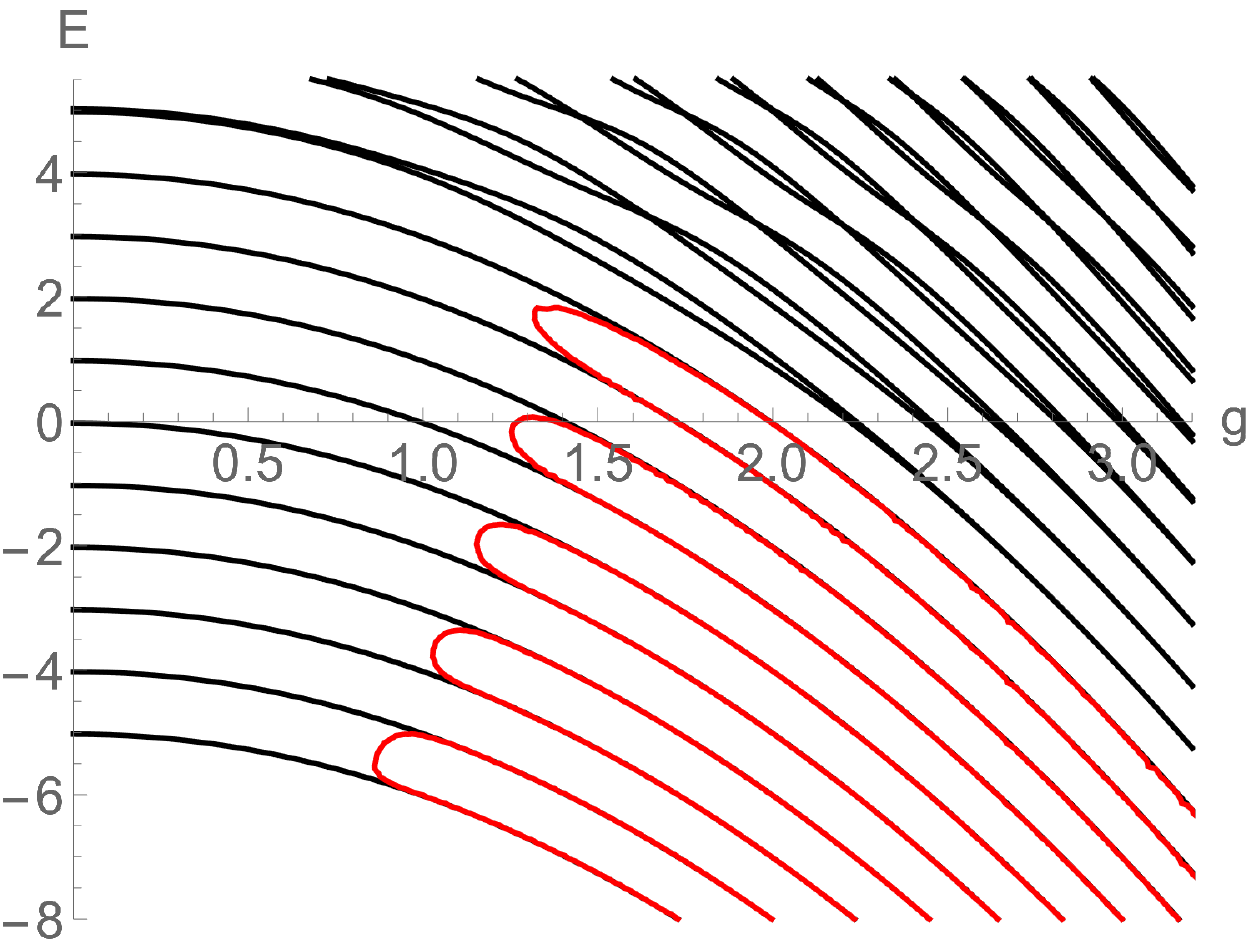}}
  ~
  \subfloat[$\ell = 17$]{
    \includegraphics[height=4cm]{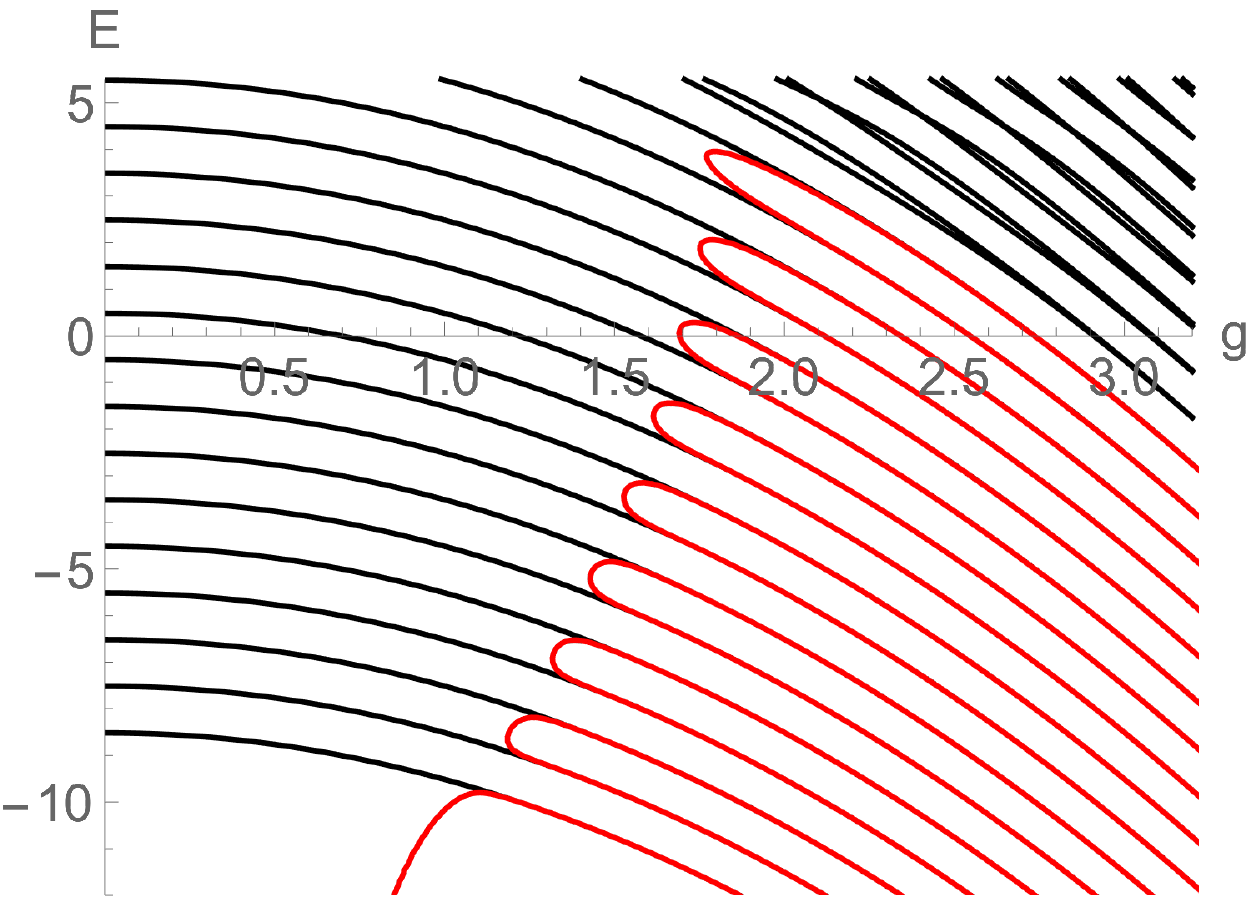}} \\
  \caption{Spectral curves (grey) and curves defined by the determinant expressions \eqref{eq:polyPandA} (red) for $\Delta = \frac12$.}
  \label{fig:Eigencurves2}
\end{figure}

\begin{rem}
  It would be interesting to verify if for other biased models with hidden symmetry and with
quadratic relation between the Hamiltonian and the commuting operator (e.g. \cite{LLMB2021a, LB2020,XC2021}) similar approximation properties hold.
\end{rem}

As we have discussed before, the question of determining whether the kernel $\Jker{\ell}$ of the operator $J_{\varepsilon}$ is trivial or not is
a delicate one and it appears to be deeply connected to the question of whether the parity decomposition of the QRM can be extended to the general \ibQRM{\ell}.
In \S~\ref{sec:preliminary-results} we showed there are no exceptional eigenfunctions of Juddian type in $\Jker{\ell}$ and under the
assumption of Conjecture \ref{conject:main} that there are no exceptional eigenvalues in general.

  We now observe some applications of the determinant expression of Theorem \ref{thm:detexpP} to the
  analysis of the kernel of $J_{\ell}$. Let $\ell\geq0$ and $N \geq 1$, we recall that for fixed $u>0$ (resp. $v>0$) the constraint polynomial
  $\cp{N,\ell /2}{N}(u,v)$ has real roots with respect to the variable $v$ (resp. $u$). By the divisibility  relation \eqref{eq:div} the
  same property holds for the polynomial $A_N^{\ell}(u,v)$. Currently it is not known whether an analog of the divisibility relation
  \eqref{eq:div} involving the polynomial $p_{\ell}(x; g , \Delta)$ for $x - \frac{\ell}{2} +g^2 \notin \Z_{\ge0}$
  exists and even if the answer is affirmative the nature of the objects corresponding to the constraint polynomials is unknown.
  Nevertheless, it is natural to expect that for $\Delta\neq 0$ we have the property
  \begin{equation}
    \label{eq:claimS}
    \text{ the polynomial } p_{\ell}(x; g, \Delta) \text{ has real roots with respect to the variable $g$ for } 
    x - \tfrac{\ell}{2} + 1 +g^2 > 0 \tag{$\star$}.
  \end{equation}
  
\begin{lem} \label{lem:positive}
  Suppose that Conjecture \ref{conject:main} holds. For fixed $g,\Delta>0$, if $x > \frac{\ell}{2}- 1 -g^2$ with $g>\frac{\sqrt{\ell-1}}{2}$ we have
  \[
    p_{\ell}(x; g, \Delta) >0.
  \]
  Moreover, further assuming  \eqref{eq:claimS}, we obtain the same estimate $p_{\ell}(x; g, \Delta) >0$ when $x > \frac{\ell}{2}- 1 -g^2$
  for any $g >0$.
\end{lem} 

\begin{proof}
  If $g>\frac{\sqrt{\ell-1}}{2}$ and $x > \frac{\ell}{2}- 1 -g^2$, from the determinant expression of Theorem \ref{thm:detexpP} it is
  immediate to verify that $p_\ell(x; g, \Delta)$ is a polynomial with positive coefficients and the results follows immediately.

  Next, let us assume \eqref{eq:claimS}, then since we know that for any eigenvalue $x$,  $p_{\ell}(x; g, \Delta)\ge0$, it is enough to show
  that it is not vanishing for $g,\Delta >0$. By \eqref{eq:polyPandA}, it is enough to verify that for fixed $g$ and $x$, $0$ is not an
  eigenvalue of the matrix $\bm{M}_{\ell}(x,g)$. As in the case of the polynomial $A_N^{\ell}((2g)^2,\Delta^2)$ (see Lemma 3.8 or
  Proposition 3.13 of \cite{KRW2017}), we see that
  \[
    \det \bm{M}_{\ell}(x,g) = p_{\ell}(x; g, 0) = (2g)^{2 \ell} (x - \tfrac{\ell}{2} + 1 + g^2)_{\ell},
  \]
  where $(a)_{\ell}$ is the Pochhammer symbol. By the condition $x > \frac{\ell}{2}-1 -g^2$, $\det \bm{M}_{\ell}(x,g)= 0$ if and only
  if $g=0$, then the proof follows in the same way as the proof of positivity of $A_N^{\ell}((2g)^2,\Delta^2)$ in \cite{KRW2017}.
\end{proof}

It is not easy to discount the possibility of vanishing of the polynomial $p_{\ell}(\lambda; g, \Delta)\ge0$ on the case of $\lambda \leq - \frac{\ell}2 -g^2$.
The example below shows that in order to prove that $\Jker{\ell}$ is trivial for all parameters $g>0$ it is not enough
to consider the properties of the polynomial $p_{\ell}(x; g ,\Delta)$.

\begin{ex} \label{ex:K0region}
  For $\ell=1$, we have
  \[
    p_{\ell}(\tfrac{\ell}{2}-1-g^2+\delta; g, \Delta) = \Delta^2 + 4 \delta g^2,
  \]
  and for $\delta > 0$, the polynomial has positive coefficients and therefore no zeros with $g,\Delta\in \R$.

  Similarly, for $\ell=2$, the polynomial
  \[
    p_{\ell}(\tfrac{\ell}{2}-1-g^2+\delta; g, \Delta) = \Delta^4+ \Delta^2+ 8 \Delta^2 g^2 \delta  + 4 \Delta^2 g^2+ 16 g^4 \delta^2+16 g^4 \delta,
  \]
  has positive coefficients for $\delta >0$. Note that in both cases, the polynomial $p_{\ell}(\tfrac{\ell}{2}-1-g^2+\delta; g, \Delta)$ also has positive
  coefficients in the case $\delta=0$. However, Lemma \ref{lem:positive} does not cover this case.

  The examples above also show that for $\delta<0$, the polynomial $p_{\ell}(\tfrac{\ell}{2}-1-g^2+\delta; g, \Delta)$ may have
  real zeros for $g,\Delta\in \R$. For instance, if $\ell=2$, the polynomial
  \begin{equation} \label{eq:vanishingpoly}
    p_{\ell}(-\tfrac{1}{2}-g^2; g,\Delta) = \Delta^4+\Delta^2- 4g^4,
  \end{equation}
  and it vanishes for  $4g^2 = \Delta^4+\Delta^2$. We remark that, however, that the vanishing of \eqref{eq:vanishingpoly}
  only indicates the non-triviality of $\Jker{\ell}$ if
  \[
    \lambda = -\tfrac{1}{2}-g^2 \in \Spec(H_{\ell}).
  \]

  Even if we limit the study to the case of exceptional solutions $\lambda = N \pm \frac{\ell}{2}-g^2$, we cannot
  discard the possibility that the condition $\lambda \in \Jker{\ell}$ may hold in the region $\lambda < \frac{\ell}{2}-1 -g^2$.
  For instance, for $\ell=3$ and $\Delta=4$, in Figure \ref{fig:Eigencurves3} we show the spectral curves of the \ibQRM{\ell} along with the baselines
  $E =i -\frac{\ell}{2}-g^2$ for $i=1,2,3$. The exceptional
  solutions, that is the crossings of the spectral lines with the baselines are shown with circle and square marks. It is known
  that these exceptional solutions must be of non-Juddian exceptional type (see Corollary 3.14 in \cite{KRW2017}). For the
  exceptional solutions with square marks there is a possibility that $\lambda \in  \Jker{\ell}$, however, in practice, condition
  \eqref{eq:kernelExcept} is difficult to verify numerically. For this purpose it may be useful to compute
  the intersection multiplicities of the baselines $E = i \pm \frac{\ell}{2} -g^2$ and the curve $p_{\ell}(E; g ,\Delta) = 0$ for
  arbitrary bias $\ell$.

\begin{figure}[h!]
  \centering
  \includegraphics[height=5cm]{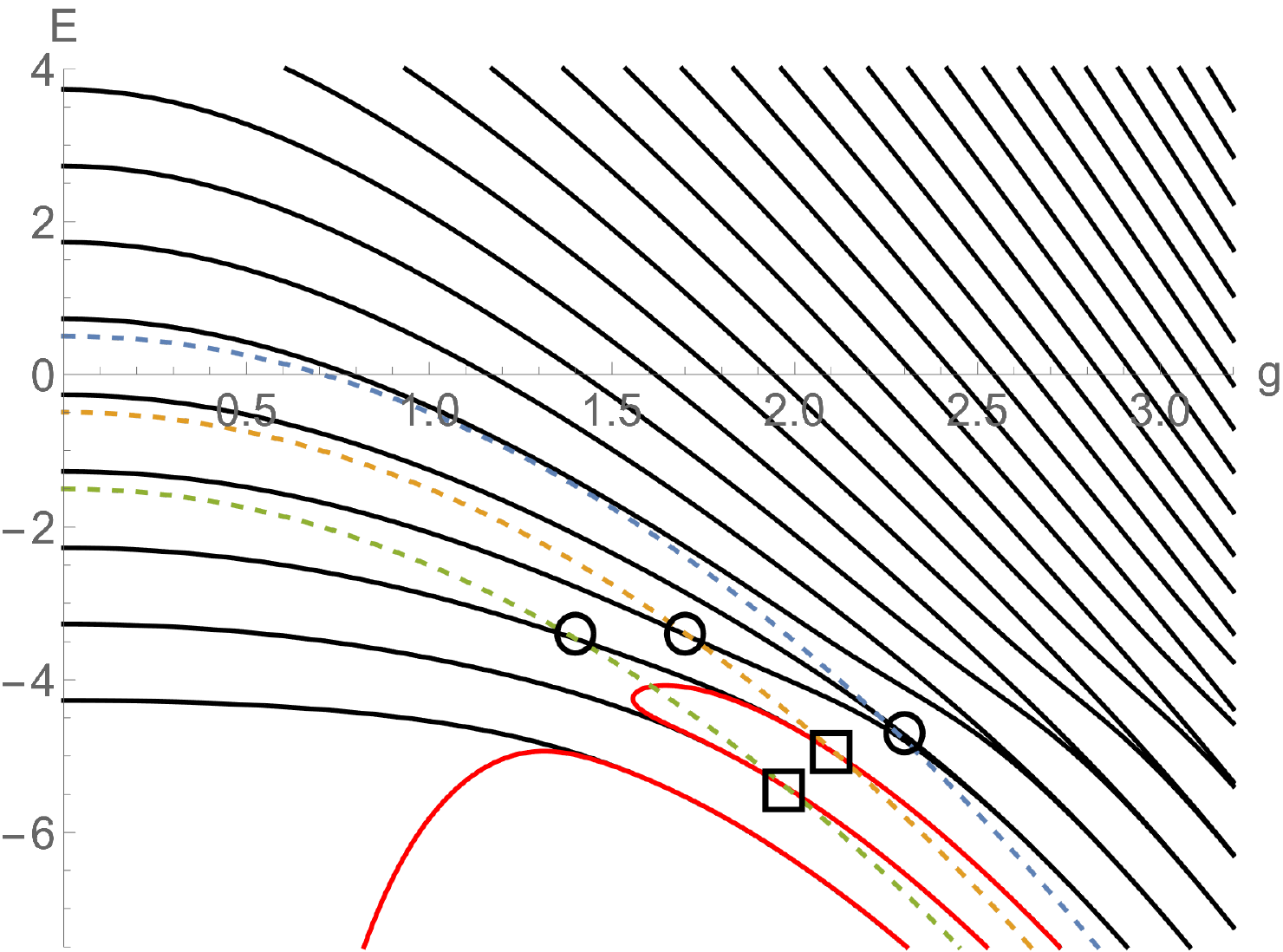}
  \caption{Spectral curves (grey) and curves defined by the determinant expressions \eqref{eq:polyPandA} (red) for $\Delta = 4$ and $\ell=3$.
    The dotted lines correspond to the baselines $E = i-\frac{\ell}{2}-g^2$ for $i=0,1,2$.}
  \label{fig:Eigencurves3}
\end{figure}

\end{ex}


Let us conclude by noting that by the determinant expression of $p_\ell(x; g, \Delta)$, we may write
\[
p_\ell(x; g, \Delta)= (2g)^{2\ell}  \prod_{i=1}^{\ell/2} ((x+\alpha_i)^2+\beta_i),
\]
for $\ell$ is even, and 
\[
p_\ell(x; g, \Delta)= (2g)^{2\ell}(x+\alpha_0)\prod_{i=1}^{(\ell-1)/2} ((x+\alpha_i)^2+\beta_i),
\]
for $\ell$ odd, with $\alpha_i, \beta_i \in \C[\frac1g, \Delta]$ for all $i$. From these expressions, the
shape of the curves
\[
  p_\ell(x; g, \Delta) = 0
\]
curves in the $(g, x)$-plane in the graphs in Figure \ref{fig:Eigencurves2} becomes clear.

\begin{rem}
  By Lemma \ref{lem:positive}, only eigenvalues satisfying $\lambda \leq \frac{\ell}{2}-1 -g^2$ may have the property that $p_{\ell}(\lambda; g, \Delta)=0$.
  The numerical computations have shown that the number of these eigenvalues appears to be $\ell$ (and that these eigenvalues satisfy
  the approximation property for $g\gg 0$). This is compatible with the {\it extended Braak conjecture} (see \cite{B2011} for the
  original case of the QRM) on eigenvalue distribution for the AQRM:
  \begin{enumerate}
  \item In each interval of length $1$ there are at most two eigenvalues of the AQRM.
  \item If an interval $[n, n+1) \, (n\in \Z)$ has two eigenvalues, then the neighboring intervals has no eigenvalues.
  \end{enumerate}
  
  Moreover, as shown in Figure \ref{fig:Eigencurves3},  we note that some of the first $\ell$ eigenvalues may be non-Juddian exceptional
  eigenvalues (cf. Corollary 5.7 in \cite{KRW2017}).

\end{rem}

\section{Geometric picture of the spectrum}
\label{sec:approximation}

  The fact that the relation between the Hamiltonian $H_{\ell}$ and $J_{\ell}$ is given by the quadratic relation
  \[
    J_{\ell}^2 = p_{\ell}(\HRabi{\ell};g,\Delta),
  \]
  determined by the polynomial $p_{\ell}(x; g, \Delta)$ suggests that in order to achieve a deeper understanding of the spectrum of the
  \ibQRM{\ell} and its properties, an investigation of the algebraic curves and surfaces determined by $y^2 = p_{\ell}(x; g, \Delta)$ seems crucial.
  While unconditionally an expression for the polynomial is unknown, under the assumption of Conjecture \ref{conject:main}, Theorem \ref{thm:detexpP} gives a determinant expression that allows the study of the geometric picture of the spectrum even for large values of $\ell$.

For fixed $g,\Delta>0$, note that the equation
\begin{equation} \label{eq:hyperelliptic}
  y^2 = p_{\ell}(x; g ,\Delta),
\end{equation}
describes a plane curve $\mathcal{C}_{\ell} = \mathcal{C}_{\ell}(g,\Delta)$
(hyperelliptic curve in the general case; elliptic curves \cite{K1992} for $\ell=3,4$) such that
\begin{align*}
  \Spec(H_{\ell}) \times \Spec(J_{\ell})& \\
         \cup  \qquad \qquad   &  \\
        \mathcal{D}_\ell(g,\Delta)         \qquad \qquad     & \subset  \qquad \mathcal{C}_{\ell}(g,\Delta),
\end{align*}
where $\mathcal{D}_\ell (= \mathcal{D}_\ell(g,\Delta))$ is the set of tuples of joint eigenvalues  $(\lambda,\mu_{\lambda})$. In other words, the joint
eigenvalues lie on the real plane curve $\mathcal{C}_{\ell}$.

In the study of the spectra of the QRM and its generalizations it is natural to consider the parameter $g$ as a variable and consider
the spectral curves. From this point of view, the equation \eqref{eq:hyperelliptic} describes an algebraic surface
$\mathcal{S}_\ell (= \mathcal{S}_\ell(\Delta)) \subset \R^3$ with variables  $(x,y,g)$.  Next, for $g>0$ the set $\mathcal{L}_\ell ( = \mathcal{L}_\ell(\Delta))$ of tuples $(\lambda(g),\mu_{\lambda}(g),g) \in \R^3$ of eigenvalues of $\HRabi{\ell}$ and $J_{\ell}$ corresponding to joint eigenfunctions is the union of spectral curves in the $(x,y,g)$ plane. By construction, we have $\mathcal{L}_\ell \subset \mathcal{S}_\ell$, that is, the spectral curves of the \ibQRM{\ell} lie on the surface $\mathcal{S}_{\ell}$. The geometric picture is illustrated in Figure \ref{fig:3dQRM} for the cases $\ell=0,1$ where the surface $\mathcal{S}_\ell$
is shown in orange.

\begin{figure}[h]
  \centering
    \subfloat[$\ell = 0,\Delta=1$]{
    \includegraphics[height=5.5cm]{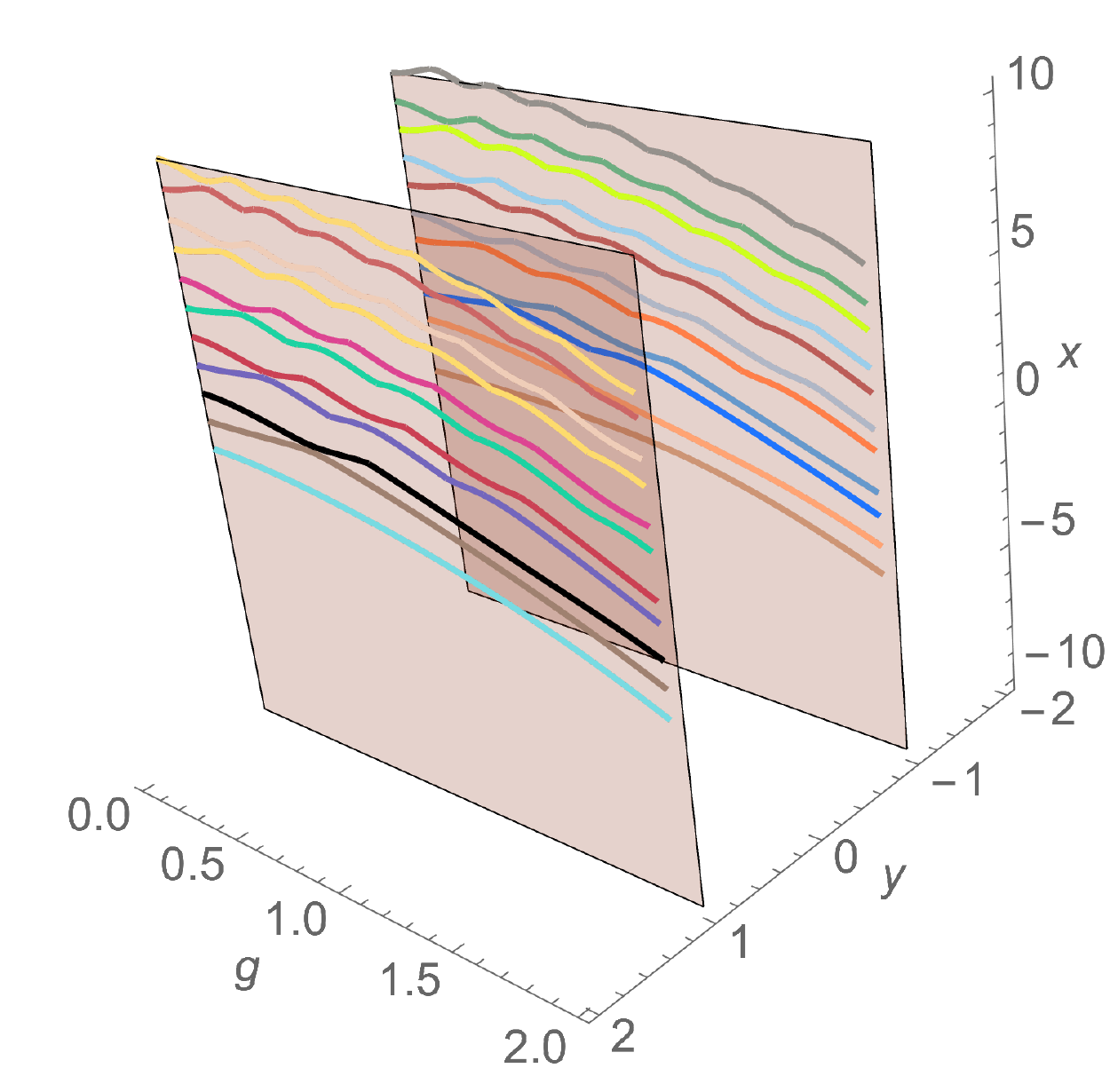}}
  ~
  \subfloat[$\ell=1, \Delta=1$]{
    \includegraphics[height=5.5cm]{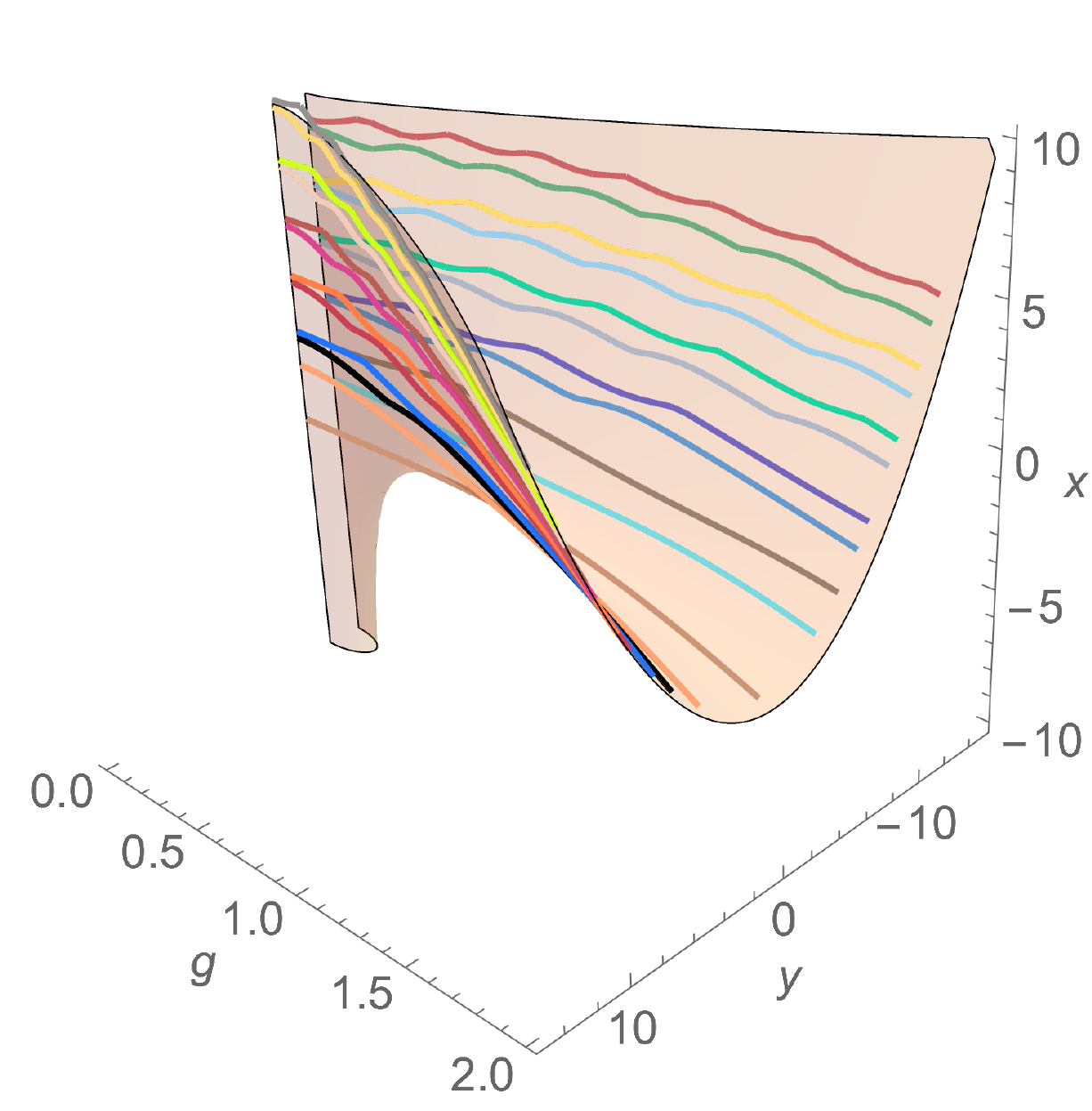}}
    ~
  \subfloat[$\ell=2,\Delta=4$]{
    \includegraphics[height=5.5cm]{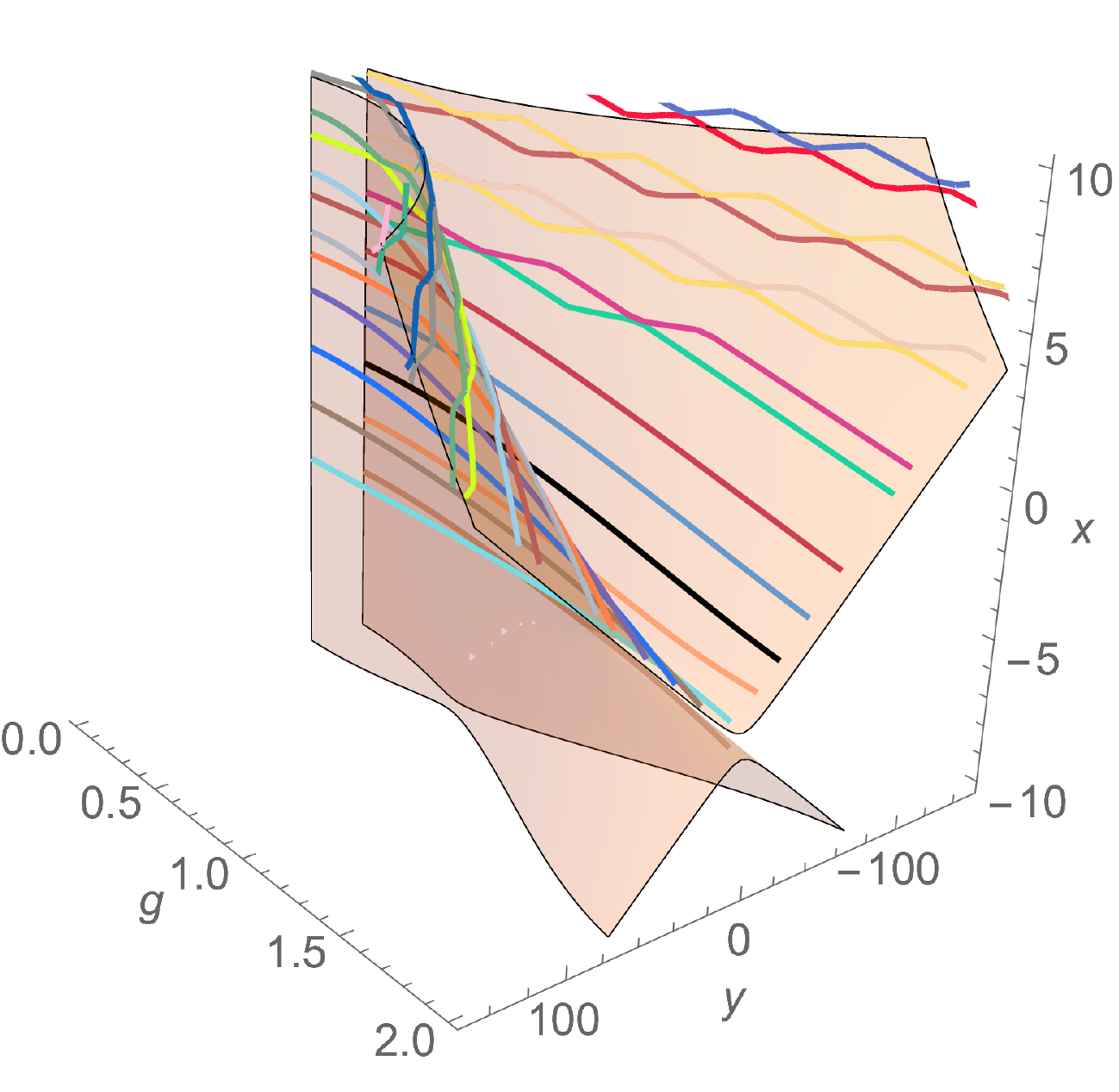}}
  \caption{Geometric picture for $\ell=0,1,2$. (a) In the case $\ell=0$ the parity decomposition appears naturally.
      (c) ``Avoided crossings'' for large values of $\Delta$.}
  \label{fig:3dQRM}
\end{figure}

The geometric picture is thus summarized in Figure \ref{fig:geompict}. By considering the projection of the surface
$\mathcal{S}_{\ell}$ for a fixed $g \in \R_{\ge 0}$ (or $\mathbb{P}^1(\R)/\Z_2$) we obtain the curve $\mathcal{C}_{\ell}$ containing the joint spectrum for the parameter $g>0$. In the elliptic case $\ell=3,4$, the surface $\mathcal{S}_\ell$ and aforementioned projection 
has the structure of an elliptic surface, we consider this case in \S\ref{sec:hyperelliptic}.
On the other hand, the projection of the surface $\mathcal{S}_\ell$ with the spectral curves onto the $(g,x)$-plane recovers
the usual picture of the spectral curves, giving rise to the approximation property of the polynomial $p_{\ell}(x; g, \Delta)$.

\begin{figure}[h]
  \centering

  \begin{tikzpicture}

    \draw node at (0,0) {$(x,y,g) \in \R^2\times \mathbb{P}^1(\R)$};
    \draw node at (0,-0.5) {$\bigcup$};
    \draw node at (0,-1)  {$y^2 = p_{\ell}(x; g ,\Delta) $};
    \draw node at (0,-1.5) {(algebraic surface)};

    \draw [<-] (-1.5,-1.5) -- node[above,sloped] {a section $g=g_0$}  (-5,-3.5);
    
    \draw node at (-6,-4) {$\underline{y}^2 = p_{\ell}(\underline{x}; g_0 ,\Delta)$};
    \draw node at (-6,-4.6) {(hyperelliptic curve ($g_0 \in \mathbb{P}^1(\R)$))};

    \draw [->] (1.5,-1.5) -- node[above,sloped] {proj. into $(g,x)$-plane} (5,-3.5);

    \draw node at (6,-4) {$0 = p_{\ell}(\underline{x}; \underline{g} ,\Delta)$};
    \draw node at (6,-4.5) {($(g,E)$-plane curves)};

     \draw [<->] (2.2,-4) -- node[above] {Approximate} node[below] {well for $g\gg1$} (4.3,-4);
    
    \draw node at (0,-4) {\{the first $\ell$ spectral curves\}};
    
  \end{tikzpicture}
  \caption{Geometric picture of the hidden symmetry of the \ibQRM{\ell}}
  \label{fig:geompict}
\end{figure}
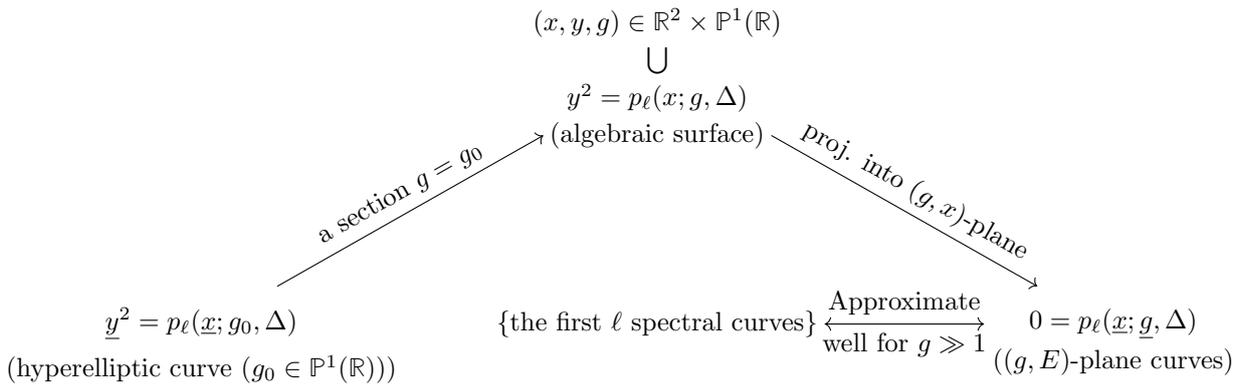

In the rest of this section we discuss the different aspects of the geometric picture described in the diagram in
Figure \ref{fig:geompict} and in \S\ref{sec:sign-eigenv-j_yell} we give a supplementary description of the method of computation
of the sign of the eigenvalues of $J_{\ell}$ used for the graphs in this section.

\subsection{Generalization of parity decomposition}
\label{sec:parity}

One of the hallmark features of the QRM ($\ell=0$ case) is the existence of a uniform (i.e. independent of the system parameters $g,\Delta$)
decomposition into invariant subspaces of the ambient Hilbert space
\[
  \mathcal{H} = \mathcal{H}_{+} \oplus \mathcal{H}_{-},
\]
by the so-called parity subspaces $\mathcal{H}_{\pm}$. The decomposition gives rise to a labeling of eigenvalues and, due the uniformity of the decomposition, to the spectral curves such that the crossings occur only between two spectral curves of different
parities.

This feature is recovered in the geometric picture described here. For the case $\ell=0$, the surface $\mathcal{S}_0$ consists of the union
of two disjoint planes $y=-1$ and $y=1$. The spectral curves of a given parity lie in one of the planes, and those of the opposite parity
in the other plane, as shown in Figure \ref{fig:3dQRM}(a). 

As mentioned in the introduction, for the case $\ell>0$ there is no uniform decomposition that holds for all values of $g>0$
allowing a concept of parity as in the QRM case. It has been suggested (e.g. in \cite{MBB2020}) that the sign of the
eigenvalues $\mu$ of $J_{\ell}$ may be used as a labeling similar to the parity in the QRM case. By Theorem \ref{thm:decomposition}
and continuity we see that indeed, spectral curves with crossings are contained in opposite half-spaces (that is, $y>0$ or $y<0$)
in the geometric picture (see for instance, Figure \ref{fig:3dQRM}(b)).

However, the situation is more complicated since the surface $\mathcal{S}_\ell$ is connected for $\ell>0$. In other words, there is the possibility that the kernel $\Jker{\ell}$ of $J_{\ell}$ is nontrivial for some $g>0$ (and therefore making the labeling ambiguous for the corresponding spectral curves) and thus is not clear at this point whether the concept of ``parity'' may be extended for the \ibQRM{\ell} with $\ell>0$ in a consistent way.
 We also note that when $\Delta\gg 0$, avoided crossings may also appear
between the spectral curves (see Figure  \ref{fig:3dQRM}(c)), this situation is a well known feature of the confluent Heun equations
that determine the spectrum of the AQRM (see e.g. \cite{SL2000} page 144-145 for the Heun case).

The conditional result of Lemma~\ref{lem:positive} shows that for $x > -\tfrac{\ell}{2} -g^2$, the surface
\[
  \mathcal{S}^*_{\ell} =  \left\{ (x,y,g) \in \mathcal{S}_\ell \, | \, x + \tfrac{\ell}{2} + g^2 >0  \right\} \subset \mathcal{S}_\ell
\]
contains two connected components. For all of the eigenvalue curves lying in $\mathcal{S}^*$, which is expected to
includes all of the eigenvalue curves having crossings, it is possible to define a labeling into two classes with the features of the parity in the QRM case.

A complete answer to the existence of an analog of the parity decomposition requires a full description of the kernel $\Jker{\ell}$ of $J_{\ell}$, which is, as we explained in \S \ref{sec:symm-decomp}, currently out of reach with the current tools, even under the assumption of Conjecture \ref{conject:main} (cf. Example \ref{ex:K0region}).

\subsection{Approximation property}
\label{sec:approximation-1}

We next discuss the approximation property of the first $\ell$ eigenvalues in the context of the geometric picture.
The projection from the $(g,x)$-plane of the geometric picture (see Figure \ref{fig:3dgx}) shows the picture of the
spectral curves in the usual way. The approximation property can be easily observed in this picture, the first $\ell$
spectral curves follow the geometry determined by the surface $\mathcal{S}$.

\begin{figure}[h]
  \centering
    \subfloat[$\ell = 1$]{
    \includegraphics[height=5cm]{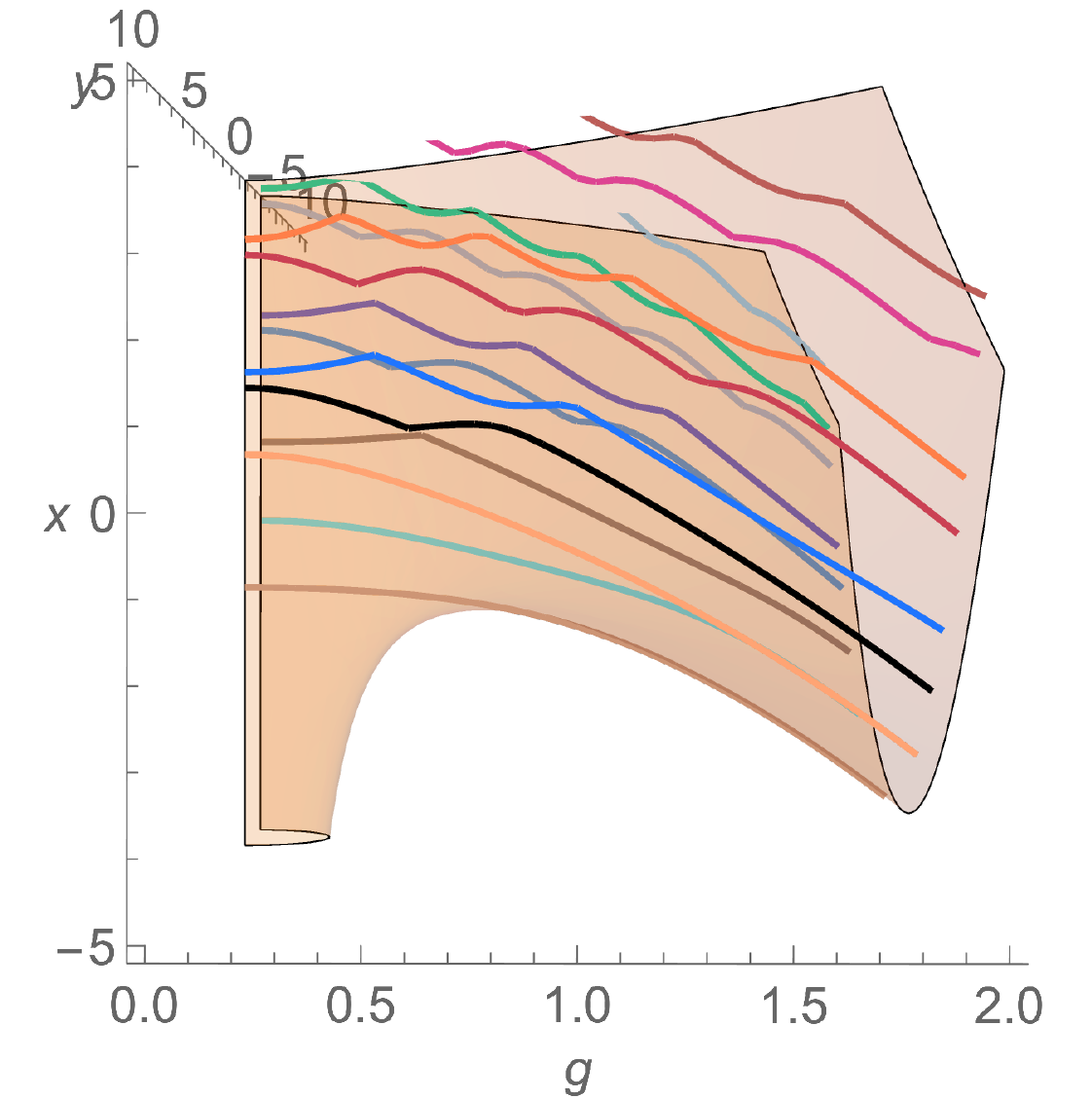}}
  ~
  \subfloat[$\ell=2$]{
    \includegraphics[height=5cm]{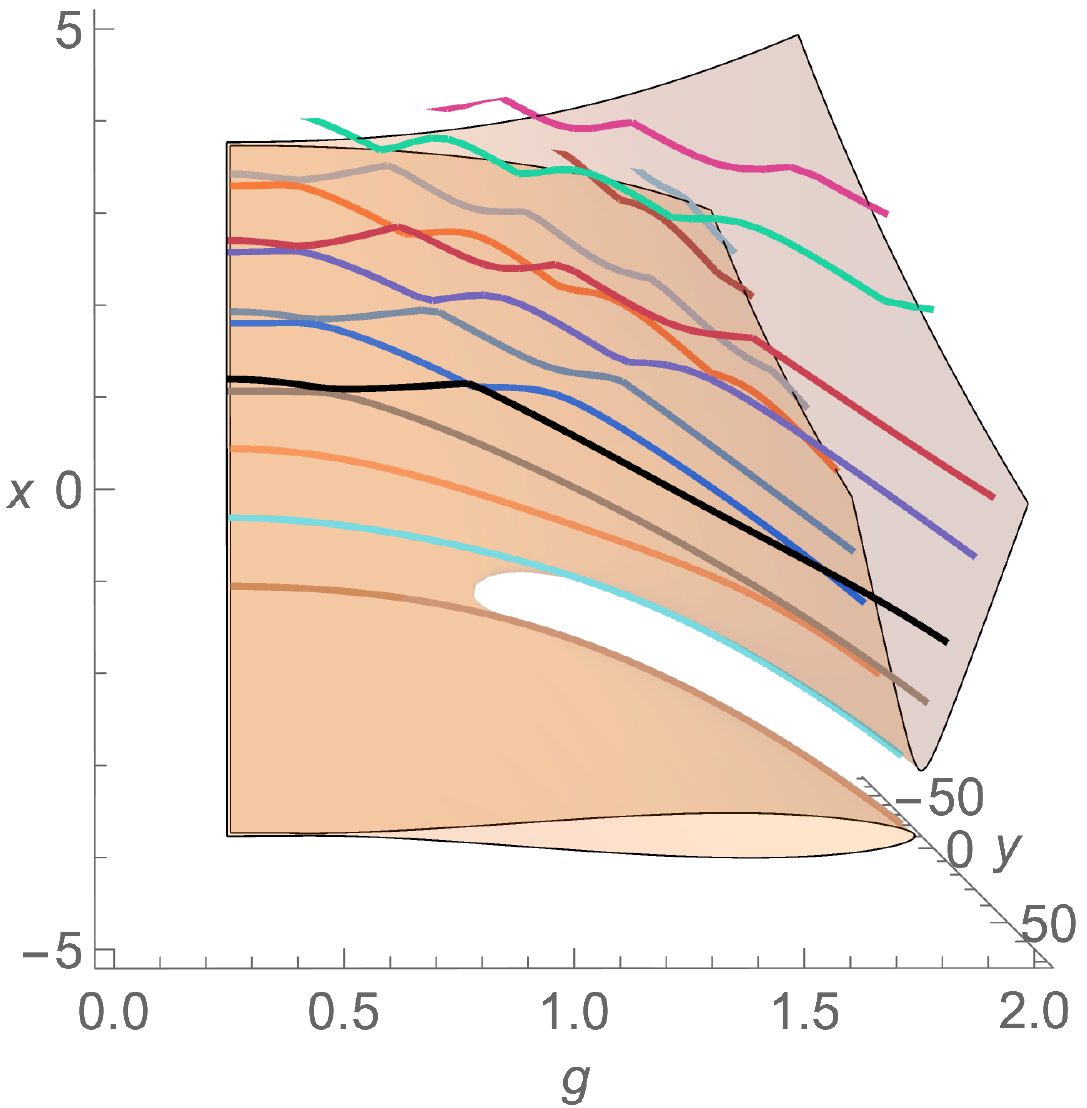}}
    ~
  \subfloat[$\ell=3$]{
    \includegraphics[height=5cm]{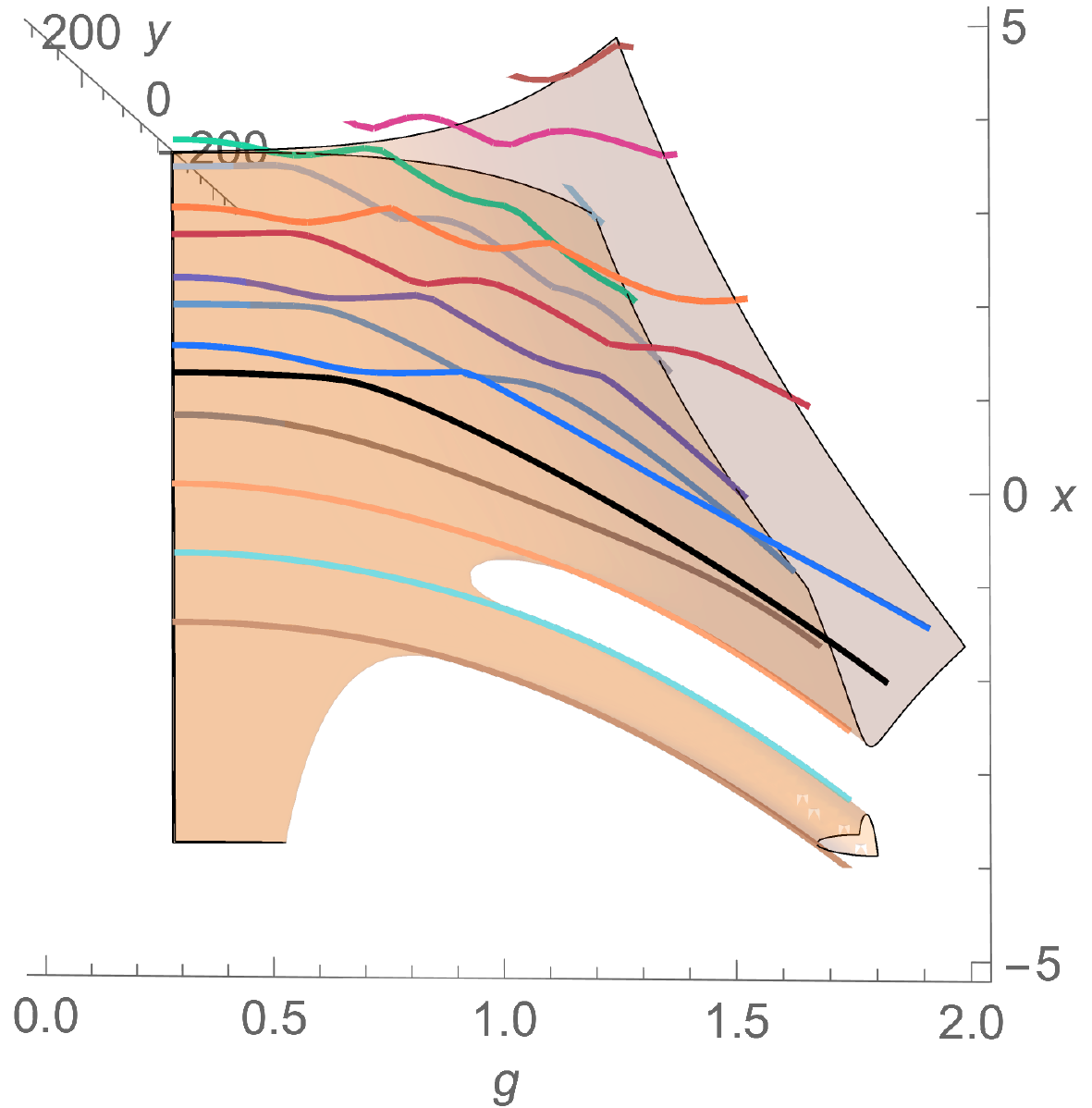}}
  \caption{View from the $(g,x)$-plane of the geometric picture for $\Delta=1$ for $\ell=1,2,3$}
  \label{fig:3dgx}
\end{figure}

The complementary projection in the $(g,y)$-plane shows how the first $\ell$ eigenvalue curves converge to the plane $y=0$ (that is,
to roots of the polynomial $p_\ell(x; g,\Delta)$ as $g$ grows larger). As we have noted before, in numerical experiments, the first
$\ell+2$ eigenvalue curves show no crossings. In the case of the QRM, the ground state and the first excited state were shown to be non-degenerate in \cite{HH2012}, however it has not been proved for general $\ell$. Note that the $\ell+1$, $\ell+2$ do not converge to the $y=0$ plane, therefore it is likely that these two curves correspond to the first two levels of the QRM case. From this point of view, the \ibQRM{\ell} appears to have a spectral structure similar to the QRM with the addition of $\ell$ levels of a different nature.

\begin{figure}[h]
  \centering
    \subfloat[$\ell = 1$]{
    \includegraphics[height=5cm]{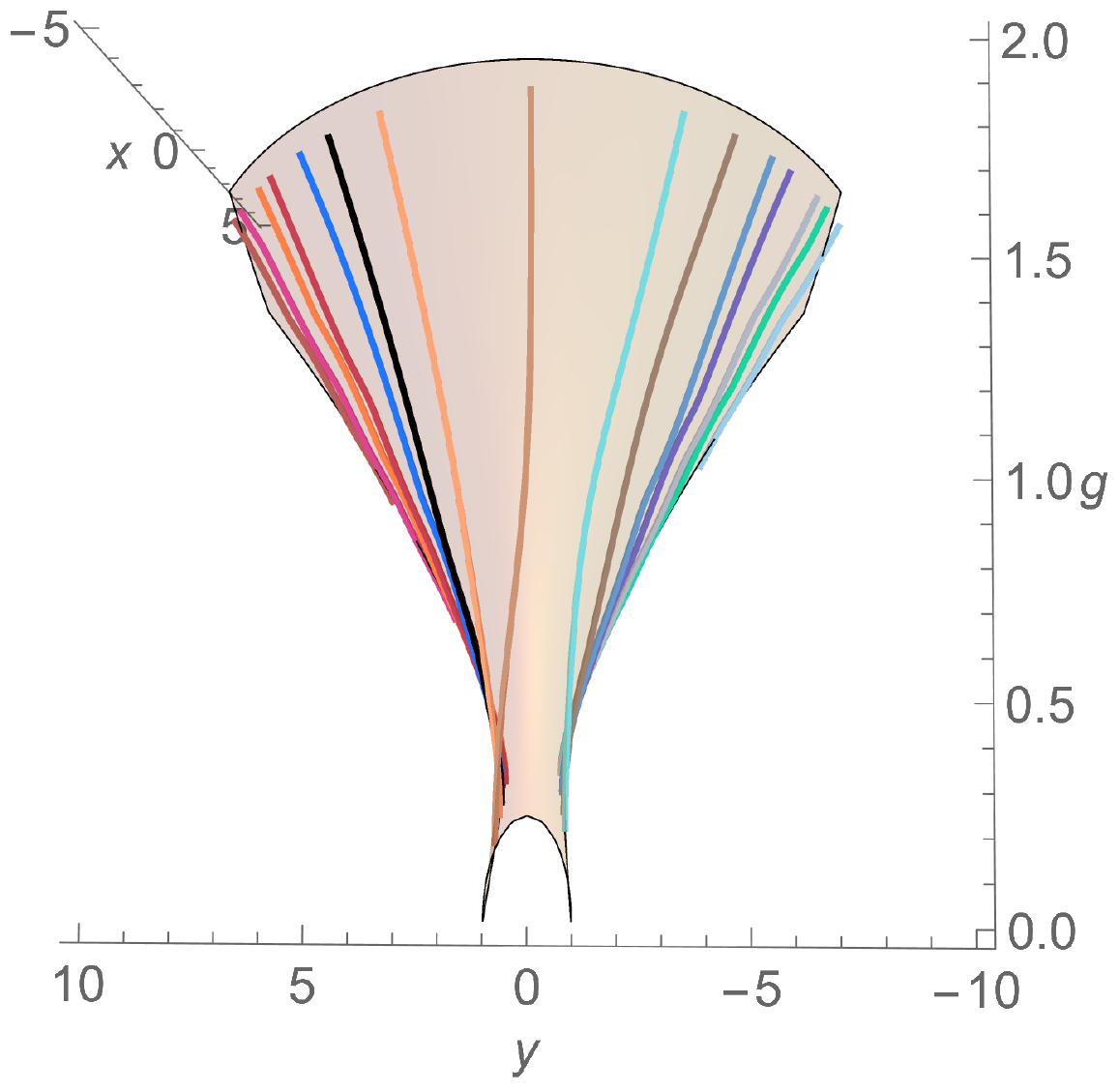}}
  ~
  \subfloat[$\ell=2$]{
    \includegraphics[height=5cm]{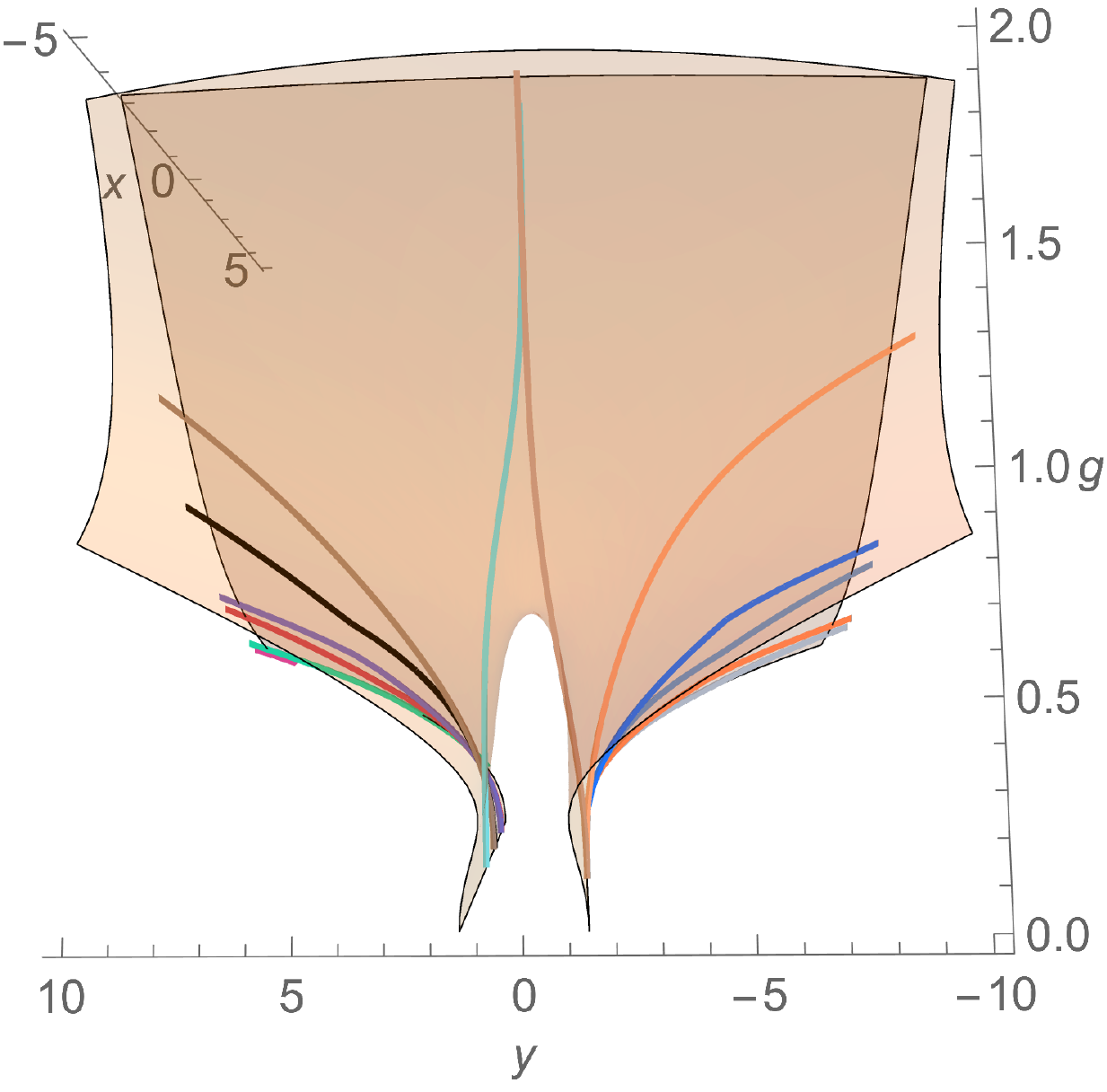}}
    ~
  \subfloat[$\ell=3$]{
    \includegraphics[height=5cm]{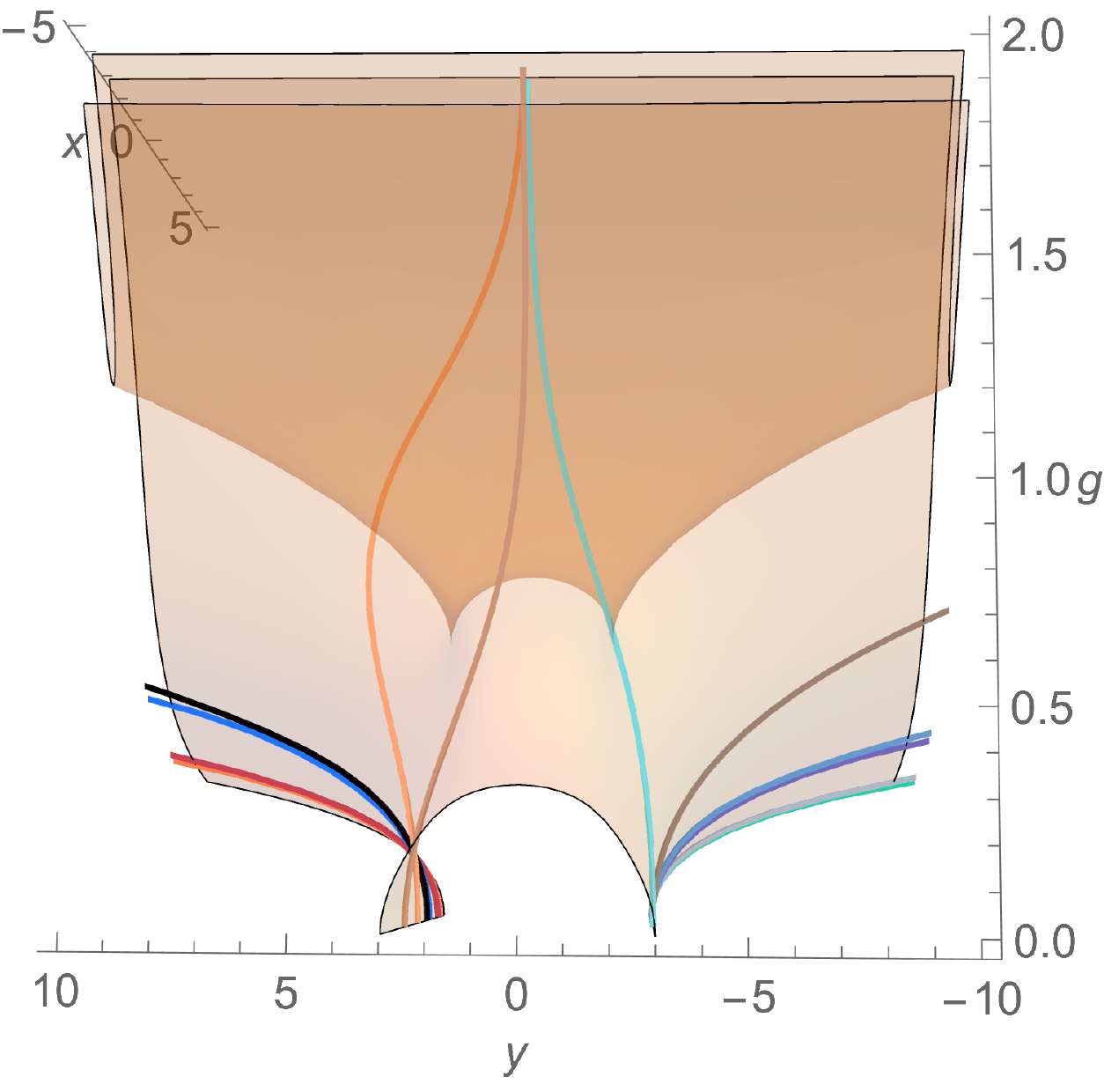}}
  \caption{View from the $(g,y)$-plane of the geometric picture for $\Delta=1$ for $\ell=1,2,3$}
  \label{fig:3dgy}
\end{figure}

In \S \ref{sec:symm-decomp} we discussed how the triviality of the kernel $\mathcal{K}_\ell$ cannot be verified only by the algebraic (or
geometric) analysis of the polynomial $p_{\ell}(x; g, \Delta)$. Similarly, the description of the first $\ell$ spectral curves, and in particular,
the approximation property require the study of both the geometric picture given by $\mathcal{S}_\ell$ and the spectrum of both  operators
$H_{\ell}$ and $J_{\ell}$. Such spectral-geometric analysis is out of the scope of this paper but it is nevertheless one of the most interesting prospects for future investigations.

  In addition, the excellent approximation given by the polynomial $p_{\ell}(x; g, \Delta)$ may allow the exploration of the spectrum from
  the point of view of Diophantine approximation as in the following claim.

  \begin{claim} Suppose $\Delta\in \Q$ and $\ell \geq 0$ are fixed. Let $E=\lambda(g,\Delta)$ be one of the first $\ell$ eigenvalues of \ibQRM{\ell}. 
    Assume that $g\in \Q$. Then, there exists a $g>0$ such that for any $\rho>2$, there are only finite many
    rational numbers $p/q\; (p, q\in \Z)$ satisfying 
    \[
      \Big|\frac{p}q -  E\Big|\leq \frac1{q^{\rho}}.
    \]
  \end{claim} 

  Indeed, let $\alpha=\alpha(g)$ be a root of the equation $p_\ell(x;g,\Delta)=0$ with respect to $x$. Since $g\in \Q$,  $\alpha$ is an algebraic number. Hence,
  the famous Roth theorem (e.g. \cite{HS2000}) asserts that for any $\delta>0$, there are only finite many rational numbers $p/q$ satisfying 
  \[
    \Big|\frac{p}q -  \alpha\Big|\leq \frac1{q^{2+2\delta}}.
  \]
By assumption, the conjectural excellent approximation of the first $\ell$ eigenvalue curves by the curve defined by $p_\ell(x;g,\Delta)=0$ in the
  $(g, x)$-plane for $g\gg 0$ holds. Thus, for any $\epsilon>0$, if $g$ is sufficiently large we have $|E-\alpha|<\epsilon$. Now we take $\epsilon$ such that
  \[
    0<\epsilon \leq \min_{p/q}\Big[\frac1{q^{2+\delta}}- \frac1{q^{2+2\delta}}\Big].
  \]
  Then, it follows that 
  \[
    \Big|\frac{p}q -  E\Big| \leq  \Big|\frac{p}q -  \alpha\Big| + |E-\alpha| \leq  \frac1{q^{2+2\delta}}+ \epsilon 
    \leq \frac1{q^{2+\delta}}.
  \]
  Hence the assertion follows.

  A complete proof of the claim requires the clarification of the nature of the excellent approximation including its proof.
  Nevertheless, the claim suggests that the study of the properties of the eigenvalues of the \ibQRM{\ell}, no limited to the first $\ell$ eigenvalues, from the point of view of Diophantine approximation may be useful to clarify the arithmetics/geometric property of the spectrum in addition to the study of spectral zeta functions discussed in \cite{RW2021} (see also \S\ref{sec:addit-remarks-degen} below).

\begin{rem} \label{rem:gaa}

  As mentioned in the introduction, the eigenvalue curves of the generalized adiabatic approximation (GAA) \cite{LB2021b}
  capture the degenerate points of the \ibQRM{\ell} correctly by construction. In fact, the energy curves of the GAA are given by
  \[
    E^{\ell}_{N,\pm} = N+\frac{\ell}{2} - g^2 \pm  \frac{(-1)^{N+\ell}(2g^2)^\ell \Delta}{2 (N!)^{\frac{3}{2}} \sqrt{(N+\ell)!}}  \exp\left(-2g^2\right) 
    \cp{N,\frac{\ell}{2}}{N}((2g)^2,\Delta^2),
  \]
  which shows that the degenerate points are given by the roots of the constraint polynomial $\cp{N,\frac{\ell}{2}}{N}((2g)^2,\Delta^2)$ as in
  the case of the \ibQRM{\ell} (see \cite{LB2021b} for the details of the derivation). Note that the energy levels $E^{\ell}_{N,+}$ and
  $E^{\ell}_{N,-}$ are symmetric with respect the baseline $E = N + \frac{\ell}{2} - g^2$ and that the approximation only considers the
  constraint polynomials $\cp{N,\frac{\ell}{2}}{N}((2g)^2,\Delta^2)$. However, the degeneracy in the \ibQRM{\ell} in the baseline $E = N + \frac{\ell}{2} - g^2$ is caused by the divisibility relation
  \eqref{eq:div} between the constraint polynomials $\cp{N,\frac{\ell}{2}}{N}((2g)^2,\Delta^2)$ and $\cp{N+\ell,-\frac{\ell}{2}}{N+\ell}((2g)^2,\Delta^2)$
  (note particularly the difference of the sign of the parameter $\pm \frac{\ell}{2}$).  This observation raises the question of whether the GAA
  may be improved in a way that the two families of constraint polynomials appear in the energy levels resulting in a better approximation.
  
\end{rem} 

\subsection{Elliptic surfaces associated to the \ibQRM{\ell} for $\ell=3$}
\label{sec:hyperelliptic}

The geometric picture described in \S\ref{sec:approximation} is naturally obtained in terms of algebraic surfaces and curves over $\R$
since the system parameters $g,\Delta>0$ and the joint eigenvalues of $H_{\ell}$ and $J_{\ell}$ are real numbers. By Lemma \ref{lem:positive}, 
under the assumption of Conjecture \ref{conject:main}, for fixed $g,\Delta>0$, the real zeros $\lambda$ of $p_{\ell}(\lambda; g, \Delta)=0$
must satisfy $\lambda \le \frac{\ell}{2}-1 -g^2$ when $g>\tfrac{\sqrt{\ell-1}}2$, however, in general, the zeros may take complex values. It is therefore important, in order to get a better understanding of the quadratic relation \eqref{eq:Jquad2}, to consider the geometric picture as algebraic surfaces over other fields in place of $\R$, in particular over $\C$ and $\Q$. 

Recall that for the cases $\ell=3,4$, for a fixed $g>0$,  $C_\ell$ is a curve of genus $1$ over an extension of $k$ of $\Q(g,\Delta)$. Under the
additional assumption that $C_{\ell}$ has a point in $k$, the curve $C_{\ell}$ is an elliptic curve. In this section, for simplicity, we consider only the case $\ell=3$.

Let $k (\supset \Q)$ be a field and denote the projection of the surface $\mathcal{S}_3$ to $k$ by $\pi$, giving
\[
  \pi \, : \, \mathcal{S}_3 \to \mathbb{P}^1(k)
\]
(the point at infinity is dealt in the usual way e.g. by the transformation $g \to \frac{1}{g}$). This gives the structure of a non-trivial elliptic surface (see e.g. \cite{S1994}  or the review paper \cite{SS2010}).

In the case of $\ell=3$, it is convenient to consider the elliptic curve by the Kodaira-Néron model of the respective elliptic curves.
Concretely, in the Kodaira-Néron model, the elliptic curve
\begin{equation} \label{eq:hyperelliptic3}
  \mathcal{E} : y^2 = p_{3}(x; g ,\Delta),
\end{equation}
is the generic fiber of the elliptic surface $\mathcal{S}_3$. The singular fibers for the case $k=\C$ can easily be described in terms of
the parameters, as we see below in Proposition \ref{prop:singularFiber}.

Let us make the detailed computations for the case $\ell=3$. We begin from the equation \eqref{eq:hyperelliptic3}
defining an elliptic curve over the field $k$. Concretely, we obtain the elliptic curve in (affine) generalized
Weierstrass form
\begin{equation}
  \label{eq:genW}
  y^2 = x^3 + a_2 x^2 + a_4 x + a^6,
\end{equation}
with
\begin{align*}
  a_2 &= 3 \left(\Delta ^2+ 2g^2(2 g^2+1)\right) \\
  a_4 &= \Delta ^2 \left(3 \Delta ^2+4\right)+48 g^8+48 g^6+4 \left(6 \Delta ^2-1\right) g^4+12 \Delta ^2 g^2 \\
  a_6 &= \Delta ^2 \left(\Delta ^2+2\right)^2+64 g^{12}+96 g^{10}+16 \left(3 \Delta ^2-1\right) g^8+24 \left(2 \Delta ^2-1\right) g^6 \\
      &\quad +12 \left(\Delta ^4+\Delta ^2\right) g^4+\left(6 \Delta ^4+8 \Delta ^2\right) g^2.
\end{align*}
By routine transformations, we obtain the Weierstrass equation in standard form
\[
  y^2 = x^3 - 27 c_4 x - 54 c_6,
\]
with
\begin{align*}
  c_4 &= 192 \left(4 g^4-\Delta ^2\right) \\
  c_6 &= -3456 \Delta^2,
\end{align*}

Next, the change of variable
\[
  A = 8 (2g^2), \qquad B = 192 \Delta^2,
\]
gives the elliptic curve

\begin{equation}
  \label{eq:elliptic}
  \mathcal{E} : \, y^2 = x^3 - 27(3 A^2-B) x + 54(18 B).
\end{equation}
defined over $\Q$ with discriminant
\[
   {\rm Disc}(\mathcal{E})=  2^6 3^9 ( (3 A^2-B)^3 - 324 B^2).
\]      

The discriminant shows that there are singular fibers in the elliptic surface. It is easy to verify then
the following result from Tate's algorithm and Kodaira classification of singular fibers.

\begin{prop} \label{prop:singularFiber}
  Let $k=\C$. For a fixed $\Delta>0$, the singular fibers of the elliptic surface $S_3$ are all of type $I_1$ (nodal type singular fibers).
  The elliptic surface $\mathcal{S}_\ell$ has at most 6 (resp. 12) singular fibers with respect to the variable $A$ (resp. $g$).
\end{prop}

\begin{proof}
  The second statement follows from the fact that for a fixed $\Delta>0$, the equation \({\rm Disc}(\mathcal{E}) = 0 \) has exactly
  twelve solutions in $g \in \C$.
  Next, we assume that $g$ is chosen such that the discriminant vanishes. In particular, we note that the order of vanishing is
  exactly $1$ (and thus singular fibers of multiplicative type have exactly one component).
  Since the characteristic of the field is $0$, it is enough to verify the vanishing of the coefficients of \eqref{eq:elliptic}
  with respect to the chosen $g$ (see \cite{SS2010} for details.). In this case, the coefficients $c_4$ and $c_6$ do not non-vanish,
  and therefore the singular fibers have one connected component and are of nodal type, that is, of type $I_1$ in the Kodaira
  classification. 
\end{proof}

The parameters $g$ that result in singular fibers are in general complex numbers. It may be interesting to investigate the spectral features of the \ibQRM{3} for the cases of $g \in \R_{\geq0}$ resulting in singular fibers.

\begin{ex}
  Let us consider an example of parameters that result on a singular curve $\mathcal{E}(g,\Delta)$. If we fix $\Delta=\tfrac{1}{4}$,
  the parameter $g=\tfrac{1}{2}$ gives the singular curve of nodal type
  \[
    y^2 = (x+\tfrac{29}{16})(x+\tfrac{5}{16})^2
  \]
  shown in Figure \ref{fig:nodal}.

  \begin{figure}[h]
    \centering
    \includegraphics[height=5cm]{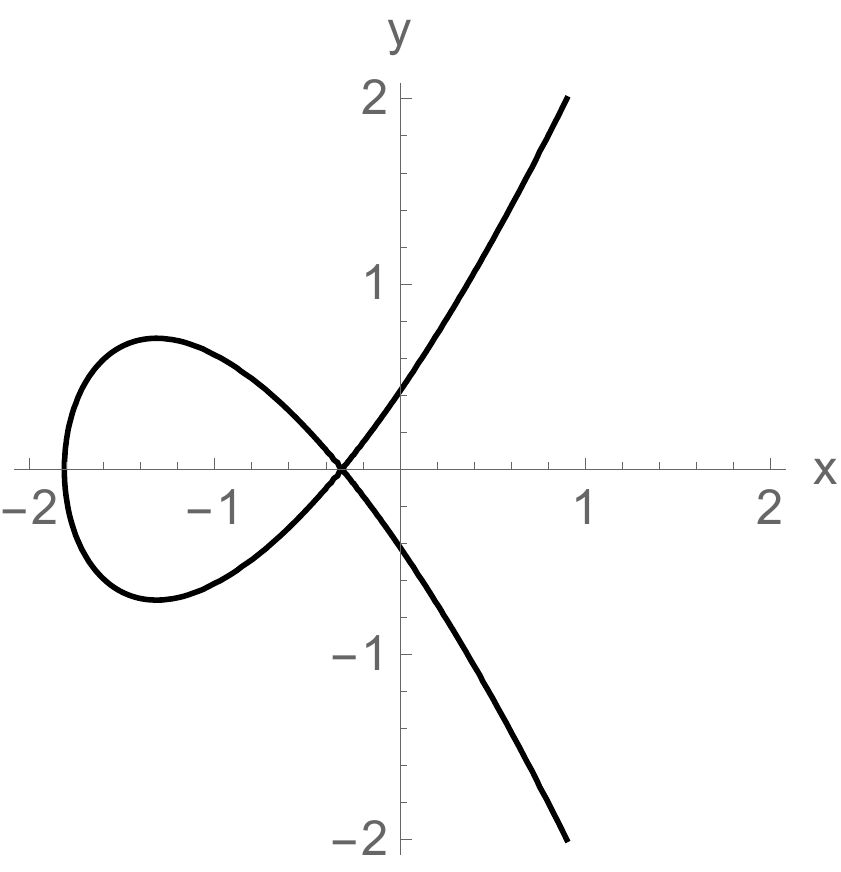}
    \caption{Singular fiber of nodal type for $g=\tfrac{1}{2}$ and $\Delta=\tfrac{1}{4}$.}
    \label{fig:nodal}
  \end{figure}
\end{ex}

It is an elementary result of the theory elliptic surfaces that the set of sections $\sigma : \mathbb{P}^1(k) \to \mathcal{S}_3$ form a group
isomorphic to the points of the generic fiber \eqref{eq:hyperelliptic3} considered as an elliptic curve over the function field $k(g)$
(see e.g. \S 3.4 of \cite{SS2010} or Proposition 3.10 of \cite{S1994}) thus allowing an interplay between the geometric and algebraic
aspects of the elliptic surface. In particular, we believe that sections $\sigma$ corresponding to parameters $g>0$ containing Juddian
points in $\mathcal{D}_\ell(g,\Delta)$ (i.e. degenerate points) are of particular importance for the understanding of Conjecture
\ref{conject:main}. For instance, the spectral points (elements from $\mathcal{D}_\ell(g,\Delta)$) in the corresponding elliptic curves may be
interpreted in terms of the group operation (see \S \ref{sec:dens-judd-solut}). A detailed study of the algebro-geometric aspects of
the (hyper)elliptic surfaces presented here is left to another occasion.

\begin{ex}

  We may consider the base field to be $\Q$ and study the arithmetic properties of the elliptic surfaces and curves.
  By considering special form of the parameters $A$, $B$ we can obtain a family of curves (depending of a single
  parameter) with fixed rational torsion. For this purpose, it is convenient to work in the generalized Weierstrass form
  given by \eqref{eq:genW}.

  Let us define the elliptic curve $\mathcal{E}(T)$ by the change of variable
  \begin{equation}
    \label{eq:changevarE}
       A = 2(T+6), \qquad B = 9 T (T + 12),
  \end{equation}
  in \eqref{eq:genW} with $T \in \Z$. It is immediate to verify that for $T \in \{12,-6,-24\}$, the curve $\mathcal{E}(T)$
  is singular. 

  Let $T \not\in \{12,-6,-24\}$, then after an admissible change of variable the elliptic curve $\mathcal{E}(T)$ has the form
  \[
    \mathcal{E}(T) : \, y^2 = (x+ 7 T (T+12)+240) (x+ T (7 T+100)+48) (x+ T (7 T+116)+432),
  \] 
  and thus, the order two elements are given (in affine form) by 
  \[
    P_1 =(-7 T (T+12)-240,0),\qquad  P_2 = (-T (7 T+100)-48,0), \qquad P_3 = (-T (7 T+116)-432,0),
  \]
  which, along with the point at infinity $\bm{O}$ form a subgroup $\Z_2 \oplus \Z_2$ of $\mathcal{E}(T)(\Q)_{\text{tor}}$.

  It is easy to verify that $\mathcal{E}(T)(\Q)_{\text{tor}} = \Z_2 \oplus \Z_2$ for the cases $T \not\equiv 0 \pmod{3}$. Indeed, since the
  discriminant is given by
  \[
    \rm{disc}(\mathcal{E}(T)) = 2^{26} (T-12)^2 (T+6)^2 (T+24)^2,
  \]
  for $T \not\equiv 0 \pmod{3}$ the curve $\mathcal{E}(T)$ has good reduction at $p=3$ resulting in the curve
  \[
    y^2 = x^3 + 2 T^2 x + 2T^4+T^6,
  \]
  defined in $\Z_3$. For $T \equiv \pm 1 \pmod{3}$ we have
  \[
    y^2 = x^3 - x,
  \]
  and we verify directly that $\# \mathcal{E}(T)(\Z_3) = 4$. By the Lutz-Nagell theorem (see e.g. Chapter 5 of \cite{K1992}),
  we have $\mathcal{E}(T)(\Q)_{\text{tor}} = \Z_2 \oplus \Z_2$.

  Similarly, if $T \equiv 0 \pmod{15}$ and $T \equiv 3 \pmod{15}$, we verify that $\rm{disc}(\mathcal{E}(T))$ has
  good reduction at $p=5$ and the resulting elliptic curves are given by
  \[
    y^2 = x^3 + x,
  \]
  and
  \[
    y^2 = x^3 - x^2 + 3 x,
  \]
  in $\Z_5$, respectively. In both cases we immediately verify that $\# \mathcal{E}(T)(\Z_5) = 4$.
  Other cases are dealt in a similar way.
   
  We also note that the change of variable \eqref{eq:changevarE} is not the only one with these properties.
  For instance, let $M,N \in \Z$ and suppose
  \[
    C(M,N) = 4 M^2 + 3 M N (4 + M) + 9 N^2 (4 + M),
  \]
  is a perfect square. Then, setting
  \[
    A' = \frac{\sqrt{C(M,N)}}{2}, \qquad B' = \frac{9}{4} N M (M + 3 N)
  \]
  with $N,M$ chosen so that $B \in \Z$, the elliptic curve $\mathcal{E}(M,N)$ with parameters $A,B$ is an elliptic curve
  (when it is non-singular) with rational torsion group
  \[
    \Z / 2\Z \oplus \Z / 2\Z.
  \]

  The family of curves $\mathcal{E}(T)$ above are obtained by taking
  \[
    M=T, \qquad N = 4.
  \]
  Note that $C(T,4)=16(T+6)^2)$ is a perfect square for any $T$.  We leave the detailed discussion of the arithmetics of the (hyper)elliptic
  curves appearing in the study of symmetry of the \ibQRM{\ell} for another occasion.
  
\end{ex}

\begin{rem}\label{rem:eta-invarinat}
In the discussion above, we have focused on the existence of degenerate points. In addition to this, we are also interested in the distribution of the sign of the spectrum of $J_\ell$ (see Figure \ref{fig:inverse}). It can be expected that those sign are almost equally distributed. Therefore, in particular, it is worth studying the naive notion of ``the number of positive eigenvalues minus the number of negative eigenvalues" of $J_\ell$ for studying a possible spectral asymmetry. In order to clarify the situation, we propose to introduce an analogue of the eta invariant that is defined for a self-adjoint elliptic differential operator on a compact manifold, initially introduced by Atiyah, Patodi and Singer (see \cite{Mu1995}). Precisely, the analogue of the eta invariant $\eta_\ell(0)=\eta_\ell(0; g,\Delta)$ for the \ibQRM{\ell}
may be defined using zeta regularization via the analytic continuation of the following Dirichlet series.
$$
\eta_\ell(s; g, \Delta)= \sum_{\lambda\not\in\Jker{\ell}}\frac{\sign(\mu_\lambda)}{|\mu_\lambda|^{s}}, \quad (\Re(s) \gg 0).
$$ 
Here the sum is over the eigenvalues $\lambda$ of $\HRabi{\ell}$ such that $\mu_\lambda^2=p_\ell(\lambda; g,\Delta)\ne0$. Notice that there is no contribution of the degenerate eigenvalues in this series. The analytic continuation can be discussed, for instance, by considering the following integral expression (Mellin transform) 
$$
\eta_\ell(s; g, \Delta)= \frac1{\Gamma\big(\frac{s+1}2\big)}\int_0^\infty t^{(s-1)/2} \Tr(J_\ell e^{-tJ_{\ell}^2})dt, \quad (\Re(s) \gg 0).
$$
If we take this way, we need to study the asymptotic expansion of $\Tr(J_\ell e^{-tJ_{\ell}^2})$ for 
$t\downarrow 0$ so we should explore a technique beyond the standard heat equation method (see e.g. \cite{D1996}) or the application of the heat kernel for studying spectral determinants of the AQRM (see \cite{RW2019, RW2021, R2020}).
In addition, we may expect to explore an alternative way to determine the distribution of the sign of $J_\ell$ using the continuity with
respect to $g$ as we discuss in \S\ref{sec:sign-eigenv-j_yell}.
\end{rem}

\subsection{Divisibility and degeneracy}
\label{sec:addit-remarks-degen}

As discussed in the previous subsections, the hidden symmetry of the \ibQRM{\ell} induces a geometric picture for the spectrum that may be interpreted as a resolution of singularities for the spectral curves at the degenerate (Juddian) points. Moreover, the main conjecture describes the relation between the symmetry and the degeneracy via the polynomial $p_{\ell}(x; g,\Delta)$, essentially determined by the divisibility relation  \eqref{eq:div} of constraint polynomials. 

The relation between divisibility of polynomials, degeneracy (intersection of curves on varieties) and resolution of singularities is not
unique to the \ibQRM{\ell} and similar situations have appeared in other contexts. For instance, in \cite{CZ2010}
(see also \cite{RTW2021} and the references therein) the degenerate integral points of curves $\mathcal{C}_i$ ($i=1,2$) lying
on a smooth projective surface $\mathcal{Y}$ are shown to be related to certain divisibility conditions of polynomials. Concretely, for
the blow-up $\pi : \mathcal{X} \to \mathcal{Y}$, along the intersection points $P_j (j=1,2,\ldots,M)$ of the two curves $\mathcal{C}_1$ and $\mathcal{C}_2$, the condition ``$\pi^{-1}(Q)$ is integral on $\mathcal{X}$ along the strict transform of $\mathcal{C}_1$'' for $P_j \neq Q\in \mathcal{Y}$ is related to a divisibility condition of polynomials defining the curves $\mathcal{C}_i$ ($i=1,2$) locally at the intersection points. 

Let us now return to the case of the \ibQRM{\ell}, more generally the AQRM, where we actually find a similar
    situation.
  Concretely, for fixed $\Delta>0$ and $N,\ell\ge 1$, we consider the curves $\mathcal{C}_1$ and $\mathcal{C}_2$ in $\mathbb{P}^2$ given by
  \begin{align*}
    \mathcal{C}_1 &= \{ (\epsilon,g)\in \mathbb{P}^2  \, : \cp{N,\frac{\epsilon}{2}}{N}((2g)^2,\Delta^2) = 0  \} \\
    \mathcal{C}_2 &= \{ (\epsilon,g)\in \mathbb{P}^2  \, : \cp{N+\ell,-\frac{\epsilon}{2}}{N+\ell}((2g)^2,\Delta^2) = 0  \}.
  \end{align*}
  In Figure \ref{fig:integrability} we show an example of the curves $\mathcal{C}_i$ ($i=1,2$) (cf. Example \ref{ex:cPoly}) in the
  $(\epsilon,g)$-plane for  $N=2, N+\ell=4$ and $\Delta=1$.
  \begin{figure}[h]
    \centering
    \includegraphics[height=5cm]{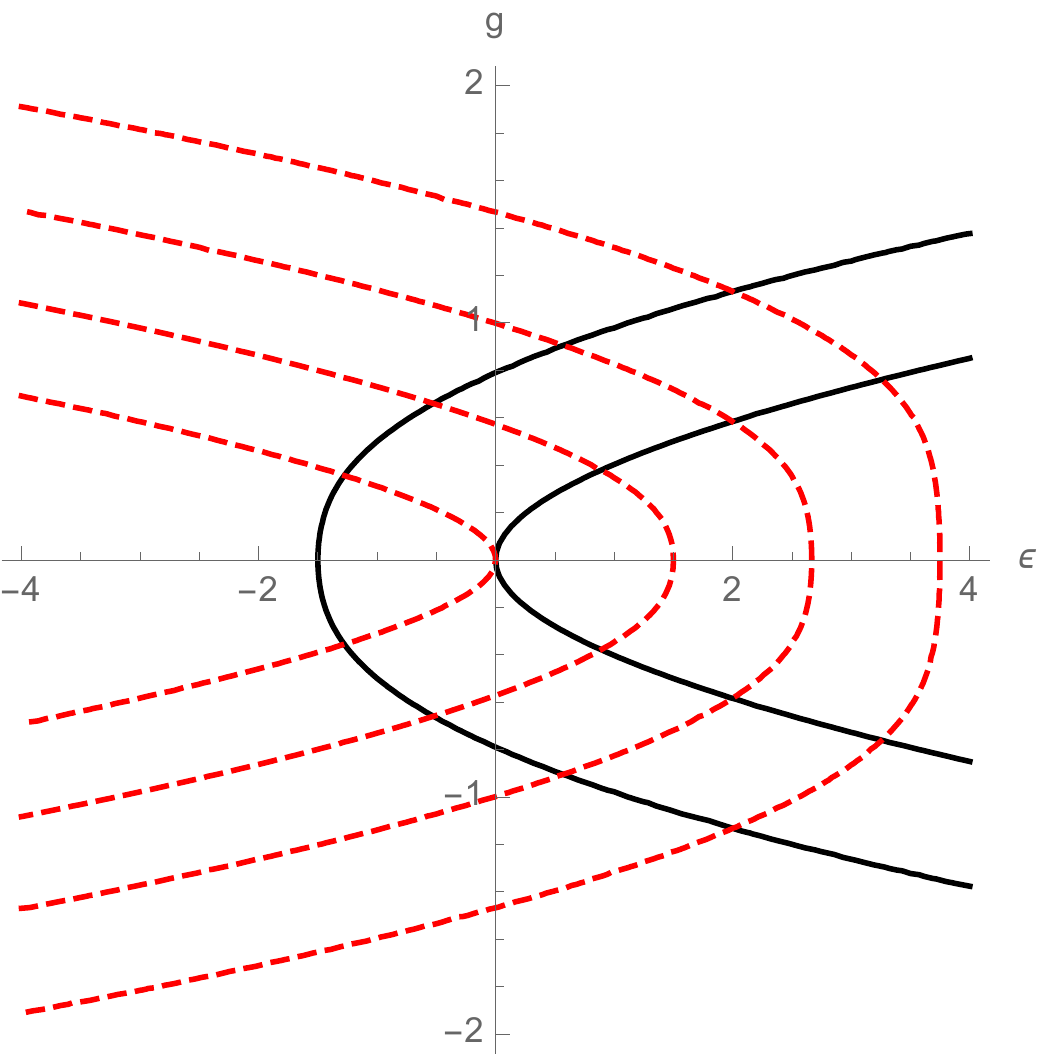}
    \caption{Curves $\mathcal{C}_1$ (black) and $\mathcal{C}_2$ (red) for $N=2, \ell=2$ and $\Delta=1$.}
    \label{fig:integrability}
  \end{figure}
  Denote the intersection points by $P_j \in \mathbb{P}^2$ ($j=1,2,\cdots,M$) and let $\pi: X \to \mathbb{P}^2$ be the blow-up along the
  intersection points $P_j$, that is
  \[
    X = \Bigg\{ (Q,L_j) \in \mathbb{P^2}\times \bigsqcup_{j} {\rm Gr}_{P_j}(1,2) : P_j,Q \in L_j  \Bigg\},
  \]
  where ${\rm Gr}_{P_j}(1,2)$ is the Grassmannian, the set of lines through the point $P_j$. Consider a point $Q \in \mathbb{P}^2$ with
  $Q \neq P_j$ ($j=1,2,\cdots,M$), and the inverse  $\pi^{-1}(Q) = ((\epsilon,g_0),L)$ along the strict transform of  $\mathcal{C}_1$ in $X$ (see e.g. \cite{H1997}).  We note that the
  results of \cite{KRW2017} on divisibility may be described as the equivalence 
  \[
    \epsilon = \ell \in \Z_{\ge 0}, g_0 \neq 0  \iff   \cp{N,\frac{\epsilon}{2}}{N}((2g_0)^2,\Delta^2) \text{ divides } 
    \cp{N+\ell,-\frac{\epsilon}{2}}{N+\ell}((2g_0)^2,\Delta^2)
  \]
  involving a divisibility condition for degeneracy and an integrability condition on the parameters.
 
  The framework described in \cite{CZ2010} may be helpful to obtain a further understanding of the degeneracy in the spectrum of the
  \ibQRM{\ell} and to provide a proof of the main conjecture, the excellent approximation of the first $\ell$ eigenvalue curves by the curves
  $p_\ell(x;g,\Delta)=0$ (for instance, by relating to
  known Diophantine approximation results) and the estimation of the density of the Juddian points (see Remark
  \ref{rem:densityCZ} in \S\ref{sec:dens-judd-solut}). In fact, since divisibility conditions are also known to be related to famous
  mathematical problems like the Vojta conjectures, we might expect such a study to give deep insights into the relation of the
  \ibQRM{\ell} spectrum and certain problems in arithmetic geometry.

  As a final remark for this section, we note that the point of view of considering the bias parameter $\epsilon$ as a parameter is also
  considered in \cite{BLZ2015} to study the nature of the degeneracies (i.e. conical intersections) in the so-called energy landscape
  of the AQRM.

  
\subsection{Supplementary discussion on the sign of the eigenvalues of $J_\ell$} 
\label{sec:sign-eigenv-j_yell}

Let $\mu$ be an eigenvalue of $J_\ell$. 
As we have noted before, the equation \eqref{eq:Jquad} determines the absolute value of $\mu$ as
\[
  |\mu| = \sqrt{p_{\ell}(\lambda; g ,\Delta)}
\]
for some eigenvalue $\lambda \in \Spec(\HRabi{\ell})$. In general, the computation of the sign of the eigenvalue $\mu$ must be done directly
from the action of $J_{\ell}$ on the corresponding eigenfunction. Since there is no general formula for $J_{\ell}$ (even a conjectural one),
the determination of the sign is in general a difficult problem (see Remark \ref{rem:eta-invarinat}). However, exploiting the continuity of the $J_{\ell}$ with respect to $g$,
the sign may be computed at $g=0$ and then extended along each of the spectral curves. Since the computation of the sign was used
for the 3D graphs on this section, we give a short description of the procedure. 

Let us consider the case $g=0$, then the Hamiltonian of the \ibQRM{\ell} is given by
\[
  H = a^\dag a+\Delta \sigma_z  + \frac{\ell}{2} \sigma_x,
\]
that is, it reduces to the Hamiltonian of a displaced quantum Harmonic oscillator. The Hamiltonian $H$ is considerably simpler than the \ibQRM{\ell} (and the AQRM) Hamiltonian. In particular, it commutes with $\mathcal{P}$ and with the matrices
\[
  \begin{bmatrix}
    \ell & \Delta \\
    \Delta & 0
  \end{bmatrix}, \qquad
  \begin{bmatrix}
    0 & \Delta \\
    \Delta & -\ell
  \end{bmatrix}.
\]
Consequently, the matrices in $\mathcal{P} \rm{Mat}_2(\C)$ commuting with $H$ form a continuous subgroup of the commutant. In order to
compute the sign of the eigenvalue of $J_{\ell}$ we need to compute $J_{\ell}$ when $g = 0$, which generates a discrete subgroup
of the commutant of $H$.

We now consider the diagonalization of $H$. By setting
\[
  C_\ell =
  \frac{1}{(2\Delta)^{\frac{1}{2}} (\Delta^2+\tfrac{\ell^2}{4})^{\frac{1}{4}}}
  \begin{bmatrix}
    \tfrac{\ell}{2} + \sqrt{\Delta^2 + \tfrac{\ell^2}{4}} & \Delta \\
    \tfrac{\ell}{2} - \sqrt{\Delta^2 + \tfrac{\ell^2}{4}} & \Delta 
  \end{bmatrix},
\]
we see that
\[
  \tilde{H} := C H C^{-1} = 
  \begin{bmatrix}
    a^\dag a + \sqrt{\Delta^2 + \tfrac{\ell^2}{4}} & 0 \\
    0 & a^\dag a -\sqrt{\Delta^2 + \tfrac{\ell^2}{4}}
  \end{bmatrix}.
\]
The spectrum of $\tilde{H}$  is
\[
  \Spec(\tilde{H}) =  \left\{ n \pm \sqrt{\Delta^2 + \tfrac{\ell^2}{4}} \, : \, n \ge 0 \right\},
\]
with  corresponding eigenfunctions in the $L^2(\R)$ realization, given by
\[
  \psi_{n,+}(x) =
  \begin{bmatrix}
    H_n(x) \\
    0
  \end{bmatrix},
  \qquad
  \psi_{n,-}(x) =
  \begin{bmatrix}
    0 \\
    H_n(x)
  \end{bmatrix},
\]
where $H_n(x)$ is the $n$-th Hermite function.

For the cases with explicitly computed $J_{\ell}$ (see Appendix \ref{sec:expl-expr-j_ell} and \cite{MBB2020, RBW2021}) we verify that
for $g=0$, we have
\begin{equation}
  \label{eq:Jg0}
  C_\ell J_{\ell} C_\ell^{-1} =
  \begin{cases}
   \mathcal{P}
  \begin{bmatrix}
    q_\ell(\Delta) & 0 \\
    0 & - q_\ell(\Delta)
  \end{bmatrix} & \text{ if } \ell \equiv 0 \pmod{2} \\
   & \\
   \mathcal{P}
  \begin{bmatrix}
    q_\ell(\Delta) & 0 \\
    0 &  \phantom{-}q_\ell(\Delta)
  \end{bmatrix} & \text{ if } \ell \equiv 1 \pmod{2}
  \end{cases}
\end{equation}
where $q_{\ell}(\Delta)$ is given by the non-negative root $\sqrt{p_{\ell}(x; 0,\Delta)}$ in both cases.

By the parity of the Hermite functions and the form of the matrices for each $\ell$, the eigenvalues of $J_{\ell}$ corresponding to
$\psi_{n,\pm}$ are then immediately seen to be of the form

\[
  \mu_{n,\pm} =
  \begin{cases}
    \pm (-1)^n q_{\ell}(\Delta) & \text{ if } \ell \equiv 0 \pmod{2} \\
     (-1)^n q_{\ell}(\Delta) & \text{ if }  \ell \equiv 1 \pmod{2}
  \end{cases}.
\]

Finally, the sign of the eigenvalue $\mu_\lambda$ is then extended to general $g>0$ along each spectral curve by continuity.
 
The procedure outlined here can be used for any $\ell\in \Z_{\ge0}$, and the shape of the matrices given in \eqref{eq:Jg0} is
expected to hold for all $\ell \geq0$, but we leave the proof for another occasion.
  

\section{Remarks on the dense distribution of Juddian eigenvalues}
\label{sec:dens-judd-solut}

In \S\ref{sec:preliminary-results} we discussed the action of the operator $J_{\ell}$ on the eigenspaces $V_{\lambda}$ for a
degenerate eigenvalue $\lambda = N \pm \frac{\ell}{2} -g^2$ for fixed parameters $g,\Delta>0$. As we discussed in
\S \ref{sec:approximation}, the set $\mathcal{D}_{\ell}(g,\Delta)$ of tuples of joint eigenvalues $(\lambda,\mu_\lambda)$ of $H_{\ell}$ and $J_{\ell}$
lie in the curves
\[
  y^2 = p_{\ell}(x; g, \Delta).
\]

This picture then allows a purely geometric description of the degenerate eigenvalues of $H_{\ell}$. Namely, the spectrum of
$H_{\ell}$ has degenerate eigenvalues for the parameters $g,\Delta>0$ if and only if there are $(\lambda,\mu_\lambda) \neq (\lambda',\mu_{\lambda'}) \in \mathcal{D}_\ell(g,\Delta)$
such that $\mu_\lambda = - \mu_{\lambda'}$. In the case of $\ell=3$, this condition is equivalent to the existence of points $P \in \mathcal{D}_{\ell}(g,\Delta)$
with $-P \in \mathcal{D}_{\ell}(g,\Delta)$ and $P \neq -P$ in the usual group operation of the elliptic curve.
We illustrate the situation in Figure \ref{fig:inverse}, where the two points of $\mathcal{D}_\ell(g,\Delta)$ corresponding to degenerate points
lie in the same vertical line (shown in red). We note that in general it is difficult to expect non-trivial relations with respect to the
group operator involving arbitrary points of $\mathcal{D}_\ell(g,\Delta)$.
  
\begin{figure}[h]
  \centering
  \includegraphics[height=5cm]{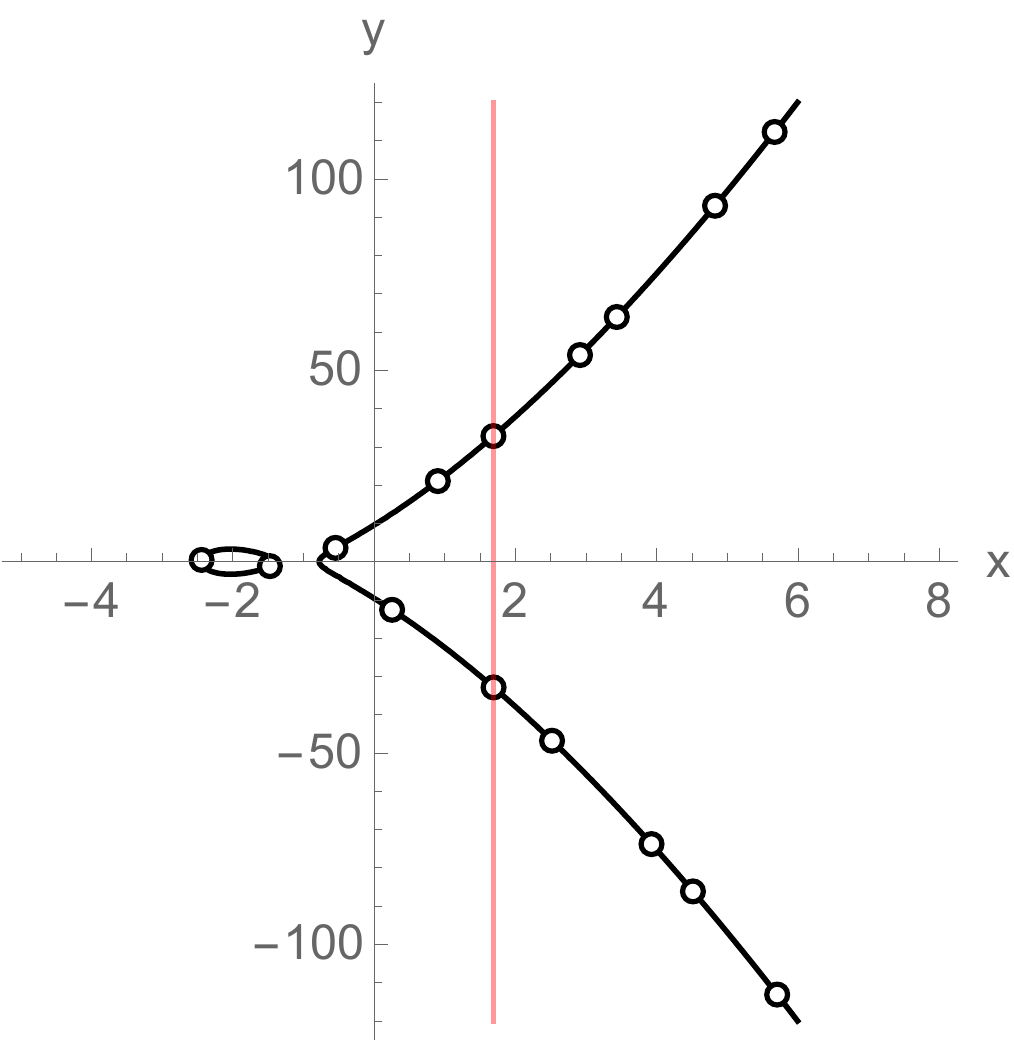}
  \caption{Points of $\mathcal{D}_3(g,\Delta)$ in the elliptic curve given by $y^2 = p_{3}(x; g ,\Delta)$ for $\Delta=\tfrac{3}{7}$ and $g \approx 0.899$.}
  \label{fig:inverse}
\end{figure}

This discussion raises the question of whether for all parameters $g,\Delta>0$ such degenerate solutions exist (for some $N\ge 0$), or if
there parameters such that the spectrum of $\HRabi{\ell}$ is multiplicity free.

There are in fact parameters that make the \ibQRM{\ell} spectrum multiplicity free. Concrete examples may be given, for instance,
if $\Delta \in \Q$, then by setting $(2g)^2 = \pi$ (or any other transcendental number), we see that $(2g)^2$ cannot be the root of $P_N^{(N, \frac{\ell}{2})}((2g)^2,\Delta^2)$ for any $N \geq 0$. Therefore, for these parameters $g,\Delta>0$, the spectrum of $\HRabi{\varepsilon}$ is non-degenerate.

It should be noted that, however, it is difficult to compute numerically examples of parameters that give multiplicity free spectrum.
Let us define the set
\[
  \Omega_{N}^{(\ell)} = \{ (g,\Delta)\in \R^2 \,|\, P_N^{(N, \frac{\ell}{2})}((2g)^2,\Delta^2) = 0  \}.
\]
The set $ \Omega_{N}^{(\ell)}$ is an algebraic curve defined by $P_N^{(N, \frac{\ell}{2})}((2g)^2,\Delta^2) = 0$ that is the union of $N$ oval-shaped
curves, as shown in Figure \ref{fig:ellipsesP}.

\begin{figure}[h]
  \centering
  \subfloat[$N=1$]{
    \includegraphics[height=4cm]{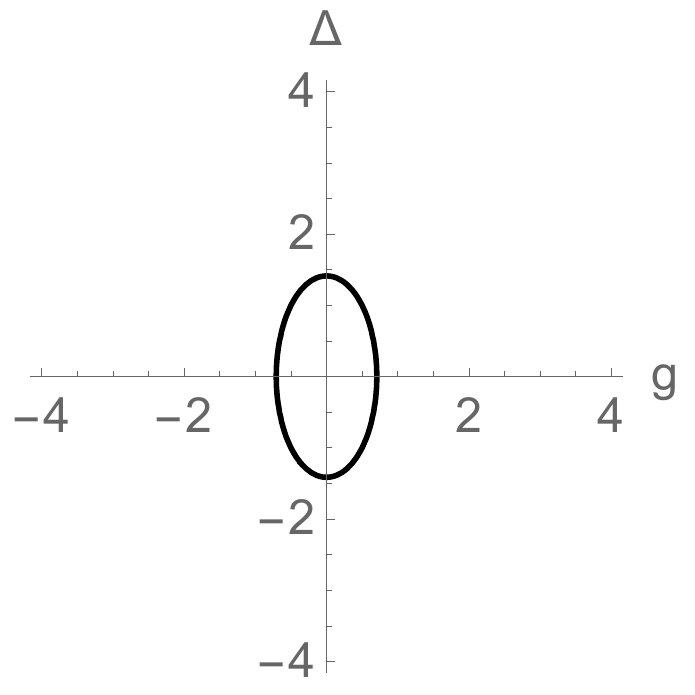}}
  ~
  \subfloat[$N=2$]{
    \includegraphics[height=4cm]{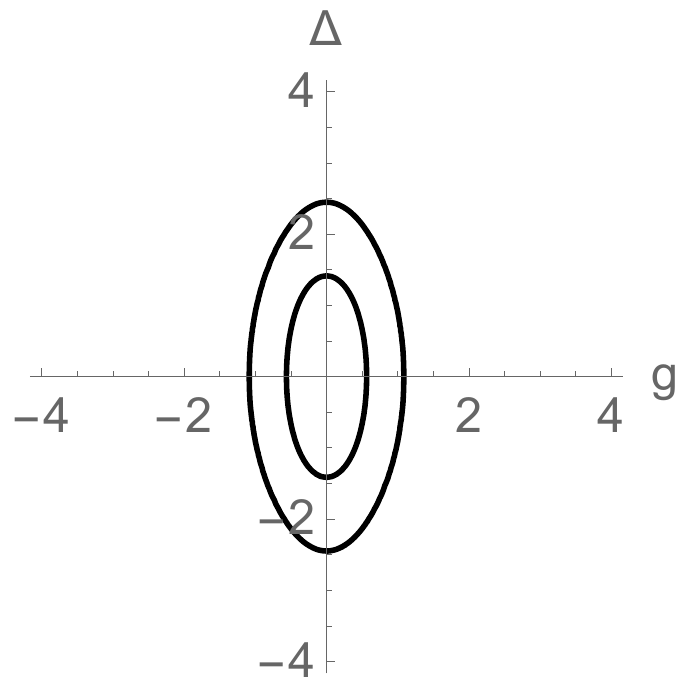}}
  ~
  \subfloat[$N=3$]{
    \includegraphics[height=4cm]{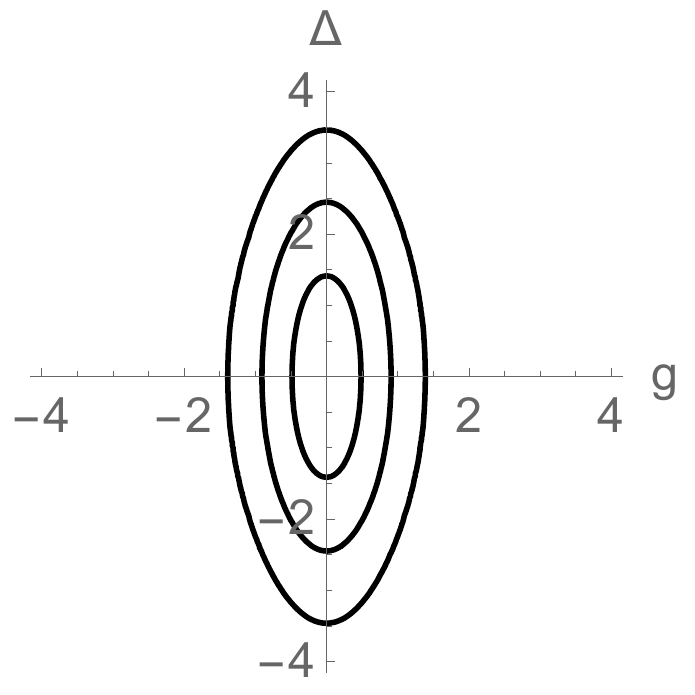}}
  \caption{Curves $\Omega_{N}^{(\ell)}$ in the $(g,\Delta)$-plane for $\ell=1$. Note that as $N$ grows larger the individual components
    become thiner with respect to the variable $g$.}
  \label{fig:ellipsesP}
\end{figure}

By the discussion above,  the union of sets $\Omega_{N}^{(\ell)}$ satisfies
\begin{equation}
  \label{eq:inclusion}
  \bigcup_{N \geq 0} \Omega_{N}^{(\ell)} \subsetneq \R^2.
\end{equation}
The fact that the set union  \eqref{eq:inclusion} is not equal to $\R^2$ can also be verified since the union is not an open set.
Thus, a natural question is to determine the image of the inclusion \eqref{eq:inclusion} in the usual $\R^2$ topology. 

\begin{conject} \label{conj:Dense}
The inclusion \eqref{eq:inclusion} is dense for $\ell \in \Z$. 
\end{conject}

If the conjecture is true, even in the case that $\HRabi{\ell}$ is non-degenerate, the parameters $g,\Delta>0$ would be arbitrarily close to parameters $g',\Delta'>0$ such that the spectrum contains degeneracies. In Figure \ref{fig:ellipse} we illustrate the situation for the case $\varepsilon=0$ by showing the curves described by $\Omega_{N}^{(\ell)}$, for $N=1,2,\cdots,12$ in the $(g,\Delta)$-plane. We see that even with a limited number of curves regions of the $\R^2$-plane begin to appear covered, providing evidence for the conjecture.

\begin{figure}[h!]
  \centering
    \includegraphics[height=4cm]{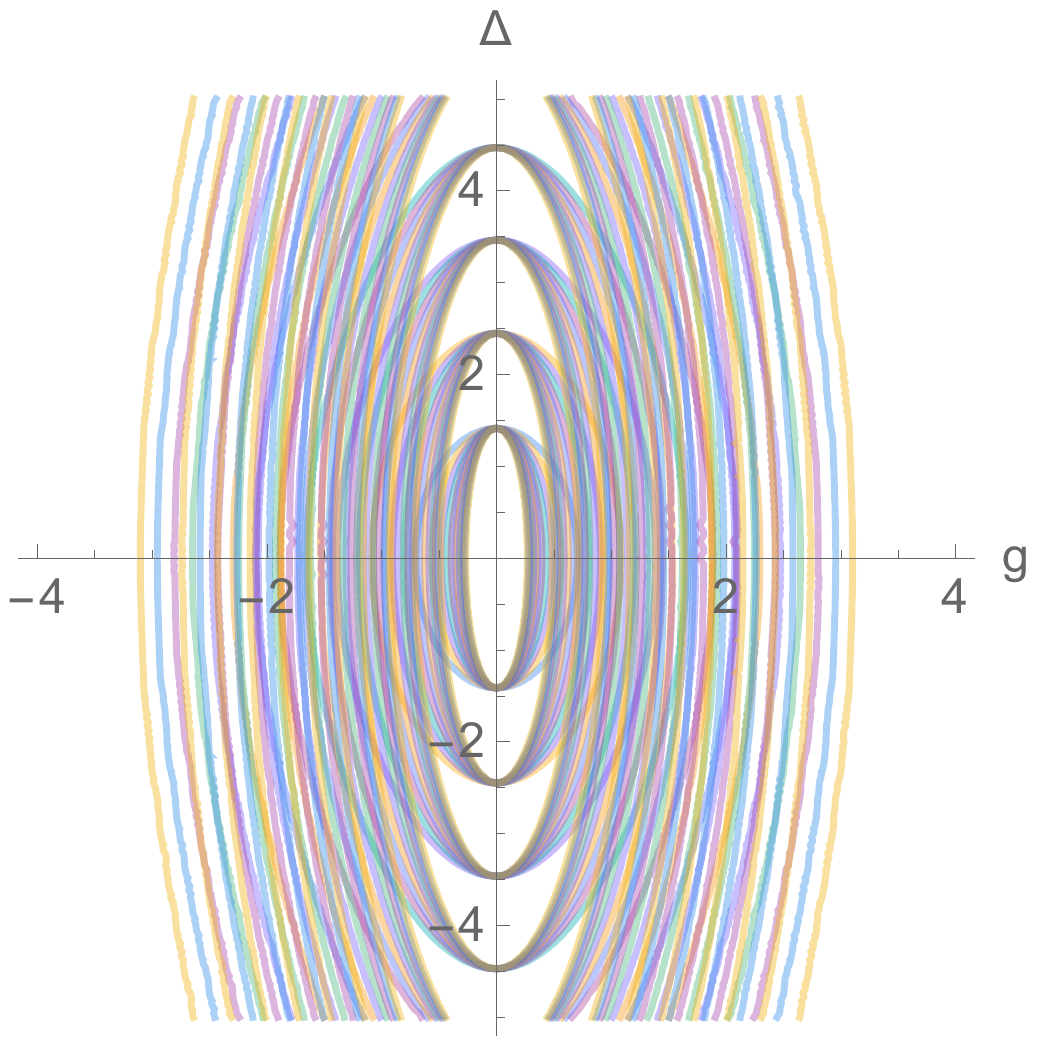}
  \caption{Algebraic curves $\Omega_{N}^{(\ell)}$ for $\ell=1$ and $N=1,2,\cdots,12$. Different colors correspond to different $N$.}
    \label{fig:ellipse}
\end{figure}

\begin{rem} \label{rem:densityCZ}
It is interesting to see the similarity between Conjecture \ref{conj:Dense} and the density discussion about integral points on surfaces studied in \S4 in \cite{CZ2010} in the setting described in \S\ref{sec:addit-remarks-degen}.
\end{rem}

\section*{Acknowledgements}
The authors would like to thank Daniel Braak for comments and suggestions in a preliminary version of this manuscript.
This work was partially supported by Grant-in-Aid for Scientific Research (C) No.20K03560, JSPS,
JST CREST Grant Number JPMJCR14D6, Japan.

\appendix
\section{Explicit expressions for $J_\ell$ for small values of $\ell$}
\label{sec:expl-expr-j_ell}

To complement the discussion of this paper and for reference to the reader, in this appendix we give the explicit expressions
for the operator $J_{\ell}$ for $0 \leq \ell \leq 6$. We note that the cases $0\leq \ell \leq 3$ were given first in \cite{LLMB2021a,RBW2021}.

To simplify the discussion, we consider the Hamiltonian \ibQRM{\ell} to be given in the equivalent form
\[
  \tilde{H}_\ell = a^{\dag}a + \Delta \sigma_x + g (a + a^{\dag}) \sigma_z + \frac{\ell}{2}\sigma_z,
\]
obtained from \eqref{eq:aH} by a unitary Cayley transform. Then, we write
\[
  J_{\ell} = \mathcal{P}
  \begin{bmatrix}
    \alpha^{(\ell)}(a,a^\dag) & \beta^{(\ell)}(a,a^\dag) \\
    \gamma^{(\ell)}(a,a^\dag) & \delta^{(\ell)}(a,a^\dag)
  \end{bmatrix},
\]
so that
\[
  [\tilde{H}_\ell,J_{\ell}]=0.
\]
We note that since $J_{\ell}$ is self-adjoint, we have $\gamma^{(\ell)}(a,a^\dag)=\beta^{(\ell)}(-a,-a^\dag)$ (see Proposition 4.3 of \cite{RBW2021}).

Next, we give the explicit values of $\alpha^{(\ell)}(a,a^\dag)$, $\beta^{(\ell)}(a,a^\dag)$ and $\beta^{(\ell)}(a,a^\dag)$ for $0\leq \ell \leq 6$. 
In addition, we give the expression for the polynomial $p_{\ell}(x; g, \Delta)$ of \eqref{eq:Jquad} and the value of $J_{\ell}$ for $g=0$.

\subsubsection*{Case $\ell=0$}

The coefficients of $J_{0}$ are given by
\begin{align*}
  \alpha^{(0)}(a,a^\dag) &= 0\\
  \beta^{(0)}(a,a^\dag) &= 1\\
  \delta^{(0)}(a,a^\dag) &= 0
\end{align*}
The polynomial $p_{0}(x; g, \Delta)$ is given by
\[
  p_{0}(x;g,\Delta) = 1.
\]
The expression for $g=0$ is given by
\[
  J_{0} = \mathcal{P}
  \begin{bmatrix}
    0 & 1 \\
    1 & 0
  \end{bmatrix}.
\]

\subsubsection*{Case $\ell=1$}

The coefficients of $J_{1}$ are given by
\begin{align*}
  \alpha^{(1)}(a,a^\dag) =&  \Delta, \\
  \beta^{(1)}(a,a^\dag) =&  2 g (g-a),\\
  \delta^{(1)}(a,a^\dag) =&  \Delta.
\end{align*}
The polynomial $p_{1}(x; g, \Delta)$ is given by
\[
  p_{1}(x;g,\Delta) =  4g^2x + (4g^4+2g^2+\Delta^2).
\]
The expression for $g=0$ is given by
\[
  J_{1} = \mathcal{P}
  \begin{bmatrix}
    \Delta & 0 \\
    0 & \Delta
  \end{bmatrix}.
\]

\subsubsection*{Case $\ell=2$}

The coefficients of $J_{2}$ are given by
\begin{align*}
  \alpha^{(2)}(a,a^\dag) =& -2 a g \Delta + 2 a^{\dag} g \Delta + 4 g^{2} \Delta + \Delta, \\
  \beta^{(2)}(a,a^\dag) =& 4 a^{2} g^{2} - 8 a g^{3} + 4 g^{4} + \Delta^{2}, \\
  \delta^{(2)}(a,a^\dag) =& -2 a g \Delta + 2 a^{\dag} g \Delta + 4 g^{2} \Delta - \Delta.
\end{align*}
The polynomial $p_{2}(x; g, \Delta)$ is given by
\[
  p_{2}(x;g,\Delta) = 16 g^4 x^2 + 8 g^2 \left(\Delta ^2+4 g^4+2 g^2\right) x + (\Delta ^4+\Delta ^2+8 \Delta ^2 g^4+4 \Delta ^2 g^2+16 g^8+16 g^6).
\]
The expression for $g=0$ is given by
\[
  J_{2} = \Delta \mathcal{P}
  \begin{bmatrix}
    1 & \Delta \\
    \Delta & -1
  \end{bmatrix}.
\]

\subsubsection*{Case $\ell=3$}

The coefficients of $J_{3}$ are given by
\begin{align*}
  \alpha^{(3)}(a,a^\dag) =& -2\Delta - 8g^2\Delta - 12g^4\Delta - \Delta^3 - 4g\Delta a^{\dag} - 12 g^3\Delta a^\dag - 4g^2\Delta (a^\dag)^2 + 4 g \Delta a+ 12 g^3 \Delta a  \\
                  &+ 4 g^2 \Delta a a^\dag - 4 g^2 \Delta a^2, \\
  \beta^{(3)}(a,a^\dag) =&  -8 g^6 - 6 g^2 \Delta^2 - 2 g \Delta^2 a^\dag + 24 g^5 a + 4 g \Delta^2 a - 24 g^4 a^2 + 8 g^3 a^3,\\
  \delta^{(3)}(a,a^\dag) =& -2 \Delta + 4 g^2 \Delta  - 12 g^4 \Delta - \Delta^3 + 4 g \Delta a^{\dag} - 12 g^3 \Delta a^\dag - 4 g^2 \Delta (a^\dag)^2 - 4 g \Delta a + 12 g^3\Delta a \\
                  &+ 4 g^2 \Delta a a^\dag - 4 g^2\Delta a^2.
\end{align*}
The polynomial $p_{3}(x; g, \Delta)$ is given by
\begin{align*}
  p_{3}(x;g,\Delta) =&  64 g^6 x^3 + 48 g^4 \left(\Delta ^2+4 g^4+2 g^2\right) x^2 + 4 g^2 (3 \Delta ^4+4 \Delta ^2+24 \Delta ^2 g^4+12 \Delta ^2 g^2+48 g^8  \\
               & +48 g^6-4 g^4) x + \Delta ^6+4 \Delta ^4+4 \Delta ^2+48 \Delta ^2 g^8+48 \Delta ^2 g^6+12 \Delta ^4 g^4+12 \Delta ^2 g^4+6 \Delta ^4 g^2 \\
               & +8 \Delta ^2 g^2 +64 g^{12}+96 g^{10}-16 g^8-24 g^6.
\end{align*}
The expression for $g=0$ is given by
\[
  J_{3} = \mathcal{P}
  \begin{bmatrix}
    -\Delta^3 - 2\Delta & 0 \\
    0 & -\Delta^3 - 2\Delta
  \end{bmatrix}.
\]

\subsubsection*{Case $\ell=4$}

The coefficients of $J_{4}$ are given by
\begin{align*}
  \alpha^{(4)}(a,a^\dag) =&  6 \Delta + 2 \Delta^3 + 24 \Delta g^2 + 8 \Delta^3 g^2 + 40 \Delta g^4 + 32 \Delta g^6 + 12 \Delta g (a^\dag) + 4 \Delta^3 g a^\dag  \\
         &+ 40 \Delta g^3 a^\dag + 48 \Delta g^5 a^\dag + 12 \Delta g^2 (a^\dag)^2 + 32 \Delta g^4 (a^\dag)^2 + 8 \Delta g^3 (a^\dag)^3 - 12 \Delta g a  \\
         &- 4 \Delta^3 g a - 40 \Delta g^3 a - 48 \Delta g^5 a - 16 \Delta g^2 a a^\dag - 32 \Delta g^4 a a^\dag  \\
         &- 8 \Delta g^3 a (a^\dag)^2 + 12 \Delta g^2 a^2 + 32 \Delta g^4 a^2 + 8 \Delta g^3 a^2 a^\dag - 8 \Delta g^3 a^3, \\
  \beta^{(4)}(a,a^\dag) =& 3\Delta^2 + \Delta^4 + 4\Delta^2g^2 + 24\Delta^2g^4 + 16g^8 + 16\Delta^2g^3a^\dag + 4\Delta^2g^2(a^\dag)^2 - 32\Delta^2g^3a - 64g^7a \\
                 &- 8\Delta^2g^2a a^\dag+ 12\Delta^2g^2a^2 + 96g^6a^2 - 64g^5a^3  + 16g^4a^4, \\
  \delta^{(4)}(a,a^\dag) =& -6 \Delta - 2 \Delta^3 + 8 \Delta g^2 + 8 \Delta^3 g^2 - 8 \Delta g^4 + 32 \Delta g^6 + 12 \Delta g a^\dag + 4 \Delta^3 g a^\dag  \\
                 &- 24 \Delta g^3 a^\dag + 48 \Delta g^5 a^\dag - 12 \Delta g^2 (a^\dag)^2 + 32 \Delta g^4 (a^\dag)^2 + 8 \Delta g^3 (a^\dag)^3 - 12 \Delta g a \\
                 &- 4 \Delta^3 g a + 24 \Delta g^3 a -  48 \Delta g^5 a + 16 \Delta g^2 a a^\dag - 32 \Delta g^4 a a^\dag - 8 \Delta g^3 a (a^\dag)^2 \\
                 &- 12 \Delta g^2 a^2 + 32 \Delta g^4 a^2 + 8 \Delta g^3 a^2 a^\dag  - 8 \Delta g^3 a^3.
\end{align*}
The polynomial $p_{4}(x; g, \Delta)$ is given by
\begin{align*}
  p_{4}(x;g,\Delta) =&  256 g^8 x^4 + 256 g^6 \left(\Delta ^2+4 g^4+2 g^2\right) x^3 + 32 g^4 (3 \Delta ^4+5 \Delta ^2+24 \Delta ^2 g^4+12 \Delta ^2 g^2+48 g^8\\
                &+48 g^6-8 g^4) x^2 + 16 g^2 (\Delta ^6+5 \Delta ^4+6 \Delta ^2+48 \Delta ^2 g^8+48 \Delta ^2 g^6+12 \Delta ^4 g^4+12 \Delta ^2 g^4\\
                & +6 \Delta ^4 g^2 +10 \Delta ^2 g^2 +64 g^{12}+96 g^{10}-32 g^8-32 g^6 ) x + \Delta ^8+10 \Delta ^6+33 \Delta ^4+36 \Delta ^2 \\
                &+256 \Delta ^2 g^{12} +384 \Delta ^2 g^{10}+96 \Delta ^4 g^8 +32 \Delta ^2 g^8+96 \Delta ^4 g^6+32 \Delta ^2 g^6+16 \Delta ^6 g^4+64 \Delta ^4 g^4 \\
                &+64 \Delta ^2 g^4+8 \Delta ^6 g^2 +40 \Delta ^4 g^2+48 \Delta ^2 g^2 +256 g^{16}+512 g^{14}-256 g^{12}-512 g^{10}.
\end{align*}
The expression for $g=0$ is given by
\[
  J_{4} = (\Delta^3 + 3 \Delta) \mathcal{P}
  \begin{bmatrix}
    2 & \Delta \\
    \Delta & -2
  \end{bmatrix}.
\]

\subsubsection*{Case $\ell=5$}

The coefficients of $J_{5}$ are given by
\begin{align*}
  \alpha^{(5)}(a,a^\dag) =& -24 y \Delta - 10 \Delta^3 - \Delta^5 - 96 \Delta g^2 - 28 \Delta^3 g^2 - 168 \Delta g^4 - 40 \Delta^3 g^4- 160 \Delta g^6 \\
         &- 80 \Delta g^8- 48 \Delta g a^\dag - 12 \Delta^3 g a^\dag - 168 \Delta g^3 a^\dag - 40 \Delta^3 g^3 a^\dag - 240 \Delta g^5 a^\dag - 160 \Delta g^7 a^\dag\\
         &- 48 \Delta g^2 (a^\dag)^2 - 12 \Delta^3 g^2 (a^\dag)^2 - 144 \Delta g^4 (a^\dag)^2 - 160 \Delta g^6 (a^\dag)^2 - 32 \Delta g^3 (a^\dag)^3 \\
        & - 80 \Delta g^5 (a^\dag)^3 - 16 \Delta g^4 (a^\dag)^4 + 48 \Delta g a + 12 \Delta^3 g a + 168 \Delta g^3 a + 40 \Delta^3 g^3 a \\
        &+ 240 \Delta g^5 a + 160 \Delta g^7 a + 72 \Delta g^2 a a^\dag + 16 \Delta^3 g^2 a a^\dag + 192 \Delta g^4 a a^\dag + 160 \Delta g^6 a a^\dag \\
        & + 48 \Delta g^3 a (a^\dag)^2 + 80 \Delta g^5 a (a^\dag)^2 + 16 \Delta g^4 a (a^\dag)^3 - 48 \Delta g^2 a^2 - 12 \Delta^3 g^2 a^2 \\
       &  - 144 \Delta g^4 a^2 - 160 \Delta g^6 a^2 - 48 \Delta g^3 a^2 a^\dag - 80 \Delta g^5 a^2 a^\dag - 16 \Delta g^4 a^2 (a^\dag)^2 \\
                 &  + 32 \Delta g^3 a^3 + 80 \Delta g^5 a^3 + 16 \Delta g^4 a^3 a^\dag - 16 \Delta g^4 a^4, \\
\beta^{(5)}(a,a^\dag) =&  -32 \Delta^2 g^2 - 10 \Delta^4 g^2 - 40 \Delta^2 g^4 - 80 \Delta^2 g^6 - 32 g^10 - 16 \Delta^2 g a^\dag \\
          &- 4 \Delta^4 g a^\dag - 16 \Delta^2 g^3 a^\dag - 80 \Delta^2 g^5 a^\dag - 40 \Delta^2 g^4 (a^\dag)^2 - 8 \Delta^2 g^3 (a^\dag)^3 \\
          &+ 24 \Delta^2 g a + 6 \Delta^4 g a + 24 \Delta^2 g^3 a + 160 \Delta^2 g^5 a + 160 g^9 a + 80 \Delta^2 g^4 a a^\dag\\
          &+ 16 \Delta^2 g^3 a (a^\dag)^2 - 120 \Delta^2 g^4 a^2 - 320 g^8 a^2 - 24 \Delta^2 g^3 a^2 a^\dag + 32 \Delta^2 g^3 a^3 \\
                 &+ 320 g^7 a^3 - 160 g^6 a^4 + 32 g^5 a^5, \\
\delta^{(5)}(a,a^\dag) =& -24 \Delta - 10 \Delta^3 - \Delta^5 + 24 \Delta g^2 + 12 \Delta^3 g^2 - 8 \Delta g^4 - 40 \Delta^3 g^4 - 80 \Delta g^8 + 48 \Delta g a^\dag  \\
         &+ 12 \Delta^3 g a^\dag - 72 \Delta g^3 a^\dag - 40 \Delta^3 g^3 a^\dag + 80 \Delta g^5 a^\dag - 160 \Delta g^7 a^\dag - 48 \Delta g^2 (a^\dag)^2 - 12 \Delta^3 g^2 (a^\dag)^2 \\
         &+ 96 \Delta g^4 (a^\dag)^2 - 160 \Delta g^6 (a^\dag)^2 + 32 \Delta g^3 (a^\dag)^3 - 80 \Delta g^5 (a^\dag)^3 - 16 \Delta g^4 (a^\dag)^4 - 48 \Delta g a - 12 \Delta^3 g a \\
         &+ 72 \Delta g^3 a + 40 \Delta^3 g^3 a - 80 \Delta g^5 a + 160 \Delta g^7 a + 72 \Delta g^2 a a^\dag + 16 \Delta^3 g^2 a a^\dag - 128 \Delta g^4 a a^\dag \\
         &+ 160 \Delta g^6 a a^\dag - 48 \Delta g^3 a (a^\dag)^2 + 80 \Delta g^5 a (a^\dag)^2 + 16 \Delta g^4 a (a^\dag)^3 - 48 \Delta g^2 a^2 - 12 \Delta^3 g^2 a^2 \\
         &+ 96 \Delta g^4 a^2 - 160 \Delta g^6 a^2 + 48 \Delta g^3 a^2 a^\dag - 80 \Delta g^5 a^2 a^\dag - 16 \Delta g^4 a^2 (a^\dag)^2 - 32 \Delta g^3 a^3 \\
         &+ 80 \Delta g^5 a^3 + 16 \Delta g^4 a^3 a^\dag - 16 \Delta g^4 a^4.
\end{align*}
The polynomial $p_{5}(x; g, \Delta)$ is given by
\begin{align*}
  p_{5}(x;g,\Delta) =& 1024 g^{10} x^5 + 1280 g^8 \left(\Delta ^2+4 g^4+2 g^2\right) x^4 + 640 g^6 (\Delta ^4+2 \Delta ^2+8 \Delta ^2 g^4+4 \Delta ^2 g^2+16 g^8 \\
               & +16 g^6-4 g^4) x^3 + 32 g^4 (5 \Delta ^6+30 \Delta ^4+42 \Delta ^2+240 \Delta ^2 g^8+240 \Delta ^2 g^6+60 \Delta ^4 g^4+60 \Delta ^2 g^4\\
               & +30 \Delta ^4 g^2  +60 \Delta ^2 g^2 +320 g^{12}+480 g^{10}-240 g^8-200 g^6) x^2 + 4 g^2 (5 \Delta ^8+60 \Delta ^6+232 \Delta ^4 \\
               & +288 \Delta ^2+1280 \Delta ^2 g^{12} +1920 \Delta ^2 g^{10}+480 \Delta ^4 g^8+480 \Delta ^4 g^6+160 \Delta ^2 g^6+80 \Delta ^6 g^4+360 \Delta ^4 g^4 \\
               &+400 \Delta ^2 g^4+40 \Delta ^6 g^2+240 \Delta ^4 g^2 +336 \Delta ^2 g^2+1280 g^{16}+2560 g^{14}-1920 g^{12}-3200 g^{10}  \\
               &+144 g^8) x + \Delta ^{10}+20 \Delta ^8+148 \Delta ^6+480 \Delta ^4+576 \Delta ^2 +1280 \Delta ^2 g^{16}+2560 \Delta ^2 g^{14}+640 \Delta ^4 g^{12} \\
               & -640 \Delta ^2 g^{12}+960 \Delta ^4 g^{10}-1280 \Delta ^2 g^{10}+160 \Delta ^6 g^8+480 \Delta ^4 g^8 +400 \Delta ^2 g^8+160 \Delta ^6 g^6 \\
               &+560 \Delta ^4 g^6+480 \Delta ^2 g^6+20 \Delta ^8 g^4+200 \Delta ^6 g^4+656 \Delta ^4 g^4+720 \Delta ^2 g^4+10 \Delta ^8 g^2 +120 \Delta ^6 g^2 \\
               &+464 \Delta ^4 g^2+576 \Delta ^2 g^2+1024 g^{20}+2560 g^{18}-2560 g^{16}-6400 g^{14}+576 g^{12}+1440 g^{10}.
\end{align*}
The expression for $g=0$ is given by
\[
  J_{5} = \mathcal{P}
  \begin{bmatrix}
    -\Delta^5 - 10 \Delta^3 - 24 \Delta & 0 \\
    0 & -\Delta^5 - 10 \Delta^3 - 24 \Delta
  \end{bmatrix}.
\]

\subsubsection*{Case $\ell=6$}

The coefficients of $J_{6}$ are given by
\begin{align*}
  \alpha^{(6)}(a,a^\dag) =& 120 \Delta + 39 \Delta^3 + 3 \Delta^5 + 480 \Delta g^2 + 160 \Delta^3 g^2 + 12 \Delta^5 g^2 + 864 \Delta g^4 + 216 \Delta^3 g^4 \\
         &+ 896 \Delta g^6 + 160 \Delta^3 g^6 + 560 \Delta g^8 + 192 \Delta g^10 + 240 \Delta g a^\dag + 78 \Delta^3 g a^\dag + 6 \Delta^5 g a^\dag \\
         &+ 864 \Delta g^3 a^\dag + 192 \Delta^3 g^3 a^\dag + 1344 \Delta g^5 a^\dag + 240 \Delta^3 g^5 a^\dag + 1120 \Delta g^7 a^\dag + 480 \Delta g^9 a^\dag \\
         &+ 240 \Delta g^2 (a^\dag)^2 + 48 \Delta^3 g^2 (a^\dag)^2 + 768 \Delta g^4 (a^\dag)^2 + 144 \Delta^3 g^4 (a^\dag)^2 + 1008 \Delta g^6 (a^\dag)^2 \\
         &+ 640 \Delta g^8 (a^\dag)^2 + 160 \Delta g^3 (a^\dag)^3 + 32 \Delta^3 g^3 (a^\dag)^3 + 448 \Delta g^5 (a^\dag)^3 + 480 \Delta g^7 (a^\dag)^3 + 80 \Delta g^4 (a^\dag)^4 \\
         &+ 192 \Delta g^6 (a^\dag)^4 + 32 \Delta g^5 (a^\dag)^5- 240 \Delta g a- 78 \Delta^3 g a - 6 \Delta^5 g a - 864 \Delta g^3 a - 192 \Delta^3 g^3 a \\
         &- 1344 \Delta g^5 a - 240 \Delta^3 g^5 a - 1120 \Delta g^7 a- 480 \Delta g^9 a - 384 \Delta g^2 a a^\dag - 72 \Delta^3 g^2 a a^\dag \\
         &- 1152 \Delta g^4 a a^\dag - 192 \Delta^3 g^4 a a^\dag - 1344 \Delta g^6 a a^\dag - 640 \Delta g^8 a a^\dag - 288 \Delta g^3 a (a^\dag)^2 - 48 \Delta^3 g^3 a (a^\dag)^2 \\
         &- 672 \Delta g^5 a (a^\dag)^2 - 480 \Delta g^7 a (a^\dag)^2 - 128 \Delta g^4 a (a^\dag)^3 - 192 \Delta g^6 a (a^\dag)^3 - 32 \Delta g^5 a (a^\dag)^4 + 240 \Delta g^2 a^2 \\
         &+ 48 \Delta^3 g^2 a^2 + 768 \Delta g^4 a^2 + 144 \Delta^3 g^4 a^2 + 1008 \Delta g^6 a^2+ 640 \Delta g^8 a^2 + 288 \Delta g^3 a^2 a^\dag \\
  &+ 48 \Delta^3 g^3 a^2 a^\dag + 672 \Delta g^5 a^2 a^\dag + 480 \Delta g^7 a^2 a^\dag + 144 \Delta g^4 a^2 (a^\dag)^2 + 192 \Delta g^6 a^2 (a^\dag)^2 + 32 \Delta g^5 a^2 (a^\dag)^3 \\
         &- 160 \Delta g^3 a^3 - 32 \Delta^3 g^3 a^3 - 448 \Delta g^5 a^3 - 480 \Delta g^7 a^3- 128 \Delta g^4 a^3 a^\dag - 192 \Delta g^6 a^3 a^\dag \\
         &- 32 \Delta g^5 a^3 (a^\dag)^2 + 80 \Delta g^4 a^4 + 192 \Delta g^6 a^4 + 32 \Delta g^5 a^4 a^\dag - 32 \Delta g^5 a^5, \\
\beta^{(6)}(a,a^\dag) =& 40 \Delta^2 + 13 \Delta^4 + \Delta^6 + 64 \Delta^2 g^2+ 12 \Delta^4 g^2 + 228 \Delta^2 g^4 + 60 \Delta^4 g^4 + 240 \Delta^2 g^6 \\
          &+ 240 \Delta^2 g^8 + 64 g^12+ 200 \Delta^2 g^3 a^\dag + 48 \Delta^4 g^3 a^\dag + 192 \Delta^2 g^5 a^\dag + 320 \Delta^2 g^7 a^\dag + 60 \Delta^2 g^2 (a^\dag)^2 \\
          &+ 12 \Delta^4 g^2 (a^\dag)^2 + 48 \Delta^2 g^4 (a^\dag)^2 + 240 \Delta^2 g^6 (a^\dag)^2 + 96 \Delta^2 g^5 (a^\dag)^3 + 16 \Delta^2 g^4 (a^\dag)^4 - 304 \Delta^2 g^3 a \\
          &- 72 \Delta^4 g^3 a - 288 \Delta^2 g^5 a - 640 \Delta^2 g^7 a - 384 g^11 a - 128 \Delta^2 g^2 a a^\dag - 24 \Delta^4 g^2 a a^\dag - 96 \Delta^2 g^4 a a^\dag \\
          &- 480 \Delta^2 g^6 a a^\dag - 192 \Delta^2 g^5 a (a^\dag)^2 - 32 \Delta^2 g^4 a (a^\dag)^3 + 120 \Delta^2 g^2 a^2 + 24 \Delta^4 g^2 a^2+ 96 \Delta^2 g^4 a^2 \\
          &+ 720 \Delta^2 g^6 a^2 + 960 g^10 a^2 + 288 \Delta^2 g^5 a^2 a^\dag + 48 \Delta^2 g^4 a^2 (a^\dag)^2 - 384 \Delta^2 g^5 a^3- 1280 g^9 a^3, \\
\delta^{(6)}(a,a^\dag) =&  -120 \Delta - 39 \Delta^3 - 3 \Delta^5 + 96 \Delta g^2+ 88 \Delta^3 g^2 + 12 \Delta^5 g^2 - 24 \Delta^3 g^4 - 64 \Delta g^6 \\
         &+ 160 \Delta^3 g^6 + 80 \Delta g^8 + 192 \Delta g^10 + 240 \Delta g a^\dag + 78 \Delta^3 g a^\dag + 6 \Delta^5 g a^\dag - 288 \Delta g^3 a^\dag \\
         &- 96 \Delta^3 g^3 a^\dag + 192 \Delta g^5 a^\dag + 240 \Delta^3 g^5 a^\dag - 160 \Delta g^7 a^\dag + 480 \Delta g^9 a^\dag - 240 \Delta g^2 (a^\dag)^2 \\
         &- 48 \Delta^3 g^2 (a^\dag)^2 + 384 \Delta g^4 (a^\dag)^2 + 144 \Delta^3 g^4 (a^\dag)^2 - 432 \Delta g^6 (a^\dag)^2 + 640 \Delta g^8 (a^\dag)^2 \\
         &+ 160 \Delta g^3 (a^\dag)^3 + 32 \Delta^3 g^3 (a^\dag)^3 - 320 \Delta g^5 (a^\dag)^3 + 480 \Delta g^7 (a^\dag)^3 - 80 \Delta g^4 (a^\dag)^4 \\
         &+ 192 \Delta g^6 (a^\dag)^4 + 32 \Delta g^5 (a^\dag)^5 - 240 \Delta g a - 78 \Delta^3 g a - 6 \Delta^5 g a + 288 \Delta g^3 a + 96 \Delta^3 g^3 a\\
         &- 192 \Delta g^5 a - 240 \Delta^3 g^5 a + 160 \Delta g^7 a - 480 \Delta g^9 a + 384 \Delta g^2 a a^\dag + 72 \Delta^3 g^2 a a^\dag \\
         &- 576 \Delta g^4 a a^\dag - 192 \Delta^3 g^4 a a^\dag + 576 \Delta g^6 a a^\dag - 640 \Delta g^8 a a^\dag - 288 \Delta g^3 a (a^\dag)^2 - 48 \Delta^3 g^3 a (a^\dag)^2 \\
         &+ 480 \Delta g^5 a (a^\dag)^2 - 480 \Delta g^7 a (a^\dag)^2 + 128 \Delta g^4 a (a^\dag)^3 - 192 \Delta g^6 a (a^\dag)^3 - 32 \Delta g^5 a (a^\dag)^4 \\
         &- 240 \Delta g^2 a^2 - 48 \Delta^3 g^2 a^2 + 384 \Delta g^4 a^2 + 144 \Delta^3 g^4 a^2 - 432 \Delta g^6 a^2 \\
         &+ 640 \Delta g^8 a^2 + 288 \Delta g^3 a^2 a^\dag + 48 \Delta^3 g^3 a^2 a^\dag - 480 \Delta g^5 a^2 a^\dag + 480 \Delta g^7 a^2 a^\dag \\
         &- 144 \Delta g^4 a^2 (a^\dag)^2 + 192 \Delta g^6 a^2 (a^\dag)^2 + 32 \Delta g^5 a^2 (a^\dag)^3 - 160 \Delta g^3 a^3 - 32 \Delta^3 g^3 a^3 \\
         &+ 320 \Delta g^5 a^3 - 480 \Delta g^7 a^3 + 128 \Delta g^4 a^3 a^\dag - 192 \Delta g^6 a^3 a^\dag - 32 \Delta g^5 a^3 (a^\dag)^2 - 80 \Delta g^4 a^4 \\
         &+ 192 \Delta g^6 a^4 + 32 \Delta g^5 a^4 a^\dag - 32 \Delta g^5 a^5.
\end{align*}
The polynomial $p_{6}(x; g, \Delta)$ is given by
\begin{align*}
  p_{6}(x;g,\Delta) =& 4096 g^{12} x^6 + 6144 g^{10} (\Delta ^2+4 g^4+2 g^2) x^5 + 1280 g^8 (3 \Delta ^4+7 \Delta ^2+24 \Delta ^2 g^4 \\
               &+12 \Delta ^2 g^2+48 g^8+48 g^6-16 g^4 ) x^4 + 256 g^6 (5 \Delta ^6+35 \Delta ^4+56 \Delta ^2+240 \Delta ^2 g^8  \\
               &+240 \Delta ^2 g^6+60 \Delta ^4 g^4+60 \Delta ^2 g^4+30 \Delta ^4 g^2+70 \Delta ^2 g^2+320 g^{12}+480 g^{10}-320 g^8 \\
               &-240 g^6) x^3 + 16 g^4 (15 \Delta ^8+210 \Delta ^6+931 \Delta ^4+1296 \Delta ^2+3840 \Delta ^2 g^{12}+5760 \Delta ^2 g^{10}\\
               &+1440 \Delta ^4 g^8-480 \Delta ^2 g^8+1440 \Delta ^4 g^6+480 \Delta ^2 g^6 +240 \Delta ^6 g^4+1200 \Delta ^4 g^4+1456 \Delta ^2 g^4 \\
               &+120 \Delta ^6 g^2+840 \Delta ^4 g^2+1344 \Delta ^2 g^2+3840 g^{16}+7680 g^{14}-7680 g^{12}-11520 g^{10} \\
               & +1024 g^8) x^2 + 8 g^2 (3 \Delta ^{10}+70 \Delta ^8+595 \Delta ^6+2176 \Delta ^4+2880 \Delta ^2+3840 \Delta ^2 g^{16} \\
               &+7680 \Delta ^2 g^{14}+1920 \Delta ^4 g^{12}-3200 \Delta ^2 g^{12}+2880 \Delta ^4 g^{10}-4800 \Delta ^2 g^{10}+480 \Delta ^6 g^8 \\
               &+1440 \Delta ^4 g^8+1472 \Delta ^2 g^8+480 \Delta ^6 g^6+1920 \Delta ^4 g^6+1792 \Delta ^2 g^6+60 \Delta ^8 g^4+680 \Delta ^6 g^4 \\
               &+2492 \Delta ^4 g^4+3008 \Delta ^2 g^4+30 \Delta ^8 g^2+420 \Delta ^6 g^2+1862 \Delta ^4 g^2+2592 \Delta ^2 g^2+3072 g^{20} \\
               &+7680 g^{18}-10240 g^{16}-23040 g^{14}+4096 g^{12}+6144 g^{10} ) x + \Delta ^{12}+35 \Delta ^{10}+483 \Delta ^8 \\
               &+3281 \Delta ^6+10960 \Delta ^4+14400 \Delta ^2+6144 \Delta ^2 g^{20}+15360 \Delta ^2 g^{18} +3840 \Delta ^4 g^{16}-11520 \Delta ^2 g^{16}\\
               &+7680 \Delta ^4 g^{14}-28160 \Delta ^2 g^{14}+1280 \Delta ^6 g^{12} +1280 \Delta ^4 g^{12} +2816 \Delta ^2 g^{12}+1920 \Delta ^6 g^{10} \\
               &+1920 \Delta ^4 g^{10}+5120 \Delta ^2 g^{10}+240 \Delta ^8 g^8+2080 \Delta ^6 g^8+6064 \Delta ^4 g^8+6400 \Delta ^2 g^8+240 \Delta ^8 g^6 \\
               &+2400 \Delta ^6 g^6+7728 \Delta ^4 g^6+8448 \Delta ^2 g^6+24 \Delta ^{10} g^4+480 \Delta ^8 g^4+3528 \Delta ^6 g^4+11392 \Delta ^4 g^4\\
               &+13824 \Delta ^2 g^4+12 \Delta ^{10} g^2+280 \Delta ^8 g^2+2380 \Delta ^6 g^2+8704 \Delta ^4 g^2 +11520 \Delta ^2 g^2 +4096 g^{24} \\
               &+12288 g^{22}-20480 g^{20}-61440 g^{18}+16384 g^{16}+49152 g^{14}.
\end{align*}
The expression for $g=0$ is given by
\[
  J_{6} = (\Delta^{5} + 13 \Delta^3 + 40 \Delta) \mathcal{P}
  \begin{bmatrix}
    3  & \Delta \\
     \Delta & -3
  \end{bmatrix}.
\]


\begin{flushleft}

\bigskip

 Cid Reyes-Bustos \par
 Department of Mathematical and Computing Science, School of Computing, \par
 Tokyo Institute of Technology \par
 2 Chome-12-1 Ookayama, Meguro, Tokyo 152-8552 JAPAN \par\par
 \texttt{reyes@c.titech.ac.jp}

 \bigskip

 Masato Wakayama \par
 Department of Mathematics, School of Science, \par
 Tokyo University of Science \par
 1-3 Kagurazaka, Shinjyuku-ku, Tokyo 162-8601 JAPAN \par\par
 \texttt{wakayama@rs.tus.ac.jp}

\end{flushleft}


\begin{thebibliography}{99}

\bibitem{A2020}
  S.~Ashhab:
  \textit{Attempt to find the hidden symmetry in the asymmetric quantum Rabi model},
  Phys. Rev. A \textbf{101} (2020), 023808.

\bibitem{B2011}
  D.~Braak:
  \textit{Integrability of the Rabi Model},
  Phys. Rev. Lett. \textbf{107} (2011), 100401.
  
  \bibitem{B2013MfI}
  D.~Braak: 
  \textit{Analytical solutions of basic models in quantum optics},
  in ``Applications + Practical Conceptualization + Mathematics = fruitful Innovation, Proceedings of the Forum of Mathematics for Industry 2014" eds. R.~Anderssen, et al., 75-92, Mathematics for Industry \textbf{11}, Springer, 2016.
  
\bibitem{B2019}
  D.~Braak:
  \textit{Symmetries in the Quantum Rabi Model},
  Symmetry \textbf{11} (2019), 1259.

\bibitem{bcbs2016}
  D.~Braak, Q.H.~Chen, M.T.~Batchelor and E.~Solano:
  \textit{Semi-classical and quantum Rabi models: in celebration of 80 years},
  J. Phys. A: Math. Theor. \textbf{49} (2016), 300301.

\bibitem{BLZ2015}
  M.T.~Batchelor, Z.-M.~Li and H.-Q.~Zhou:
  \textit{Energy landscape and conical intersection points of the driven Rabi model},
  J. Phys. A: Math. Theor. \textbf{49} (2015), 01LT01.

\bibitem{CZ2010}
  P.~Corvaja and U.~Zannier:
  \textit{Integral points, divisibility between values of polynomials and entire curves on surfaces},
  Adv. Math. \textbf{225} (2010), 1095-1118.
  
\bibitem{C2011}
  S.~Caux and J.~Mossel:
  \textit{Remarks on the notion of quantum integrability},
  J. Stat. Mech. \textbf{2011} (2011), P02023.

\bibitem{D1996}
J.~J.~Duistermaat:
The heat kernel Lefschetz fixed point formula for the spin-$c$ Dirac operator, Birkh\"auser, 1996.

\bibitem{GD2013}
  B.~Gardas and J.~Dajka:
  \textit{New symmetry in the Rabi model},
  J. Phys. A: Math. Theor. \textbf{46} (2013), 265302.

\bibitem{H1997}
R.~Hartshorne:
Algebraic Geometry, Graduate Texts in Mathematics \textbf{52}, Springer, 1977.

\bibitem{HH2012}
  M.~Hirokawa and F.~Hiroshima:
  \textit{Absence of energy level crossing for the ground state energy of the Rabi model},
  Comm. Stoch. Anal. \textbf{8} (2014), 551-560.
  
\bibitem{HS2000}
M.~Hindry and J.~H.~Silverman:
Diophantine Geometry - An Introduction, GTM \textbf{201}, Springer, 2000.

\bibitem{HT1992}
R.~Howe and E.~C.~Tan:
Non-Abelian Harmonic Analysis. Applications of $SL(2,\R)$, Springer, 1992.

\bibitem{JC1963}
  E.T.~Jaynes and F.W.~Cummings:
  \textit{Comparison of quantum and semiclassical radiation theories with application to the beam maser},
  Proc. IEEE  \textbf{51} (1963), 89-109.

\bibitem{KRW2017}
  K.~Kimoto, C.~Reyes-Bustos and M.~Wakayama:
  \textit{Determinant expressions of constraint polynomials and degeneracies of the asymmetric quantum Rabi model}.
  Int. Math. Res. Notices Vol. 2021, Issue \textbf{12}, 9458–9544 (2021). Published online April 2020.

\bibitem{K1992}
A.~Knapp:
Elliptic curves, Math. Notes \textbf{40}, Princeton Univ. Press, 1992.

\bibitem{K1985JMP}
  M.~Ku\'s:
  \textit{On the spectrum of a two-level system}, J. Math. Phys., \textbf{26} (1985), 2792-2795.
  
\bibitem{LB2015JPA}
  Z.-M.~Li and M.T.~Batchelor:
  \textit{Algebraic equations for the exceptional eigenspectrum of the generalized Rabi model},
  J. Phys. A: Math. Theor. \textbf{48} (2015), 454005.

\bibitem{LB2020}
  Z.-M.~Li and M.T.~Batchelor:
  \textit{Hidden symmetry and tunneling dynamics in asymmetric quantum Rabi models},  
  Phys. Rev. A \textbf{103} (2021), 023719.

\bibitem{LB2021b}
  Z.-M.~Li, D.~Ferri, D.~Tilbrook and M.T.~Batchelor:
  \textit{Generalized adiabatic approximation to the asymmetric quantum Rabi model: conical intersections and geometric phases},  
  Preprint arXiv:2007.11969 (2021).  
  
\bibitem{LLMB2021}
  X.~Lu, Z.-M.~Li, V.~V.~Mangazeev and M.~T.~Batchelor:
  \textit{Hidden symmetry in the biased Dicke model},
  Preprint 2021. arXiv:2103.13730 [quant-ph].

\bibitem{LLMB2021a}
  X.~Lu, Z.-M.~Li, V.~V.~Mangazeev and M.~T.~Batchelor:
  \textit{Hidden symmetry operators for asymmetric generalized quantum Rabi models},
  Preprint 2021. arXiv:2104.14164 [quant-ph].

\bibitem{MBB2020}
  V.~V.~Mangazeev, M.~T.~Batchelor and V.~V.~Bazhanov:
  \textit{The hidden symmetry of the asymmetric quantum Rabi model},
  J. Phys. A: Math. Theor. \textbf{54} (2021), 12LT01. 


\bibitem{Mu1995}
W.~M\"uller:
\textit{The eta invariant (some recent developments)}, 
Astérisque, \textbf{227} (1995), S\'eminaire Bourbaki, exp. no 787, 335-364.

\bibitem{M1984}
D.~Mumford: Tata Lectures on Theta II, Birkhauser, 1984.

\bibitem{Ni2010}
  T.~Niemczyk {\it et al.}:
  \textit{Beyond the Jaynes-Cummings model: circuit QED in the ultrastrong coupling regime},
  Nature Physics \textbf{6} (2010), 772-776.

\bibitem{R1936}
I.~I.~Rabi:
\textit{On the process of space quantization},
Phys. Rev. \textbf{49} (1936), 324. 

\bibitem{R1972}
  M.~Reed and B.~Simon:
  \textit{Methods of Modern Mathematical Physics I: Functional Analysis},
  Academic Press, New York (1972). 

\bibitem{R2020}
C.~Reyes-Bustos:
\textit{The heat kernel of the asymmetric quantum Rabi model}, 
Preprint 2020, arXiv:2012.13595 [math-ph]. 

\bibitem{CRB2020}
  C.~Reyes-Bustos,
  \textit{Extended divisibility relations for constraint polynomials of the asymmetric quantum Rabi model},
  in ``International Symposium on Mathematics, Quantum Theory, and Cryptography (MQC 2019)'', eds. T. Takagi et al.
  Mathematics for Industry \textbf{33}, 149-168, Springer Singapore, 2020.

\bibitem{RBW2021}
  C.~Reyes-Bustos, D.~Braak and M.~Wakayama:
  \textit{Remarks on the hidden symmetry of the asymmetric quantum Rabi model},
  J. Phys. A: Math. Theor. \textbf{54} (2021), 285202. 

  
\bibitem{RW2019}
C.~Reyes-Bustos and M.~Wakayama:
\textit{The heat kernel for the quantum Rabi model}, 
Preprint 2020, arXiv:1906.09597 [math-ph] [math.NT] [quant-ph]. 
 
\bibitem{RW2021}
C.~Reyes-Bustos and M.~Wakayama:
\textit{Heat kernel for the quantum Rabi model: II. Propagators and spectral determinants}, J. Phys. A: Math. Theor. \textbf{54}
(2021), 115202.
  
  
\bibitem{EVBSS2017}
D.~Rossatto, C-J.~Villas-B\^oas, M.~Sanz and E.~Solano:
\textit{Spectral classification of coupling regimes in the quantum Rabi model}, 
Phys. Rev. A \textbf{96} (2017), 013849.

\bibitem{RTW2021}
  E.~Rousseau, A.~Turchet and J.~T.-Y.~Wang:
  \textit{Divisibility of polynomials and degeneracy of integral points},
  Preprint arXiv:2106.11337 (2021).
  
\bibitem{SS2010}
M.~Sch\"utt and T.~Shioda:
\textit{Elliptic surfaces}, Adv. Stud. Pure Math. ``Algebraic Geometry in East Asia - Seoul 2008'' \textbf{60} (2010),
51-160.

\bibitem{Sc1967AP}
   S.~Schweber:
  \text{On the application of Bargmann Hilbert spaces to dynamical problems}, Ann. Phys. \textbf{41} (1967), 205-229.

\bibitem{SK2017}
J.~Semple and M.~Kollar:  
  \textit{Asymptotic behavior of observables in the asymmetric quantum Rabi model},
  J. Phys. A: Math. Theor. \textbf{51} (2017), 044002.

\bibitem{S1994}
  J.~H.~Silverman:
  Advanced Topics in the Arithmetic of Elliptic Curves,
 GTM \textbf{151}, Springer, 1994.
  
\bibitem{SL2000}
  S.~Y.~Slavyanov and W.~Lay:
  A Unified Theory Based on Singularities, Oxford Mathematical Monographs, 2000.
  
\bibitem{VR1976}
M.~Vergne and H.~Rossi:
\textit{Analytic continuation of the holomorphic discrete series of a semi-simple Lie group}, 
Acta Math. \textbf{136} (1976), 1-59. 

\bibitem{W2016JPA}
  M.~Wakayama:  
  \textit{Symmetry of asymmetric quantum Rabi models}. J. Phys. A: Math. Theor.
  \textbf{50} (2017), 174001.

\bibitem{XC2021}
  Y.-F.~Xie and Q.~H.~Chen:
  \textit{Double degeneracy associated with hidden symmetries in the asymmetric two-photon Rabi
    model},
  Preprint 2021.  arXiv:2102.03944v2 [quant-ph].

\bibitem{Y2017}
  F.~Yoshihara \textit{et al.}:
  \textit{Superconducting qubit-oscillator circuit beyond the ultrastrong-coupling regime},
  Nature Physics \textbf{13} (2017), 44.

\bibitem{YS2018}
F.~Yoshihara, T.~Fuse, Z.~Ao, S.~Ashhab, K.~Kakuyanagi, S.~Saito, T.~Aoki, K.~Koshino, and K.~Semba:  
\textit{Inversion of Qubit Energy Levels in Qubit-Oscillator Circuits in the Deep-Strong-Coupling Regime}, 
Phys. Rev. Lett. \textbf{120} (2018), 183601. 

\end{thebibliography}
\end{document}